\documentclass[%
 reprint, superscriptaddress,
 amsmath,amssymb,
 aps,
]{revtex4-2}
\pdfoutput=1 
\usepackage{booktabs}
\usepackage{multirow}
\usepackage{graphicx}
\usepackage[table,xcdraw,dvipsnames]{xcolor}
\usepackage{subcaption}
\usepackage{hhline}
\usepackage{dcolumn}
\usepackage{bm}
\usepackage{bbold}
\usepackage{scalerel}
\usepackage{multirow}
\usepackage{makecell}
\usepackage{tikz-cd}
\usepackage{array}
\usepackage{booktabs}
\usepackage{physics}
\usepackage{amsmath,amsthm,amssymb,amsfonts,physics,siunitx,bbm,graphicx}
\usepackage{xfrac, hyperref}
\hypersetup{
    colorlinks,
    citecolor=blue,
    filecolor=blue,
    linkcolor=blue,
    urlcolor=blue
}
\usepackage{natbib}
\usepackage{mathtools}
\usepackage{comment}
\usepackage{stmaryrd}
\usepackage{braket}

\usepackage{svg}
\usepackage{subcaption}

\newtheorem{theorem}{Theorem}
\newtheorem{defn}{Definition}
\newtheorem{result}{Result}
\newtheorem{corollary}{Corollary}
\newtheorem{lemma}{Lemma}
\newenvironment{lemmap}[1]
  {\renewcommand{\thelemma}{\ref*{#1}$'$}%
   \addtocounter{lemma}{-1}%
   \begin{lemma}}
  {\end{lemma}}
\newenvironment{lemmapp}[1]
  {\renewcommand{\thelemma}{\ref*{#1}$''$}%
   \addtocounter{lemma}{-1}%
   \begin{lemma}}
  {\end{lemma}}

\newcommand{\ba}{\begin{eqnarray}}
\newcommand{\ea}{\end{eqnarray}}
\newcommand{\ban}{\begin{eqnarray*}}
\newcommand{\ean}{\end{eqnarray*}}
\newcommand{\one}{\mathbb{1}}
\DeclareMathOperator{\pos}{pos}

\newcommand{\chn}{\mathcal{E}}
\newcommand{\tchn}{\mathcal{T}}
\newcommand{\tpose}{\text{T}}
\newcommand{\cchn}{\varphi}
\newcommand{\rf}{\alpha}
\newcommand{\crf}{\gamma}
\newcommand{\cbath}{\beta}

\newcommand{\TR}{\textrm{TR}}

\newcommand{\rvp}{{\hat{\cchn}_{\crf}}}

\newcommand{\petz}{{\hat{\chn}_{\rf}}}
\newcommand{\TrB}{\Tr_{B}}
\newcommand{\swap}{U_{\leftrightarrowtriangle}}

\newcommand{\therm}[1]{\tau_{\kappa}({#1})}
\newcommand{\ret}[1]{\mathcal{R}_{q}\left[#1\right]}
\newcommand{\cret}[1]{\mathcal{R}_{c}\left[#1\right]}

\newcommand*{\sdots}{\ifmmode\mathinner{\ldotp\kern-0.2em\ldotp\kern-0.2em\ldotp}\else.\kern-0.13em.\kern-0.13em.\fi}

\newcommand{\ctrace}{\Sigma}
\newcommand{\bctrace}{\mathbf{\Sigma}}

\newcommand{\qam}[1]{\hat{\text{T}}\text{r}_{B,#1}}
\newcommand{\vone}{v_\textbf{1}}
\newcommand{\vzero}{v_\textbf{0}}
\newcommand{\vass}[1]{v({\Lambda_{#1 B}})}

\newcommand{\changessecond}[1]{{ #1}}

\allowdisplaybreaks
\begin{document}

\title{Role of Dilations in Reversing Physical Processes: \\ Tabletop Reversibility and Generalized Thermal Operations}%
\author{Clive Cenxin Aw}
\affiliation{Centre for Quantum Technologies, National University of Singapore, 3 Science Drive 2, Singapore 117543}

\author{Lin Htoo Zaw}
\affiliation{Centre for Quantum Technologies, National University of Singapore, 3 Science Drive 2, Singapore 117543}

\author{Maria Balanzó-Juandó}
\affiliation{ICFO-Institut de Ci\`encies Fot\`oniques, The Barcelona Institute of Science and Technology, Av.~Carl Friedrich Gauss 3, 08860 Castelldefels (Barcelona), Spain}

\author{Valerio Scarani}
\affiliation{Centre for Quantum Technologies, National University of Singapore, 3 Science Drive 2, Singapore 117543}
\affiliation{Department of Physics, National University of Singapore, 2 Science Drive 3, Singapore 117542}

\date{\today}

\begin{abstract}
Irreversibility, crucial in both thermodynamics and information theory, is naturally studied by comparing the evolution---the (forward) channel---with an associated reverse---the reverse channel. There are two natural ways to define this reverse channel. Using logical inference, the reverse channel is the Bayesian retrodiction (the Petz recovery map in the quantum formalism) of the original one. Alternatively, we know from physics that every irreversible process can be modeled as an open system: one can then define the corresponding closed system by adding a bath (``dilation''), trivially reverse the global reversible process, and finally remove the bath again. We prove that the two recipes are strictly identical, both in the classical and in the quantum formalism, once one accounts for correlations formed between system and the bath. Having established this, we define and study special classes of maps: \textit{product-preserving maps} (including \textit{generalized thermal maps}), for which no such system-bath correlations are formed for some states; and \textit{tabletop time-reversible maps}, when the reverse channel can be implemented with the same devices as the original one. We establish several general results connecting these classes, and a very detailed characterisation when both the system and the bath are one qubit. In particular, we show that, when reverse channels are well-defined, product-preservation is a sufficient but not necessary condition for tabletop reversibility; and that the preservation of local energy spectra is a necessary and sufficient condition to generalized thermal operations.
\end{abstract}

\maketitle
\section{Introduction}\label{intro}

Irreversibility is ubiquitous in real life. In science, it was first studied systematically in the context of thermodynamics: this is captured by the Second Law, which stipulates the impossibility of putting all the energy to fruition, leading to the necessary generation of heat---or more generally, entropy \cite{gibbs1879equilibrium, onsager1931reciprocal, Seifert_2012,nelly-brandao2015second,evans2002fluctuation, crooks-reversal, campisi-haenggi-review-2011}. Eventually, information theory became the setting in which to study irreversibility: a process is irreversible for an agent when the agent is unable to retrieve from the output all the information about the input. In turn, a theory of optimal retrieval of information was developed, both in classical and in quantum theories \cite{ jaynes1957information,shannon1948mathematical, bennett1982thermodynamics, barnum-knill, anders-hilt2011landauer}.

Meanwhile, the field of stochastic thermodynamics developed quantitative approaches to irreversibility, based on \textit{statistical comparisons between the process under study and its associated reverse process}. But how to define the latter? In the case of fully reversible, deterministic processes, the reverse process is obviously the dynamics played backwards. For isothermal evolutions (driven Hamiltonian evolution while the system is in contact with a thermal bath), a possible and very natural reverse process consists in driving the evolution backwards in the presence of the same bath \cite{jarz2000,AWWW18}. Reverse process have also been found for more complex processes, through expert control of the model and its assumptions (see e.g.~\cite{manzano-PRX}). \changessecond{A general recipe may be built on the observation that any irreversible process can be seen as a marginal of a global, reversible process involving the system and some environment. The \textit{recipe through dilation} is then: add a suitable environment (\textit{dilation}), trivially reverse the reversible global process, and finally remove the environment.

Recently, it was proposed to define the reverse process using the Bayesian recipe for information retrieval, a.k.a.~Bayesian inversion or \textit{retrodiction} \cite{BS21,AwBS}. This \textit{recipe through retrodiction} requires only choosing a reference state, which plays the role of a Bayesian prior. The connection between reverse processes and Bayesian logic had not been noticed in the context of classical stochastic thermodynamics. In quantum thermodynamics, one the main tools for information recovery had been used, first occasionally \cite{Alhambra16,AWWW18}, then systematically \cite{kwon-kim,kwon2022}: the \textit{Petz recovery map} \cite{petz,barnum-knill, wilde_2013}, which was also proposed early on as a quantum analog of Bayesian inversion \cite{Leifer-Spekkens, PB22, PF22,  QPRPetzPaper}. 

In this work, after a review of known material on reverse processes (Section \ref{ConSec}), we start by proving that the two recipes by dilation and by retrodiction are identical, both in the classical and the quantum case (Section \ref{DilSec}). The fact that the two proposed general recipes coincide, combined with the knowledge that all the previously known special cases can be recovered with these recipes, shows that \textit{we have the definition of the reverse process under control}. Next we bring up the observation that \textit{a process and its associated reverse process may be very different}. It is indeed well known in the quantum case that implementing the reverse (Petz) of a channel may require very different resources than those needed to implement the channel itself (see e.g.~\cite{quek2020quantum}). The cases mentioned above of the reversible and the isothermal processes, whose reverses are ``what one would expect'' and can be implemented with the same control and the same environment, seem to be the exception. Based on this observation, we set to study which processes have a reverse that can be implemented with the same, or similar, resources (Fig.~\ref{fig:TTRDiag}). We shall say that the latter processes possess \textit{tabletop reversibility}. This is of interest for the structure of the theory of the reverse processes, as well as for possible experimental tests of fluctuation theorems in situations that are not unitary or isothermal. The mathematical definition of tabletop reversibility is defined in Section \ref{sec:ttr}, together with the auxiliary notion of \textit{product-preserving channels}, which may be of interest in its own right.} In Section \ref{ResSec}, we present both general results valid and a thorough characterization of two-qubit channels. In Section \ref{sec:further}, we highlight the implications of these results and scenarios for energetics and reversibility in the quantum regime. Section \ref{Concl} is a conclusion.

\begin{figure}[ht!]
    \includegraphics[width=\linewidth]{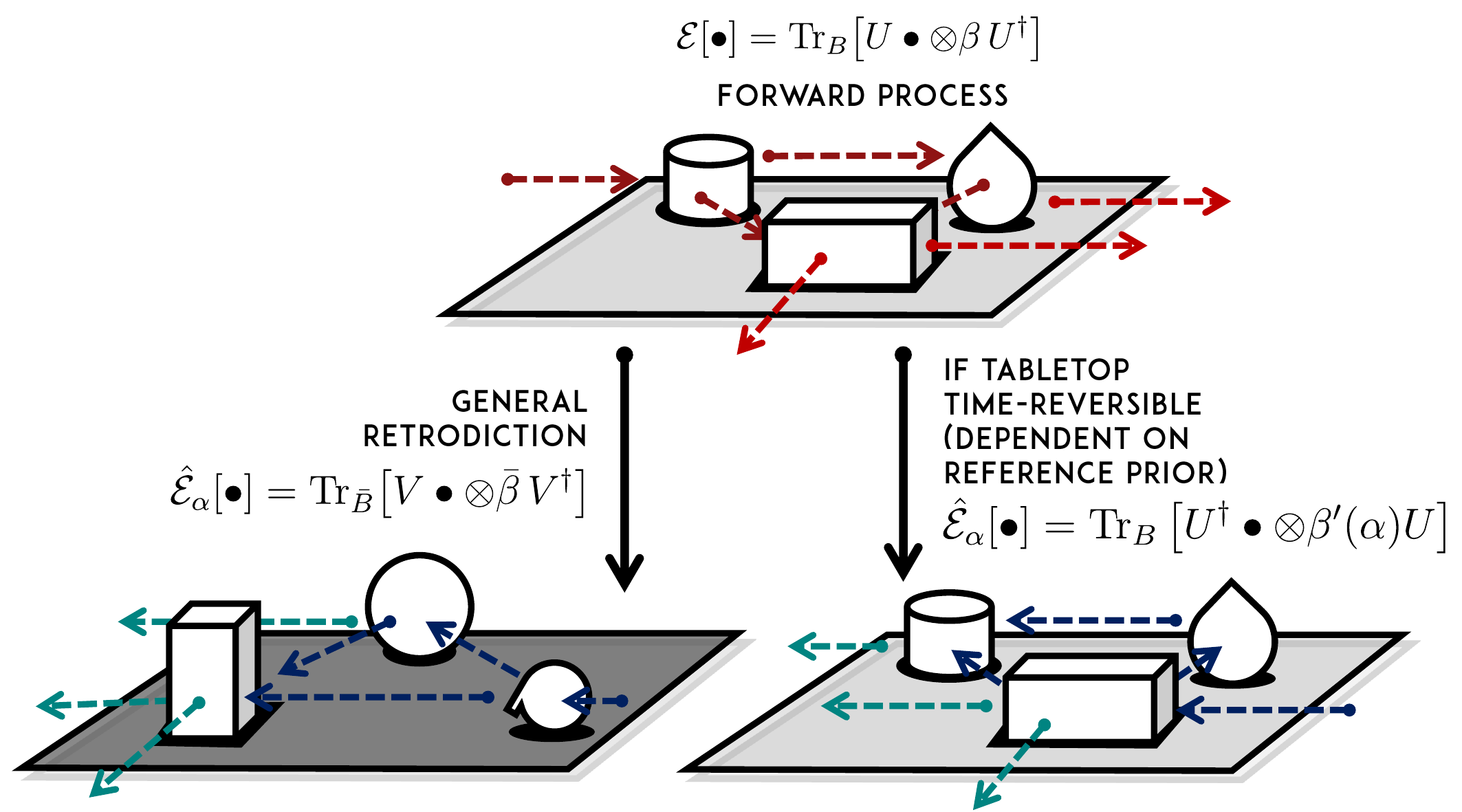}
    \caption{\changessecond{An illustration of the main goal of this paper. Any channel $\chn$ on a system can be viewed as a larger reversible unitary process $\mathcal{U}$ involving the system and a bath (top). Applying the recipes by dilation or by retrodiction (proved identical in Section \ref{DilSec}), one finds that the reverse process must in general be implemented with completely different tools than the original process (bottom left). We set out to characterize the \textit{tabletop reversible} situations, in which the reverse process can be implemented with the same, or similar, tools as the forward process: namely, by appending an ancilla and inverting the original unitary $\mathcal{U}$. The exact definitions will be given in Section \ref{sec:ttr}, and the results in the following Sections.} }
    \label{fig:TTRDiag}
\end{figure} 

\section{Reverse Processes} \label{ConSec}

\subsection{Notations}

In classical theory, we consider a discrete state space with $d$ distinct states. A generic state is represented by a probability distribution $p$: $p(j)\geq 0$ for $j=1,...,d$; and $\sum_{j=1}^d p(j)=1$. It can be represented by a $d\times 1$ probability vector $\mathbf{p}$, whose entries are $\mathbf{p}_j:= p(j)$. A generic channel is a stochastic map $\cchn$, defined by $d^2$ probabilities $\varphi(j'|j)$ of transiting from the input state $j$ to the output state $j'$. These probabilities must satisfy $\varphi(j'|j)\geq 0$ for all $j,j'$ and $\sum_{j'=1}^d\varphi(j'|j)=1$ for all $j$. The channel can be represented by the $d\times d$ stochastic matrix $\boldsymbol{\cchn}$, whose entries are $\boldsymbol{\cchn}_{j'j}:=\varphi(j'|j)$. In this representation, the composition of channels is represented by the standard matrix multiplication: if $\varphi=\varphi_2\circ \varphi_1$, then $\boldsymbol{\varphi}=\boldsymbol{\varphi}_2\, \boldsymbol{\varphi}_1$. 

An important remark for what follows: even if the matrix $\boldsymbol{\cchn}$ has an inverse, in general the entries of the matrix $\boldsymbol{\cchn}^{-1}$ do not define a valid stochastic map. Analogously, while every matrix can be transposed, the map corresponding to $\boldsymbol{\cchn}^{\tpose}$ is a valid map if and only if the channel is bistochastic, i.e.~satisfies the additional property $\sum_{j=1}^d\varphi(j'|j)=1$ for all $j'$. When the inverse or the transpose of the matrix do define valid channels, we shall denote those channels as $\cchn^{-1}$ and $\cchn^{\tpose}$ respectively.

In quantum theory, we consider a finite-dimensional complex vector space of dimension $d$. A generic state is described by a semidefinite operator $\rho \succeq 0$ in this space with $\Tr(\rho)=1$. Channels are represented by completely positive, trace preserving (CPTP) maps. Given a CPTP map $\chn$, the adjoint $\chn^\dagger$ is the unique map for which 
\begin{equation}\label{eq:adj}
    \Tr(\chn[X] Y) = \Tr(X\chn^\dagger[Y])
\end{equation}
for operators $X$ and $Y$. Just as $\varphi^\tpose$ is not a valid stochastic map in general, $\chn^\dagger$ is in general not a valid quantum channel.

\begin{figure*}
    \centering
    \includegraphics[width=0.71\textwidth]{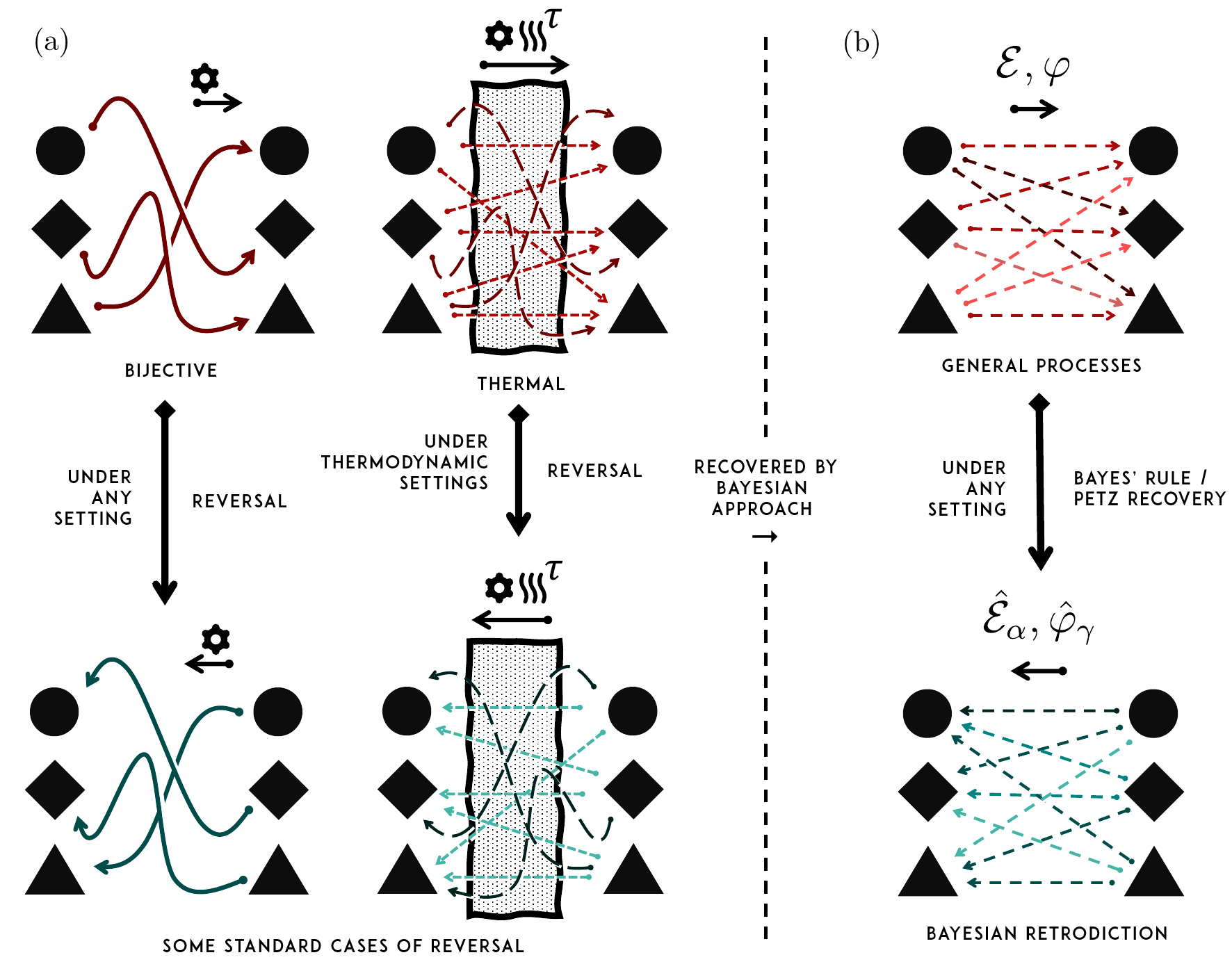}
    \caption{(a) Illustrations for standard examples of reversal. (b) Bayesian inversion or ``retrodiction'' is a formal recipe that reproduce results of the standard approach while generalizing the definition of reverse processes for any characterized process, and any under setting as captured by a reference state.}
    \label{fig:RevChan}
\end{figure*}

\subsection{Reverse Process: Standard Examples}

We first review some classes of processes where the reverse---or at least, a candidate for it---is known (see Fig.~\ref{fig:RevChan}(a) for an illustration).

The most obvious class is that of \textit{reversible processes}, where the map is a \textit{bijection} between the space of states. Their reverse processes are naturally defined as the evolution played backwards, i.e.~the inverse map. In classical theory, such processes are Hamiltonian evolutions $\Phi$ obeying Liouville's theorem: reversal is given by inverting the trajectories in the configuration space. In a discrete state space, the matrix $\boldsymbol{\Phi}$ is a permutation matrix, whose inverse and transpose coincide ($\boldsymbol{\Phi}^{-1}=\boldsymbol{\Phi}^{\tpose}$) and define a valid map. Analogously, in quantum theory, bijective transformations are unitary channels $\mathcal{U}[\bullet] = U \bullet U^\dag$ where $U U^\dag = \one$: their inverse exists, coincides with the adjoint, and defines a valid map. For a reversible process, the inverse map is the only reasonable candidate of the corresponding reverse process.

Moving to \textit{bistochastic/unital processes}, their inverse is in general not defined, but their transpose/adjoint is a valid channel and thus provides an immediate, natural candidate for the reverse map. Beyond this class, the transpose/adjoint ceases to define a valid channel. On this basis, it has been argued that only bistochastic/unital processes are fundamental if one wants the theory to be fundamentally reversible \cite{DiBiagio2021arrowoftimein, CL22, chiribella2021symmetries}. We do not take sides in that discussion: we are not concerned with ultimate constraints on fundamental theories, but with the description of practical irreversibility. 

The most obvious irreversible processes describe the dissipation of information in an unmonitored environment, or arise from a coarse-graining over a chaotic dynamics. In both cases, the dynamics is generally \textit{not bistochastic/unital}. The reverse of some such processes has been constructed on a case-by-case basis by invoking physical arguments. For \textit{a system undergoing Hamiltonian evolution while in contact with thermal baths}, under the assumption of detailed balance, the recipe is very sinmple: the reverse process consists in playing the Hamiltonian evolution backward, while staying in contact with the same thermal baths \cite{crooks-reversal,jarz2000}. This recipe was used in the experimental demonstrations \cite{exp-ciliberto2010fluctuations} of Crooks' fluctuation theorem with biological systems \cite{exp-Collin05,exp-hayashi2018application} and levitating nanospheres \cite{exp-hoang2018experimental, exp-carberry2004fluctuations}. For quantum channels, the same result holds for the so-called thermal operations $\tchn$ \cite{horo-oppen-thermal,AWWW18,santos_landi,landi-pater-21-review,horo2003-thermal-first,brandao2013-thermal,lostaglio2018-ETOs, faist-oppo-renner-2015-thermal, hu2019thermal}. 



\subsection{Reverse Processes: General Recipe through Bayesian Retrodiction}

\subsubsection{Generalities}

Having reviewed the prime examples of constructed reverse processes, we now describe a general recipe applicable to every process: \textit{Bayesian retrodiction} (sometimes called ``Bayesian inversion") \cite{BS21,AwBS}. \changessecond{It can be traced back to the works of Watanabe \cite{watanabe55,watanabe65}, ultimately building on the observation that the laws of physics give us the knowledge of the evolution (the channel, in our language) and not the initial state.}

In classical theory, given the channel $\cchn$, the recipe for Bayesian inversion is standard: one postulates a \textit{reference prior} $\crf$, then applies Bayes' rule to the joint probability distribution $P(a,a'):=\varphi(a'|a)\gamma(a)$. The resulting reverse map is
\begin{equation} \label{bayes}
    \rvp(a|a')=\cchn(a'|a)\frac{\crf(a)}{\varphi[\crf](a')}
\end{equation}
where the distribution $\varphi[\crf]$ is given by $\cchn[\crf](a')=  \sum_a \cchn(a'|a)\crf(a)$, the output obtained by propagating the reference prior $\crf$ through the channel $\cchn$. In matrix notation, Eq.~\eqref{bayes} reads
\begin{eqnarray}\label{bayesmat}
    \boldsymbol{\hat{\varphi}}_\gamma &=& \mathbf{D}_\crf \boldsymbol{{\varphi}}^\tpose \mathbf{D}^{-1}_{\varphi[\crf]}.
\end{eqnarray} where $\mathbf{D}_p$ is the diagonal matrix with entries corresponding to the distribution $p$. 

The extension of Bayesian formalism to the quantum formalism has been the object of many studies \cite{Leifer-Spekkens, parzygnat2022noncommutingbayes, PF22, PB22,surace2023state-retrieval-beyond-bayes}. Here, we don't need an exhaustive Bayesian toolbox: only a candidate for the reverse map. Our choice is the \textit{Petz Recovery map} $\petz$ \cite{petz1,petz, wilde-recov} 
\begin{equation}\label{petz}
    \petz[\bullet] = \sqrt{\rf} \; \chn^\dagger\!\left[\frac{1}{\sqrt{\chn[\rf]}} \bullet \frac{1}{\sqrt{\chn[\rf]}}\right] \sqrt{\rf},
\end{equation}
where $\alpha$ is a state that plays the role of reference prior. The choice of the Petz map as the quantum analog of Bayes' rule is standard \cite{Leifer-Spekkens, PF22, QPRPetzPaper}. It has recently been shown to fulfill the most crucial intuitions about reversal \cite{wilde-recov, PB22, PF22,kwon2022} and to be suited to recover results in stochastic quantum thermodynamics \cite{kwon-kim,BS21}. 

From now onwards, we identify 
\begin{align}\label{eq:notationchoice}
    \cret{\cchn,\crf} &= \rvp \\ 
    \ret{\chn,\rf} &= \petz 
\end{align} and will use these notations interchangeably when convenient. A property of the Bayes/Petz maps that we shall use later is \textit{composability} \cite{PB22}:
\begin{align}
    \mathcal{R}_{c}[\cchn_2 \circ \cchn_1,\crf]=& \cret{\varphi_1, \crf} \circ\cret{\varphi_2, \cchn_1[\crf]},\label{composability_c}\\
    \ret{\chn_2 \circ \chn_1, \rf} = &\ret{\chn_1,\rf} \circ \ret{{\chn}_{2}, \chn_1[\rf]}\,.\label{composability_q}
\end{align}
One may see Appendix \ref{app:proofcomposability} for proofs. This property holds even when the maps are not stochastic.

We note in passing that, due to presence of the term $\chn[\rf]^{-1/2}$ ($\mathbf{D}^{-1}_\cchn[\crf]$), $\petz$ ($\rvp$) is ill-defined when the propagated reference $\chn[\rf]$ ($\cchn[\crf]$) is rank-deficient. Of course, one could attempt to side-step this problem several ways: for example, by defining $\chn[\rf]^{-1/2}$ ($\mathbf{D}^{-1}_\cchn[\crf]$) only on its support, or by considering a neighbourhood of full-rank states around the rank-deficient output and taking some limit. However, this boils down to a matter of convention, where each approach gives a different retrodiction channel, as we discuss in Appendix~\ref{apd:ill-defined-how} in some detail. We do not make a particular choice, and instead leave the retrodiction channel undefined in this case.

\subsubsection{Examples revisited}

The examples of reverse maps of the previous subsection are recovered, and possibly clarified, in the retrodictive approach. \textit{Reversible maps} are the only ones, for which the Bayes/Petz map does not depend on the reference prior \cite{watanabe65,AwBS,PF22}; and, unsurprisingly, coincides with the inverse:
\begin{equation}
    \begin{aligned}\label{eq:deterministic}
    \cret{\Phi}= \hat{\Phi}_\crf = \Phi^\tpose = {\Phi}^{-1}  &\quad \forall \crf\,,\\
    \quad \ret{\mathcal{U}}= \hat{\mathcal{U}}_\rf = \mathcal{U}^\dagger = \mathcal{U}^{-1}  &\quad \forall \rf\,.
\end{aligned}
\end{equation}

For the case of \textit{bistochastic/unital maps}, the Bayes/Petz does depend on the reference prior \cite{PF22, AwBS}. The reverse described above is obtained for a very natural choice of reference prior, namely the uniform: $\gamma:= u$ with $u(j)=1/d$ in the classical case, $\alpha:=\one/d$ in the quantum case. Indeed, this prior is preserved by these maps ($\varphi[u]=u$, $\chn[\one/d]=\one/d$), and one can immediately see that
\begin{equation}
    \begin{aligned}\label{unitalchn}
    & \cret{\cchn,u} = \cchn^{\tpose}, \quad  \ret{\chn,\one/d} = \chn^\dagger \\ & \textrm{[bistochastic/unital]}. 
\end{aligned}
\end{equation}

Lastly, let us look at \textit{thermal operations} in the quantum language. Given a non-interacting system-bath Hamiltonian $H_A\otimes\one + \one \otimes H_B$, one defines 
\begin{equation}\label{eq:thermchn}
    \tchn[\bullet] = \TrB\left\{U \left[\bullet \otimes \, \therm{H_B}\right] U^\dag\right\}
\end{equation}
where $\therm{H} = e^{\kappa H}/\tr(e^{\kappa H})$ is the thermal (or Gibbs) state with $\kappa= -1/k_B T$ (usually denoted $-\beta$, but later in the paper we use $\beta$ will denote a state of the bath) and where $U$ is constrained to satisfy
\begin{equation}\label{eq:thermal-hamiltonian-definition}
    [U,H_A\otimes\one + \one \otimes H_B] = 0.
\end{equation} A channel thus constructed preserves the thermal, or Gibbs, state of the system for the same $\kappa$: $\tchn[\therm{H_A}]=\therm{H_A}$. The Petz map with the Gibbs state as reference prior is found to be (\cite{AWWW18}, see also Appendix \ref{GTherm} for the derivation)
\begin{equation}\label{eq:thermrev2}
    \hat{\tchn}_{\therm{H_A}}[\bullet] = \TrB\left\{U^\dag \left[\bullet \otimes \therm{H_B}\right] U\right\}\,,
\end{equation} which describes indeed a reversal of the unitary dynamics while in contact with the same thermal bath. Notice how, having adopted the retrodictive approach, the usual thermodynamical assumptions called ``microreversibility'' and ``detailed balance'' are replaced by the single assumption on the choice of the reference prior.

\section{Reverse processes and dilations} \label{DilSec} 

Before studying tabletop reversibility, we need to introduce the notion of \textit{dilation of a process}. The word, common in the language of quantum channels, describes \textit{the extension of a process on system A to include an environment (or ``bath'', or ``ancilla'') B}. 

Typically, a dilation is performed with the goal of making the extended system AB a closed one, whose dynamics is therefore reversible. Hence, it is natural to look at defining reverse processes by the following \textit{recipe}: dilate by adding the environment, reverse the map of the dilation (trivial if reversible), then remove the environment. \textit{A priori}, this recipe is different from the Bayesian one applied on the system alone: dilations are not unique, and the reverse process obtained by this recipe might depend on the details of the chosen dilation. We proceed to prove that the two recipes actually coincide: given a choice of dilation, only the reference prior chosen on the system determines the reverse process. This holds both for classical and quantum systems.

\subsection{Classical Dilations \& Bayes' Rule}

Every classical process $\cchn$ on a system state space $A$ ($|A|=d_A$) can always be expressed as a marginal of a global process $\Phi_{AB}$ on an extended state space $AB$ ($|B|=d_B$), alongside some potentially-correlated environment $\cbath_{B}$. This may be expressed by:
\begin{eqnarray} \label{eq:cdil}
    \cchn(a'|a)= \sum_{bb'} \Phi(a',b'|a,b)\cbath(b|a) 
\end{eqnarray}
A tuple $(\Phi_{AB}, \cbath_{B})$ that fulfills \eqref{eq:cdil} will be called a \textit{dilation} of $\cchn$. We proceed to prove our claim for classical processes: 

\begin{result}\label{res1}
    Given a classical map $\cchn$, the reverse obtained by dilating with an environment, reversing the dilated map, then removing the environment, is the same as that obtained directly through the Bayesian recipe \eqref{bayes} on the system. 
\end{result}

\begin{proof}
    Let us construct the reverse with the dilation. Besides the reference prior on the system A, we have the additional freedom of choosing a dilation $(\Phi_{AB},\beta_B)$. The total reference prior is then $\Gamma(a,b)=\gamma(a)\beta(b|a)$, and we define the reverse of the dilated map by applying Bayes' rule to $P(a,b,a',b')\equiv\Phi(a',b'|a,b)\Gamma(a,b)$:
    \begin{equation}
\Phi(a',b'|a,b)\Gamma(a,b)=\hat{\Phi}_\Gamma (a,b|a',b')\Phi[\Gamma](a',b')\,.
    \end{equation}
    Finally, we remove the environment. Writing $\eta:=\Phi[\Gamma]$ for readability, we have
\begin{eqnarray*}
    \underbrace{\sum_{bb'} \Phi(a',b'|a,b)\beta(b|a)}_{\cchn(a'|a)} 
    \crf(a)&=& \sum_{bb'}\hat{\Phi}(a,b|a',b')\underbrace{\Phi[\Gamma](a',b')}_{\eta(b'|a')\eta(a')}\\
    &:=& \hat{\varphi}'_\Gamma (a|a')\eta(a').
\end{eqnarray*}    
where on the left-hand side we have used the fact that $(\Phi_{AB},\beta_B)$ is a dilation of $\cchn$, and where $\hat{\varphi}'_\Gamma (a|a')$ is the reverse map obtained through this recipe. By summing on both sides over $a$, we see that $\eta(a')=\varphi[\gamma](a')$: whence $\hat{\varphi}'_\Gamma (a|a')$ is identical to Eq.~\eqref{bayes}. In particular, the only freedom left is indeed that of choosing $\gamma_A$.
\end{proof}
Notice that we did not have to assume that $\Phi_{AB}$ is reversible: the proof is valid for any dilation. Also, we did not have to assume that $\beta(b|a)=\beta(b)$ carries no initial correlations; of course, one can choose a dilation with this property, if deemed physically important. By contrast, having chosen the dilation, the posterior $\eta:=\Phi[\Gamma]$ is what it is: one cannot enforce $\eta(b'|a')=\eta(b')$.

\subsection{Assignment Maps}
Let us now have a more detailed look at the structure of dilations. The first operation (appending an environment) appears as the natural reverse of the last operation (tracing out the environment). We are going to show that this is indeed the case (see Appendix \ref{app:proofs} for supplementary proofs).

Denote the operation of tracing out the environment by $\ctrace_B$. The map of appending an environment B to the system A (\textit{assignment map}), with reference state on AB given by $\Lambda$, is given by the Bayesian reverse of $\ctrace_B$:
\begin{equation}\label{eq:cgenassign}
    \hat{\ctrace}_{B,\Lambda}[\bullet_A]= \bullet_A \cdot \left(\frac{\Lambda}{\ctrace_B [\Lambda]}\right)_{\! B|A}\,.
\end{equation} Explicitly, $\hat{\ctrace}_{B,\Lambda}[p_A](a,b)=p(a)\Lambda(b|a)$ has the form of \textit{Jeffrey's update}: given a reference joint distribution $\Lambda(a,b)$, if one gets the information that the distribution of $A$ is actually given by $p$, the rational way of updating one's knowledge is to update $A$'s marginal while keeping what attains to $B$ unchanged.

In turn, the Bayesian reverse of any classical assignment map, for which $\Lambda(a,b)$ is product, is the partial trace, for any choice of reference prior: $\cret{\hat{\ctrace}_{B,\square \otimes \beta},\gamma}=\ctrace_B$ for all $\beta$ and $\gamma$.

The generic definition \eqref{eq:cdil} of the dilation $(\Phi_{AB},\beta_B)$ can then be written as
\begin{equation}
\cchn= \ctrace_B \circ {\Phi} \circ \hat{\ctrace}_{B, \Lambda}
\end{equation} with a choice of $\Lambda$ such that $\frac{\Lambda(a,b)}{\sum_a\Lambda(a,b)}=\beta(b|a)$. By using the composability property \eqref{composability_c}, one obtains
\begin{equation}
\hat{\varphi}_\gamma = \ctrace_B \circ \hat{\Phi}_\Gamma \circ \hat{\ctrace}_{B, \Phi[\Gamma]}
\end{equation} which is what we proved in Result \ref{res1}. For relevant proofs see Appendix \ref{app:proofs-ccomp-bayes}.

Classically, by choosing $\Lambda(b|a)\neq \Lambda(b)$, initial system-bath correlation are straightforwardly described. The quantum analog, by contrast, took some discussions to be clarified
\cite{pechukas,alicki-comment,buscemi-NCPTP}. The Petz map of the partial trace satisfies all the properties of a completely positive assignment map. For a generic reference state $\Omega$ of $AB$, it reads
\begin{equation}\label{qass}
        \qam{\Omega}[\bullet_A] = \sqrt{\Omega}\left[ \left(\frac{1}{\sqrt{\TrB[\Omega]}} \bullet_A 
        \frac{1}{\sqrt{\TrB[\Omega]}} \right) \otimes \one_B \right] \sqrt{\Omega}\,. 
\end{equation} For an uncorrelated $\Omega=\square \otimes \beta$, it takes the form \begin{equation}
\label{eq:qassprod}
\qam{\square \otimes \beta}[\bullet_A]  = \bullet_A \otimes \beta\,.
\end{equation}

With this definition of the \textit{quantum assignment map}, we now tackle retrodiction on dilations in the quantum formalism.

\begin{figure*}
\includegraphics[width=0.65\textwidth]{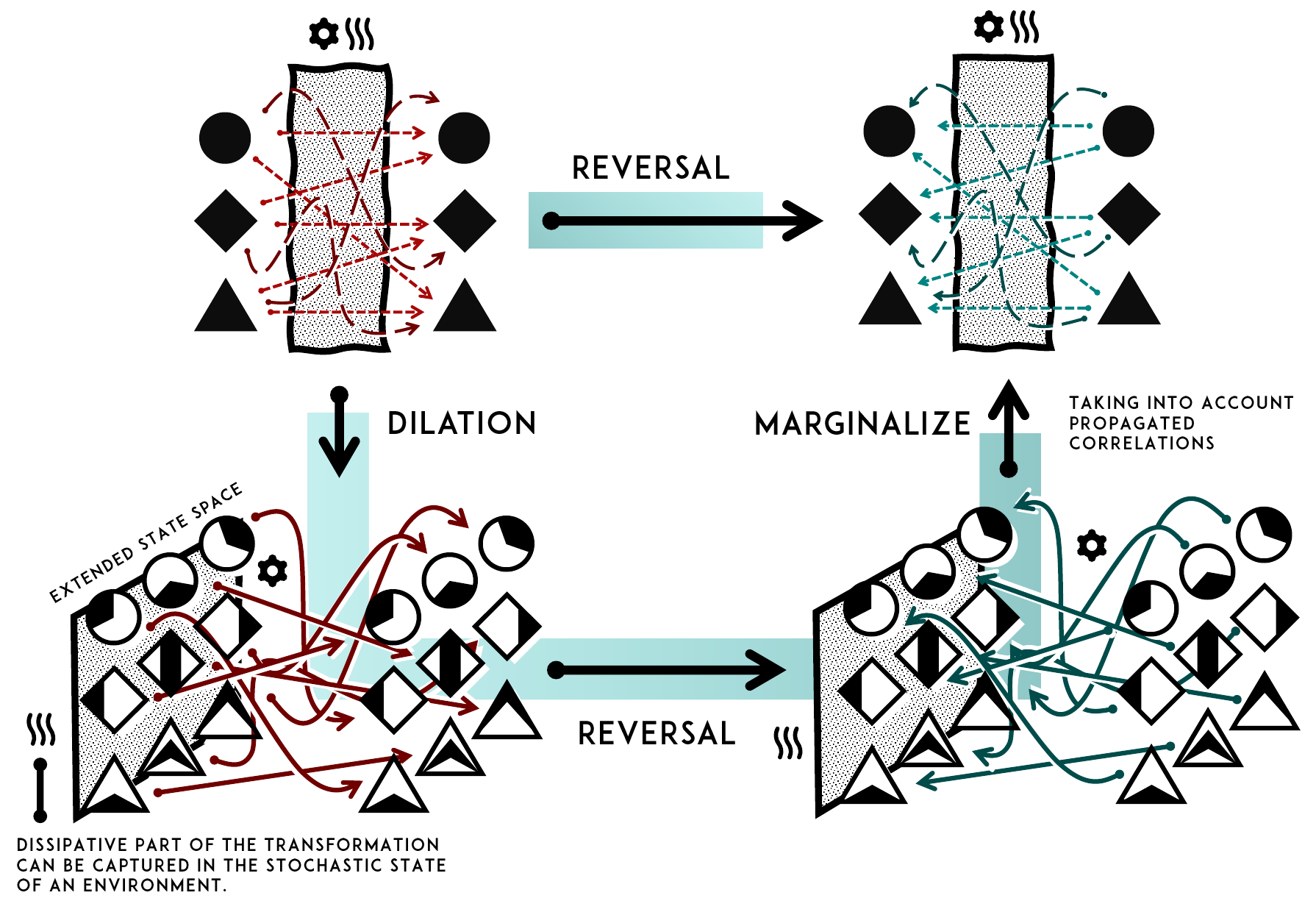}
    \caption{Two routes for Bayesian retrodiction illustrated. One can show that these two protocols always give the same reverse process, as long as the propagated correlations formed across the reference prior and the ancillary environment is accounted for. This is captured by the retrodictive assignment map \eqref{ptz3}.}
    \label{fig:DilDiag} 
\end{figure*}

\subsection{Quantum Dilations \& the Petz Recovery Map}
For quantum processes, we focus on unitary dilations with an initially uncorrelated state of the environment. Any quantum channel $\chn$ can be seen as the marginal of a global unitary $\mathcal{U}$ acting on a target input and an ancillary system prepared in a suitable density operator $\beta$ \cite{wilde_2013, nielsen-chuang}: \begin{equation}\label{eq:chndil}
    \chn[\bullet] = \TrB \left[U (\bullet \otimes \beta) U^\dagger\right] 
\end{equation}
As in the classical case, it can be seen as the composition of channels
\begin{equation}\label{eq:chndildecomp}
    \chn[\bullet] = \TrB{} \circ \mathcal{U} \circ \qam{\square \otimes \beta} [\bullet]\,.
\end{equation}
The Petz map of $\chn$ with reference state $\alpha$ can be computed directly, using the fact that the adjoint of \eqref{eq:chndil} is $\chn^\dag[\bullet] = \TrB[\sqrt{\one \otimes \beta} U^\dag (\bullet \otimes \one) U \sqrt{\one \otimes \beta}]$:
\begin{widetext}
\begin{equation}
\begin{aligned}\label{ptz3}
    \petz[\bullet] &= \TrB \Bigg\{\sqrt{\alpha \otimes \beta}\, U^\dag 
    \left[
        \left(\frac{1}{\sqrt{\TrB[U (\alpha \otimes \beta) U^\dagger]}} \, \bullet \, 
        \frac{1}{\sqrt{\TrB[U (\alpha \otimes \beta) U^\dagger]}}\right)
        \otimes \one
    \right]
    U \,\sqrt{\alpha \otimes \beta}\Bigg\}\\ &=\TrB \Bigg\{ U^\dag 
    \underbrace{\left[
        \sqrt{U (\alpha \otimes \beta) U^\dagger}
        \left[
        \left(\frac{1}{\sqrt{\TrB[U (\alpha \otimes \beta) U^\dagger]}} \, \bullet \, 
        \frac{1}{\sqrt{\TrB[U (\alpha \otimes \beta) U^\dagger]}}\right)
        \otimes \one
        \right] 
        \sqrt{U (\alpha \otimes \beta) U^\dagger}
    \right]}_{\qam{\, \mathcal{U}[\alpha\otimes\beta]}[\bullet]}
    U \Bigg\} \\
    &= \TrB{} \circ \mathcal{U}^\dag \circ \qam{\, \mathcal{U}[\alpha\otimes\beta]} [\bullet]\,,
\end{aligned}
\end{equation}
\end{widetext} where between the first and the second line we inserted two identities $U^\dagger U=\one$ and used $U\sqrt{ (\alpha \otimes \beta)}U^\dagger=\sqrt{U (\alpha \otimes \beta) U^\dagger}$, \changessecond{and where we used \eqref{qass}}.

\begin{table*}
\begin{tabular}{@{}lcc@{}}
\toprule
                                         & Classical                                                                                       & Quantum                                                                                   \\ \midrule
Dilation Definition                      & $\cchn(a'|a) = \sum_{bb'} \Phi(a'b'|ab) \beta(b)$                                               & $\chn[\bullet] = \TrB [U \bullet \otimes \beta U^\dag]$                                   \\
Dilation as Channels                     & $\cchn = \ctrace_B \circ \Phi \circ \hat{\ctrace}_{B,\square \otimes \cbath}$                   & $\chn = \TrB{} \circ \mathcal{U} \circ \qam{\square \otimes \beta}$                       \\
Bayesian Inversion                       & $\rvp = \ctrace_B \circ \Phi^{-1} \circ \hat{\ctrace}_{B, \Phi[\crf \otimes \cbath]}$           & $\petz = \TrB{} \circ \mathcal{U}^\dag \circ \qam{\, \mathcal{U}[\alpha\otimes\beta]}$    \\ \midrule
\multicolumn{3}{c}{Via decomposability}                                                                                                                                                                                                \\ \midrule
Assignment Map to Partial Trace          & $\cret{\hat{\ctrace}_{B,\square \otimes \cbath},\crf} = \ctrace_B$                              & $\ret{\qam{\square \otimes \beta}, \rf} = \TrB{}$                                         \\
Inversion of Global Process              & $\cret{\Phi, \crf \otimes \cbath} = \Phi^{-1}$                                                  & $\ret{\mathcal{U}, \rf \otimes \beta} = \mathcal{U}^\dag$                                 \\
Partial Trace to Retrodictive Assignment & $\cret{\ctrace_B,\Phi[\crf \otimes \cbath]} = \hat{\ctrace}_{B,\Phi[\crf \otimes \cbath]}$      & $\ret{\TrB{}, \mathcal{U}[\rf \otimes \beta]} = \qam{\, \mathcal{U}[\alpha\otimes\beta]}$ \\ \midrule
\multicolumn{3}{c}{If no correlations are formed, and reverse map is well-defined $\Rightarrow$ tabletop reversibility}                                                                                                                \\ \midrule
Tabletop Time-Reversibility              & $\rvp^{\text{TR}} = \ctrace_B \circ \Phi^{-1} \circ \hat{\ctrace}_{B, \square \otimes \cbath'}$ & $\petz^{\text{TR}} = \TrB{} \circ \mathcal{U}^\dag \circ \qam{\square\otimes\beta'}$  \\ \bottomrule   
\end{tabular}
\caption{Summary table of the relation between retrodiction and dilations. To stress the comparison between classical and quantum theory, in this table the classical dilation $\Phi$ is a reversible channels, and the bath is taken as initially uncorrelated with the system (although these restrictions are not needed, as proved in the text).}
\label{tab:dilretsyms}
\end{table*}

Thus, we verified directly the composition that was expected on formal grounds \eqref{composability_q}. We can then state the claimed: 

\begin{result}\label{res2}
    Given a quantum map $\chn$, the reverse obtained by dilating with an environment, reversing the dilated map (accounting for propagated correlations through the assignment map), then removing the environment, is the same as the Petz map \eqref{petz} computed directly on the map. In other words, the knowledge of a dilation of $\chn$ does not add any useful information to define the reverse of $\chn$. 
\end{result}

\begin{proof}
The recipe through the dilation is the composition of 
\begin{equation}
    \begin{aligned}\label{eq:qdildecomp}
    & \ret{\qam{\square \otimes \beta}, \rf} = \TrB{} \\
    & \ret{\mathcal{U}, \qam{\square \otimes \beta}[\rf]} = \hat{\mathcal{U}}_{\rf \otimes \beta} = \mathcal{U}^\dag \; \; \; \because \eqref{eq:deterministic} \\ 
    & \ret{\TrB{}, \mathcal{U}\circ \qam{\square \otimes \beta}[\rf]} = \qam{\, \mathcal{U}[\alpha\otimes\beta]}.
\end{aligned}
\end{equation}
 Proofs for each individual part of the decomposition are found in Appendix \ref{app:retonqdil}. This composition indeed coincides with the Petz map, as proved in \eqref{ptz3}.
\end{proof} 

We summarize the structural symmetries that relate dilation and retrodiction, for classical and quantum theory, in Table \ref{tab:dilretsyms}. In both regimes, the role of \textit{retrodictive assignment maps} $\qam{\, \mathcal{U}[\alpha\otimes\beta]}, \hat{\ctrace}_{B, \Phi[\crf \otimes \cbath]}$ ensures consistency in expressing the reverse process. \changessecond{The last line of the table anticipates the definition of tabletop reversibility, the central object of this paper, which we are going to introduce next.}

\section{Definition of tabletop reversibility and related classes of channels}
\label{sec:ttr}

\changessecond{In this Section, we introduce the new classes of channels that are the central object of this work.} From here onward, we work only in the quantum formalism. When required, we shall highlight whether a result is purely quantum, or is also true for classical processes.

\subsection{Tabletop Reversibility}

\changessecond{Our primary concern is the implementation of a Petz map $\petz$, given the control on the implementation of the channel $\chn[\bullet] = \TrB[U(\bullet\otimes\beta)U^\dag]$.

Implementing the Petz map is not straightforward, and approximate realisations have been studied recently \cite{quek2020quantum,numerical-retrodiction}. With what we introduced, we can understand the reason of this difficulty. On the one hand, since the Petz map is a CPTP map, there exist a unitary $V$ and an ancillary state $\bar{\beta}$ such that $\petz[\bullet] = \text{Tr}_{\bar{B}}\bqty{V \pqty{\bullet\otimes\bar{\beta}} V^\dagger}$. But in general, $V\neq U^\dagger$: we may have to build a dedicated unitary. On the other hand, we have just seen in Eq.~\eqref{ptz3} that every Petz map can be written as $\petz[\bullet] = \TrB\bqty{U^\dagger  \qam{\, \mathcal{U}[\alpha\otimes\beta]} [\bullet] U}$. But in general, $\qam{\, \mathcal{U}[\alpha\otimes\beta]} [\bullet]\neq \bullet\otimes\beta'$, as shown in \eqref{qass}: we may have to do something more complicated than appending an ancilla.

We want to identify the special cases, in addition to unitary and isothermal channels, where \textit{the reverse channel can be obtained by just appending an ancilla and inverting the unitary}:
\begin{defn} \label{def-tabletopR}
A quantum channel $\chn$ with a unitary dilation $\chn[\bullet] = \TrB[U(\bullet\otimes\beta)U^\dag]$ is called \textbf{tabletop reversible for the prior $\alpha$} [shorthand $\TR(\alpha,\beta'|U,\beta)$] if there exists a state $\beta'=\beta'(\alpha)$ of the bath such that the Petz map with respect to $\alpha$ is
\begin{equation}\label{eqfriendly}
    \petz[\bullet] = \TrB\bqty{U^\dagger \pqty{\bullet\otimes\beta'} U}
\end{equation}
for the same $U$ that enters the dilation of $\mathcal{E}$.
\end{defn}
Notice that this definition does not mean that the reverse should be implementable by acting only on the system, a situation studied by Aberg \cite{aberg-quantum-fluct} and inspired by dynamical decoupling. Even in the generic case of isothermal processes one may have to invert the system-bath interaction, if the latter is not constant.}

\subsection{Product-Preserving Maps \& Generalized Thermal Operations}

Here we introduce two more definitions that will be used below.

\begin{defn}
   A unitary that acts in a joint Hilbert space $\mathcal{H}_A\otimes \mathcal{H}_B$ is \textbf{product-preserving} if
\begin{equation}\label{eq:pipo-definition}
    \exists(\alpha, \beta, \alpha', \beta'): \, U(\alpha \otimes \beta)U^\dag = \alpha' \otimes \beta'
\end{equation}
Here, $\alpha,\alpha' \in S(\mathcal{H}_A)$ and $\beta, \beta' \in S(\mathcal{H}_B)$ are positive semidefinite operators with trace one, and at least one of $\alpha$, $\beta$ is not maximally mixed (to exclude the obvious case $U(\one \otimes \one)U^\dag=\one \otimes \one$). We shall also call any tuple $(U,\alpha,\beta)$ for which $U$ is product-preserving with respect to $\alpha$ and $\beta$ a \textbf{product-preserving tuple}.
\end{defn}
Contrary to the production of \textit{correlations} (e.g.~in universal entanglers \cite{karol2015certainty,brahmachari2022dual}) and the preservation of maximal entanglement (Bell-to-Bell maps \cite{brahmachari2022dual}), product-preservation has not been studied systematically prior to this work. It appears as a natural property of thermal maps, and is at the origin of several results in entropy production, thermalization and reversibility in the quantum regime \cite{horo-oppen-thermal,AWWW18,santos_landi,landi-pater-21-review}. By relaxing Eq.~\eqref{eq:thermal-hamiltonian-definition}, we enlarge that natural setting in a way that was already used in some other works in the literature \cite{AWWW18,Alhambra16}:
\begin{defn}
    A unitary that acts on a joint Hilbert space $\mathcal{H}_A\otimes \mathcal{H}_B$ is a \textbf{generalized thermal unitary} if
\begin{equation}\label{eq:gg-definition} 
\begin{aligned}
    &\exists(H_A , H_B, H'_A, H'_B): \\
    &\qquad U(H_A\otimes\one + \one\otimes H_B)U^\dag 
    = H_A'\otimes\one + \one \otimes H_B',
    \end{aligned}
\end{equation}
where either $H_A$ or $H_B$ is not proportional to the identity. The corresponding \textbf{generalized thermal map} is given by $\chn[\bullet] = \TrB\{U [\bullet\otimes \,\therm{H_B}] U^\dag\}$. 
\end{defn}
It is straightforward  (see Appendix~\ref{GTherm}) that 
\begin{equation}\label{eq:generalized-themal-full-rank}
    U \therm{H_A} \otimes \therm{H_B} U^\dag = \therm{H'_A} \otimes \therm{H'_B}\,,
\end{equation} that is, $(U,\therm{H_A},\therm{H_B})$ are a product-preserving tuple for all generalized thermal operations.

\section{Results}\label{ResSec}

Having defined tabletop time-reversibility \eqref{eqfriendly}, product-preservation \eqref{eq:pipo-definition} and generalized thermal channels \eqref{eq:gg-definition}, we prove several results connecting them. 

In subsection \ref{ss:gen}, we prove some general connections between these classes of processes. At a glance:
\begin{itemize}
 \item Theorem \ref{thm:FR1} fully characterizes the generalized thermal as a subset of product-preserving unitaries.
 \item Theorem \ref{thm:local-spectra-preservation} establishes a strong relation between the input pair $(\alpha,\beta)$ and the output pair $(\alpha',\beta')$ of any product-preserving unitary (and Corollary \ref{thm:local-energy-conservation} interprets that relation in the thermal case).
 \item Theorem \ref{thm:main} is almost obvious, but is central to our work: it proves that product preservation leads to tabletop reversibility.
 \item Theorem \ref{thm:TR-not-nec-PP} proves that the converse is not true, by exhibiting examples of tabletop reversibility that do not arise from product preservation.
\end{itemize}
In subsection \ref{ss:qubits}, we present a thorough study for two-qubit unitaries. At a glance:
\begin{itemize}
    \item Theorem \ref{thm:all-GG} provides a parametric characterization of two-qubit generalized thermal unitaries.
    \item Theorem \ref{thm:all-2q-U-PIPO} shows that, given any two qubit unitary $U$ and a pure state $\ket{\beta}$, there always exists $\ket{\alpha}$ such that $U(\ket{\alpha}\otimes\ket{\beta})$ is product. This is unique to low-dimensional cases, as the existence of non-product-preserving unitaries have been proven in higher dimensions \cite{UniversalEntangler}.
    \item Theorems \ref{thm:GG-vs-PP} and \ref{thm:puremixed} provide two connections between the generalized thermal character and the product-preservation properties of two-qubit unitaries.
\end{itemize}

\subsection{General Results}
\label{ss:gen}

\begin{theorem}\label{thm:FR1}
     $U$ is generalized thermal if and only if it is product-preserving with regard to \textit{full rank} $\alpha, \beta$.
\end{theorem}
\begin{proof}
The ``only if'' direction has already been established in Eq.~\eqref{eq:generalized-themal-full-rank}. For the ``if'' implication, we note that for any full-rank state $\alpha$, one can always find a corresponding $H_A$ for which $\alpha = \therm{H_A} = \exp{\kappa(H_A - Z_{H_A})}$, and likewise for the bath state $\beta$ with the same inverse temperature $-\kappa$. Then, \begin{eqnarray*}
\ln(\alpha\otimes\beta)&=&\kappa(H_A\otimes\one+\one\otimes H_B)-(Z_{H_A}+Z_{H_B})\one\otimes\one.
\end{eqnarray*}
The same construction can be done for the logarithm of the output product states $\ln(\alpha'\otimes\beta')$. By invoking
\begin{equation} \label{uonf}
    U A U^\dag = A' \; \iff \; U f(A) U^\dag = f(A'), 
\end{equation}
the product-preserved behavior of the tuple implies $U\ln(\alpha\otimes\beta)U^\dagger=\ln(\alpha'\otimes\beta')$, which is Eq.~\eqref{eq:gg-definition}. Thus there will always exist for every product-preserved tuple with full rank $\alpha, \beta$ some $H_A, H_B, H_A', H_B'$ such that $U$ is a generalized thermal unitary. \end{proof}

Notice that, in the previous theorem, the condition of full rank cannot be relaxed. Indeed, on the one hand, the logarithm of rank-defective states is ill-defined. On the other hand, it is simple to find unitaries that preserve one pure product state, and that are not even close (in any meaningful distance) to a generalized thermal unitary. One such example is the two-qubit unitary $U\ket{00}=\ket{00}$, $U\ket{01}=\frac{1}{\sqrt{3}}(\omega\ket{01}+\omega^*\ket{10}+\ket{11})$, $U\ket{10}=\frac{1}{\sqrt{3}}(\omega^*\ket{01}+\omega\ket{10}+\ket{11})$, $U\ket{11}=\frac{1}{\sqrt{3}}(\ket{01}+\ket{10}+\ket{11})$ with $\omega=e^{2\pi i/3}$.

Next, in the definition of product-preservation, we have merely required the input states $\alpha\otimes \beta$ and output states $\alpha'\otimes\beta'$ to be uncorrelated. Now we show that product-preservation implies a stronger relationship between the two:
\begin{theorem}\label{thm:local-spectra-preservation}
If $(U,\alpha,\beta)$ is a product-preserving tuple with output states $\alpha'$ and $\beta'$, then there exist local unitaries $u_A$ and $u_B$ such that $\alpha'\otimes\beta' = (u_A\alpha u_A^\dag)\otimes(u_B\beta u_B^\dag)$. In the case $d_A=d_B$, this may hold up to a swap: i.e.~it could be $\alpha'\otimes\beta' = (u_A\beta u_A^\dag)\otimes(u_B\alpha u_B^\dag)$.
\end{theorem}
\begin{proof}
    We denote $\sigma [\rho]$ as the eigenspectrum of $\rho$, and write the spectra of $\alpha$ and $\beta$ in the following way:
    \begin{enumerate}
        \item $\sigma [\alpha] = \{ \lambda_1, \dots, \lambda_{d_A}\}$, where $\forall i: \lambda_i \geq \lambda_{i+1}$;
        \item $\sigma [\beta] = \{ \mu_1, \dots, \mu_{d_B}\}$, where $\forall j: \mu_j \geq \mu_{j+1}$.
    \end{enumerate}
    
    As such, $\sigma[\alpha \otimes \beta]$ is given by the set of values $m_{i,j}=\lambda_i \mu_j$. This implies that $m_{i,j} \geq m_{i+1,j}$ and $m_{i,j} \geq m_{i,j+1}$ for all $i,j$. Since the input and output products differ by a unitary transformation, $\sigma[\alpha' \otimes \beta']$ is also given by the same set of $m_{i,j}$, for which some $\sigma [\alpha'] = \{ \lambda_1', \dots, \lambda_{d_A}'\}$ and $\sigma [\beta'] = \{ \mu_1', \dots, \mu_{d_B}'\}$ exist such that $m_{i,j} = \lambda_i' \mu_j'$. It is therefore necessary that $\forall i: \lambda_i' \geq \lambda_{i+1}'$, and likewise for $\mu_j'$. 
    
    Now, assume that there is some $i$ for which $\lambda_i' = c \lambda_i$. This means that for every $j$, $\mu_j' = c^{-1} \mu_j$. But since $\sum_{j} \mu_j=1$ and $\sum_{j} \mu_j' =1$, summing over $j$ for both sides of $\mu_j' = c^{-1} \mu_j$ gives $c=1$. This argument also works for the values of $\mu_j$. Therefore, the spectra must always be conserved.
    
    An ``up to a swap'' is obtained if we begin by assuming that there is some $i$ for which $\lambda_i' = c \mu_i$. Thus, $(U, \alpha, \beta)$ is a product-preserving tuple if and only if the spectra of $\alpha$ and $\beta$ are conserved up to a swap:
    \begin{equation}
    \begin{aligned}
    & \sigma[\alpha']= \sigma[\alpha] \wedge \sigma[\beta']=\sigma[\beta]  \\
    \vee \quad & \sigma[\alpha']=\sigma[\beta] \wedge \sigma[\beta']=\sigma[\alpha].
    \end{aligned}
    \end{equation}
    Note that the second set of conditions can only be fulfilled if $|\text{supp}(\alpha)| \leq d_B$ and $|\text{supp}(\beta)| \leq d_A$.
    
    Since the spectra are conserved, up to a swap, we can always find some local unitary $u_A$ that brings $\alpha$ to $\alpha'$, and similarly for $u_B$, which completes the proof.
\end{proof}

Notably, Theorem~\ref{thm:local-spectra-preservation} shows that demanding the global unitary to preserve noncorrelation is enough to ensure that the \emph{local spectra} of the input states are preserved, up to a swap between the input and ancilla. This provides us with a corollary on the level of the Hamiltonian:
\begin{corollary}\label{thm:local-energy-conservation}
If $(U,H_A,H_B)$ is a generalized thermal tuple with output Hamiltonians $H_A'$ and $H_B'$, then $H_A' = u_A H_A u_A^\dag$ and $H_B' = u_B H_B u_B^\dag$, or $H_A' = u_A H_B u_A^\dag$ and $H_B' = u_B H_A u_B^\dag$, where $u_A$ and $u_B$ are some local unitaries on the system and ancilla respectively.
\end{corollary}
\begin{proof}
    To prove this, we use Theorem~\ref{thm:local-spectra-preservation} on Eq.~\eqref{eq:generalized-themal-full-rank}, and then apply Eq.~\eqref{uonf} to find $\tau_c(H_A')\otimes\tau_c(H_B') = [u_A \tau_c(H_A) u_A^\dag]\otimes[u_B \tau_c(H_B) u_B^\dag] = \tau_c(u_AH_Au_A^\dag)\otimes\tau_c(u_BH_Bu_B^\dag)$, up to a swap. Finally, noting that the exponential of a Hermitian operator is full rank, we can take the logarithm to complete the proof.
\end{proof}

Let us now state the following theorem

\begin{theorem}\label{thm:main}
    If $(U,\alpha,\beta)$ is a product-preserving tuple, the channel $\mathcal{E}[\bullet]=\TrB\bqty{U (\bullet\otimes\beta) U^\dagger}$ is tabletop reversible for reference prior $\alpha$.
\end{theorem}

\begin{proof}
Using \eqref{eq:qassprod}, it is immediate that
\begin{equation}
    \begin{aligned} \label{pp-tr}
        \mathcal{U} [\rf \otimes \beta] &= \rf'\otimes\beta' \\
        \Rightarrow \petz[\bullet] &= \TrB{} \circ \mathcal{U}^\dag \circ \qam{\rf'\otimes\beta'} [\bullet]\\
        &= \TrB\bqty{U^\dagger \pqty{\bullet\otimes\beta'} U} \\
    \end{aligned}
\end{equation} which is tabletop reversibility \eqref{eqfriendly}.
\end{proof}

Next, we prove that the converse of Theorem \eqref{thm:main} is not true:
\begin{theorem} \label{thm:TR-not-nec-PP}
$\TR(\alpha,\beta'|U,\beta)$ does not imply that $(U,\alpha,\beta)$ is generalized thermal.
\end{theorem}
\begin{proof}
We prove this with two counterexamples.

For the first example, we look at the two-qubit channel $\chn_1[\bullet] := \TrB[U_1(\bullet\otimes\beta_1)U_1^\dag]$ with $U_1 = \ketbra{0} \otimes \one + \ketbra{1}\otimes\sigma_x$ and $\comm{\beta_1}{\sigma_x} \neq 0$.

Given a prior $\alpha_1 = a_0\ketbra{0} + (1-a_0)\ketbra{1}$, it can be verified that $\chn_1[\alpha_1] = \alpha_1$, while $U_1(\alpha_1\otimes\beta_1)U_1^\dag = a_0\ketbra{0} \otimes \beta_1 + (1-a_0) \ketbra{1} \otimes (\sigma_x\beta_1\sigma_x) \neq \alpha_1\otimes\beta_1$. Though separable, this is not a product state, so $(U_1,\alpha_1,\beta_1)$ is not generalized thermal. 

The retrodiction of $\chn_1$ with respect to $\alpha_1$ can be found to be $\hat{\chn}_{1,\alpha_1}[\bullet] = \sum_k \ketbra{k}\bullet\ketbra{k}$. Meanwhile, by setting $\beta_1' = (\one + b_y'\sigma_y + b_z'\sigma_z)/2$, we have $\overline{\chn}_1[\bullet] := \TrB[U_1^\dag(\bullet\otimes\beta_1')U_1] = \hat{\chn}_{1,\alpha_1}[\bullet]$. Therefore, $\TR(\alpha_1,\beta_1'|U_1,\beta_1)$. 

The second example is $\chn_2[\bullet] := \TrB[U_2(\bullet\otimes\beta_2)U_2^\dag]$ with
\begin{equation}
\begin{aligned}
    U_2 &= \spmqty{
         0 & 0 & 0 & 0 & 0 & 1 & 0 & 0 \\
         0 & 1 & 0 & 0 & 0 & 0 & 0 & 0 \\
         0 & 0 & 1 & 0 & 0 & 0 & 0 & 0 \\
         0 & 0 & 0 & 0 & 0 & 0 & 1 & 0 \\
         0 & 0 & 0 & 1 & 0 & 0 & 0 & 0 \\
         0 & 0 & 0 & 0 & 0 & 0 & 0 & 1 \\
         0 & 0 & 0 & 0 & 1 & 0 & 0 & 0 \\
         1 & 0 & 0 & 0 & 0 & 0 & 0 & 0
    }, \\
    \beta_2 &= \frac{1}{217}\left(\small
        \begin{array}{llll}
         70 & 84 e^{\frac{i \pi }{3}} & 35 & 30 e^{-\frac{2 i \pi }{3}} \\
         84 e^{-\frac{i \pi }{3}} & 112 & 42 e^{-\frac{i \pi }{3}} & -35 \\
         35 & 42 e^{\frac{i \pi }{3}} & 21 & 15 e^{-\frac{2 i \pi }{3}} \\
         30 e^{\frac{2 i \pi }{3}} & -35 & 15 e^{\frac{2 i \pi }{3}} & 14 \\
        \end{array}
    \right).
\end{aligned}
\end{equation}
The retrodiction channel $\hat{\mathcal{E}}_{2,\alpha_2}$ with respect to the prior $\alpha_2 = \one/2$ is the depolarizing channel $\hat{\mathcal{E}}_{2,\alpha_2}[\bullet] = \text{Tr}[\bullet]\one/2$. At the same time, $\overline{\mathcal{E}}_2[\bullet] := \TrB[U^\dag (\bullet\otimes\beta_2')U]$ with $\beta_2' = \one/2$ is also the depolarizing channel, so $\TR(\alpha_2,\beta_2'|U_2,\beta_2)$.

Meanwhile, the log negativity of $\Omega_2 := U_2(\alpha_2\otimes\beta_2)U_2^\dag$ is found to be $\log_2\tr\|\Omega_2^{T_B}\|_1 \approx 0.9135$, where $\bullet^{T_B}$ is a partial transpose on the ancilla and $\|\bullet\|_1$ is the trace norm. Since log negativity is an entanglement monotone \cite{log-negativity}, and considering that the maximally-entangled pure state has a log negativity of 1, $\Omega_2$ is in fact very highly entangled. So, $(U_2,\alpha_2,\beta_2)$ is certainly not a product-preserving tuple, and hence is not generalized thermal either.

These examples demonstrate that channels can be tabletop time-reversible without being generalized thermal with respect to certain priors. Note that the first example also holds for the classical case by setting $b_y' = 0$, while the second example is inherently quantum due to the presence of entanglement.
\end{proof}

\begin{center}
\begin{figure}[b]
    \includegraphics[width=.9\columnwidth]{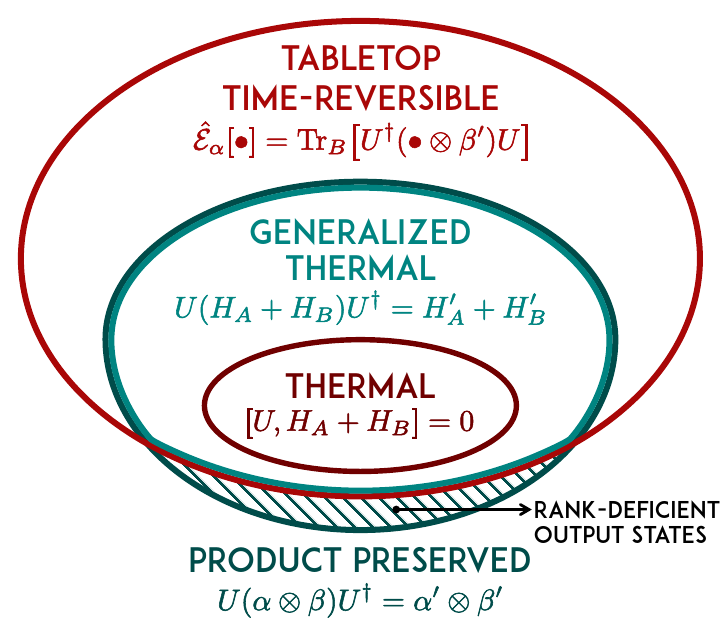}
    \caption{Schematic of the order relationships between the tuple sets we have introduced. Note that the non-tabletop time-reversible but product-preserved tuples are those for which a Bayesian inversion is not well-defined in general due to the presence of rank-deficient output states: although it is still possible to define some retrodiction channel in this case, it will depend on the chosen convention, and the different approaches are inconsistent with each other (Section~\ref{sec:cannotanyhowpipo}, Appendix~\ref{apd:ill-defined-how}).}
    \label{fig:venn1}
\end{figure}
\end{center}

\subsection{Results for Qubit Channels with Two-Qubit Dilations}
\label{ss:qubits}

We shall now focus our attention to one-qubit generalized thermal channels with two-qubit dilations, and product-preserving unitaries acting on two qubits. It is known that every two-qubit unitary permits the Cartan decomposition \cite{two-qubit-decomposition}
\begin{equation}\label{eq:standard-two-qubit-unitary-decomposition}
    U = (u_A\otimes u_B)\exp\left(i\sum_{k=1}^3 t_k \sigma_k\otimes\sigma_k\right)(v_A^\dag \otimes v_B^\dag),
\end{equation}
where $\{\sigma_k\}_{k=1}^3$ are the usual Pauli operators, while $v_A$, $v_B$, $u_A$, and $u_B$ are single-qubit unitaries. Hence, every $U$ would be specified by these local unitaries and angles $\{t_k\}_{k=1}^3$. With reference to this parametrization, we fully characterize all two-qubit generalized thermal unitaries:

\begin{theorem}\label{thm:all-GG}
A two-qubit unitary $U$, parameterized as Eq.~\eqref{eq:standard-two-qubit-unitary-decomposition}, is generalized thermal if and only if $(t_j-t_k)\bmod(\pi/2)=0$ or $(t_j+t_k)\bmod(\pi/2)=0$ for some $j \neq k$.
\end{theorem}

\begin{proof}
    For a given two-qubit unitary $U$, we prove this by characterizing every possible pair of Hamiltonians $H_A$ and $H_B$ such that $(U,H_A,H_B)$ is a generalized thermal tuple. The proof by direct inspection is done in Appendix \ref{app:lemmas2q}, divided in three lemmas: the main one covers all the $U$ such that $t_k \bmod (\pi/4) = 0$ for at most one $t_k$; the other two settle the remaining special cases.
\end{proof}

Due to Theorem \ref{thm:FR1}, the above also fully characterizes all two-qubit unitaries that are product-preserving \textit{with respect to full rank states}. There are in fact many more product-preserving unitaries for two-qubits. Indeed, \textit{every} two-qubit unitary is not just product-preserving, but product-preserving with respect to every pure ancilla:

\begin{theorem}\label{thm:all-2q-U-PIPO}
For every two-qubit $U$ and ancilla $\ket{\beta}$, there exists an $\ket{\alpha}$ such that $(U,\ket{\alpha},\ket{\beta})$ is a product-preserving tuple.
\end{theorem}
\begin{proof}
    Let $v_B^\dag\ket{\beta} \widehat{=} (b_0,b_1)^T$, $v_A^\dag\ket{\alpha} \widehat{=} (a_0,a_1)^T \propto (1,x)^T$, where $v_A$ and $v_B$ are the same as in Eq.~\eqref{eq:standard-two-qubit-unitary-decomposition}, and we have assumed for now that $a_0 \neq 0$. Then,
    \begin{equation}\label{eq:pure-PIPO-equation}
    \begin{aligned}
        (u_A^\dag \otimes u_B^\dag) U (\ket{\alpha}\otimes\ket{\beta}) 
        \;&{}\widehat{=}{} \pmqty{f_{00}(x)\\f_{01}(x)\\f_{10}(x)\\f_{11}(x)},\\
        (u_A^\dag\ket{\alpha'}) \otimes  (u_B^\dag\ket{\beta'}) 
        \;&{}\widehat{=}{} \pmqty{a_0'\\a_1'} \otimes \pmqty{b_0'\\b_1'},
    \end{aligned}
    \end{equation}
    where 
    $f_{jk}(x)$ are linear functions of $x$. The condition that $U$ is product-preserving with respect to $\ket{\alpha}\otimes\ket{\beta}$ requires $(a_0'b_0')(a_1'b_1') = (a_0'b_1')(a_1'b_0')$, which implies
    \begin{equation}\label{eq:quadratic-equation}
        f_{00}(x)f_{11}(x) = f_{01}(x)f_{10}(x).
    \end{equation}
    Since this is at most a quadratic equation in $x$, a solution almost always exists for any choice of $U$, $b_0$, and $b_1$. If a solution for $x$ exists, then $U$ is product-preserving with respect to $\ket{\alpha} = v_A \pqty{\ket{0} + x\ket{1}}/\sqrt{1+\abs{x}^2}$ and the specified $\ket{\beta}$. A solution would not exist for Eq.~\eqref{eq:quadratic-equation} if it results in a contradiction of the form $c=0$ for a nonzero constant $c$. However, in those cases, it is shown in Appendix~\ref{apd:special-case-pure} that $U$ is product-preserving with respect to $\ket{\alpha} = v_A\ket{1}$ and $\ket{\beta}$.
\end{proof}
\begin{corollary}\label{cor:all-2q-U-PIPO}
Every two-qubit unitary is product-preserving with regard to some states.
\end{corollary}
This corollary was already known in the context of universal entanglers, where it has been shown that there is no two-qubit unitary that takes \emph{every} product state to an entangled state \cite{UniversalEntangler}.

Finally, we present an characterisation of the generalized thermal two-qubit unitaries, an alternative to Theorem \ref{thm:all-GG}, in terms of the product states that they preserve:
\begin{theorem}\label{thm:GG-vs-PP}
    A two-qubit unitary $U$ is generalized thermal if and only if it is product-preserving with respect to two pure states $\ket{\alpha_+}\otimes\ket{\beta_+}$ and  $\ket{\alpha_-}\otimes\ket{\beta_-}$, such that $\langle{\alpha_+}|{\alpha_-}\rangle = \langle{\beta_+}|{\beta_-}\rangle = 0$.
\end{theorem}
\begin{proof}
    The ``if'' direction: We shall first consider the case where $H_A$ and $H_B$ are nondegenerate. Let $\ket{\alpha_\pm}$ and $\ket{\beta_\pm}$ be the eigenstates of $H_A$ and $H_B$ respectively, with the ground states labelled $\ket{\alpha_-}$ and $\ket{\beta_-}$. Taking the low temperature limit $\lim_{c\to\infty} U\left[\tau_{\pm c}(H_A)\otimes\tau_{\pm c}(H_B)\right]U^\dag = \lim_{c\to\infty}\tau_{\pm c}(H_A') \otimes \tau_{\pm c}(H_B')$, we have
    \begin{equation}
        U\left[\ketbra{\alpha_\mp} \otimes \ketbra{\beta_\mp}\right]U^\dag = \ketbra{\alpha'_\mp} \otimes \ketbra{\beta'_\mp},
    \end{equation}
    so $U$ is product-preserving with respect to $\ket{\alpha_\pm}\otimes\ket{\beta_\pm}$.

    The ``only if'' direction: Let $\ket{\alpha_\pm}\otimes\ket{\beta_\pm} =: (v_A \otimes v_B )(\ket{\pm}\otimes\ket{\pm})$ and $\ket{\alpha'_\pm}\otimes\ket{\beta'_\pm} =: (u_A\otimes u_B)(\ket{\pm}\otimes\ket{\pm})$. Then,
    \begin{equation}
    \begin{aligned}
    &(u_A^\dag\otimes u_B^\dag)U(v_A\otimes v_B) \\
    &\quad{}\widehat{=}{} \left(\begin{matrix}
    1 & 0 & 0 & 0\\
    0 & e^{i(a+\phi)}\cos(\frac{\theta}{2}) & e^{-i(b-\phi)}\sin(\frac{\theta}{2}) & 0 \\
    0 & -e^{i(b+\phi)}\sin(\frac{\theta}{2}) & e^{-i(a-\phi)}\cos(\frac{\theta}{2}) & 0 \\
    0 & 0 & 0 & 1
    \end{matrix}\right),
    \end{aligned}
    \end{equation}
    where the top-left and bottom-right entries originate from the product-preserving condition, while the rest of the entries are parameterised to impose unitarity. Define $H_A := \omega_A\ketbra{\alpha_+} - \omega_A\ketbra{\alpha_-}$ and $H_B := \omega_B\ketbra{\beta_+} - \omega_B\ketbra{\beta_-}$. From direct computation, $U$ will be found to be generalized thermal with respect to $H_A$ and $H_B$ for any $\omega_A$ and $\omega_B$ if $\theta = 2n\pi$ for some integer $n$, and for $\omega_A = \omega_B$ if $\theta \neq 2n\pi$.

    To handle the degenerate case, we use the full characterization of two-qubit generalized thermal unitaries from Appendix~\ref{app:lemmas2q}. Specifically, lemmas \ref{thm:general-GG-special-1} \& \ref{thm:general-GG-special-2} states that $U$ is generalized thermal with respect to $H_A\not\propto\one$ and $H_B \propto \one$ if and only if it is generalized thermal with respect to $H_A\not\propto\one$ and $\bar{H}_B:= v_B v_A^\dag H_A v_A v_B^\dag$. Since $\bar{H}_B$ shares the same spectrum as $H_A$ and is therefore nondegenerate, the rest of the proof follows as stated above.
\end{proof}

If we have a situation that only \textit{one} of the states in the product-preserved tuple is pure, then we can also conclude that the unitary is generalized thermal:
\begin{theorem}\label{thm:puremixed}
    For a two-qubit $U$ and full-rank $\alpha$, if $(U,\alpha,\ketbra{\beta})$ is a product-preserving tuple, then $U$ is generalized thermal.
\end{theorem}
\begin{proof}
    Up to a swap, $U$ is product-preserving with respect to a full-rank $\alpha$ and pure $\ketbra{\beta}$ if
    \begin{equation}
        U\left(\alpha\otimes\ketbra{\beta}\right)U^\dag = \alpha'\otimes\ketbra{\beta'}, 
    \end{equation}
    \newline
    where we have used Theorem~\ref{thm:local-spectra-preservation} to conclude that $\alpha'$ is full-rank and $\ketbra{\beta'}$ is pure. Consider first the special case of $\alpha = \one/2$. Let $\ket{\beta_+} := \ket{\beta}$ and $\ket{\beta_-}$ to be its orthogonal state. Then,
    \begin{equation}
    \begin{aligned}
        U\pqty{\one\otimes\ketbra{\beta_-}{\beta_-}}U^\dag &= 
        \one - U\pqty{\one\otimes\ketbra{\beta_+}{\beta_+}}U^\dag \\
        &= \one\otimes\ketbra{\beta'_-}.
    \end{aligned}
    \end{equation}
    Hence, $\beta = \tau_\kappa \pqty{\ketbra{\beta_+} - \ketbra{\beta_-}}$ for any $\kappa$ will satisfy $U(\alpha\otimes\beta)U^\dag = \alpha'\otimes\beta'$, so $U$ is generalized thermal.
    
    For $\alpha \neq \one/2$, let $\alpha = p_+\ketbra{\alpha_+} + p_-\ketbra{\alpha_-}$ with $p_+ > p_-$. It is clear that
    \begin{equation}
    \begin{aligned}
        U\pqty{\alpha^n\otimes\ketbra{\beta_+}}U^\dag &= \bqty{U\pqty{\alpha\otimes\ketbra{\beta_+}}U^\dag}^n \\
        &= \alpha^{\prime n}\otimes\ketbra{\beta_+}.
    \end{aligned}
    \end{equation}
    Taking the limit $n\to\infty$ with $\alpha^{ n}/\Tr(\alpha^{n})$ leads to $U\pqty{\ket{\alpha_+}\otimes\ket{\beta_+}} = \ket{\alpha_+'}\otimes\ket{\beta_+'}$. Then, with $\ketbra{\alpha_-} \propto \alpha - p_+\ketbra{\alpha_+}$, where the product form of both terms on the right are preserved by $U$ with respect to the same ancilla state $\ketbra{\beta_+}$, we also have $U\pqty{\ket{\alpha_-}\otimes\ket{\beta_+}} = \ket{\alpha_-'}\otimes\ket{\beta_+'}$.

    Therefore, $U$ is product-preserving with respect to $ \ketbra{\alpha_+} + \ketbra{\alpha_-} = \one$ and $\ketbra{\beta}$. From the first part of the proof, this implies that $U$ is generalized thermal.
\end{proof}
Note that the converse does not hold. From the proof, for a $U$ to be product-preserving with respect to some input state of the form $\alpha\otimes\ketbra{\beta_+}$, it must necessarily be generalized thermal with respect to $\one\otimes\tau_\kappa(\ketbra{\beta_+} - \ketbra{\beta_-})$ for all $\kappa$. From the proofs in Appendix \ref{app:lemmas2q}, this is only possible when there exists $j,k\in\{1,2,3\}$ and $j\neq k$ such that either $t_j,t_k \bmod (\pi/2) = 0$ or $t_j,t_k \bmod (\pi/2) = \pi/4$. Therefore, a generalized thermal $U$ is in general not product-preserving with respect to some $\alpha \otimes \ketbra{\beta}$.

\section{Further observations}
\label{sec:further}

While the preceding sections put emphasis on the mathematical aspects of product-preserving and generalized thermal unitaries, we now consolidate notable physical insights that the results elucidate. We devote one subsection to each of the notions we introduced: generalized thermal operations, product-preservation, and tabletop reversibility. 

\subsection{On Generalized Thermal Maps}
\label{ss:thermal}

Generalized thermal operations are the largest class of operations that can be described by the following procedure \cite{thermodynamics-thesis}: a system (with free Hamiltonian $H_A$) and a bath (with free Hamiltonian $H_B$) are brought into contact, allowed to interact for some time with a unitary $U$, then are decoupled again. As a result, the final Hamiltonian must be without interaction. A scattering experiment, where two distinct collection of particles start far apart, come together to interact in a complicated way, and two (possibly different) collection of particles leave the interaction region, would be an example of such a process. That said, we do not make any assumptions about the interaction unitary $U$ apart from the fact that the initial and final Hamiltonians are decoupled. This procedure indeed describes our definition of generalized thermal operations [Eq.~\eqref{eq:gg-definition}]. 

Corollary~\ref{thm:local-energy-conservation} shows that every process described by the procedure above must obey a strong constraint: the \emph{local conservation of energy spectra}, up to a possible swap (if the system and the bath have the same dimension, we can always choose to redefine which is which).

The requirement [Eq.~\eqref{eq:thermal-hamiltonian-definition}] that the local Hamiltonians be unchanged, which is the standard definition of thermal operation \cite{thermodynamics-thesis,thermodynamics-review}, is strictly stronger. That being said, given a generalized thermal tuple $(U,H_A,H_B)$, up to a swap there exists local unitaries $v_A$ and $v_B$, such that $(U',H_A,H_B)$ is a thermal tuple, with $U'= (v_A \otimes v_B)U$. Indeed, using Corollary~\ref{thm:local-energy-conservation}:
\begin{align}
    &U \bqty{H_A\otimes\one + \one\otimes H_B} U^\dag =  u_A H_A u_A^\dag \otimes \one +\one \otimes u_B H_B u_B^\dag \nonumber \\
    &\underbrace{(u_A^\dag \otimes u_B^\dag) U}_{=:U'} \bqty{H_A\otimes\one + \one\otimes H_B} = \bqty{H_A\otimes\one + \one\otimes H_B} U^{\prime} \nonumber\\
    &\implies [U', H_A\otimes\one + \one\otimes H_B] = 0
\end{align} which is the claimed result with $v_X=u^\dagger_X$.

In Appendix \ref{app:lemmas2q}, we have characterized all possible generalized thermal tuples $(U,H_A,H_B)$ for two qubits; whence all generalized thermal maps that possess a two-qubit dilation can be inferred. This complements the recent characterization of all thermal maps for one qubit coupled to a bosonic bath \cite{bath-given-system} and similar studies in resource theoretic contexts \cite{hu2019thermal}. Such characterizations are necessary for one of the main goals of a resource theory, namely, to identify all states reachable from some initial state using only free operations \cite{future-thermal-cone}. 

\begin{table*}[]
\resizebox{\textwidth}{!}{%
\begin{tabular}{@{}lllll@{}}
\toprule
                                                & \multicolumn{4}{c}{For all $\mathbb{C}^d$ channels}                                                                                                                                                             \\ \midrule
\multirow{2}{*}{For any $\alpha, \beta$} & $U, H_A, H_B$ is a generalized thermal tuple to $H_A', H_B'$ & $\Rightarrow$     & $U, \alpha, \beta$ is product preserved to $\alpha',\beta'$    & Eq.\eqref{eq:generalized-themal-full-rank} \\
                                                & $U, \alpha, \beta$ is product preserved to $\alpha',\beta'$  & $\Rightarrow$     & $\chn(U,\beta)$ is tabletop reversible with regard to $\alpha$ & Eq.\eqref{pp-tr}                           \\ \midrule
For any full-rank $\alpha, \beta$        & $U, \alpha, \beta$ is product preserved to $\alpha',\beta'$  & $\Leftrightarrow$ & $U, H_A, H_B$ is a generalized thermal tuple to $H_A', H_B'$   & Thm.\ref{thm:FR1}                          \\ \bottomrule
\end{tabular}%
}
\caption{Summary of key results for general channels. For relevant results, $\alpha = \therm{H_A}, \beta = \therm{H_B}, \alpha' = \therm{H_A'}, \beta' = \therm{H_B'}$ and $\chn[\bullet] = \TrB{[U  (\bullet \otimes \beta) U^\dag]}$}
\label{tab:t1}
\end{table*}
\begin{table*}[]
\resizebox{\textwidth}{!}{%
\begin{tabular}{@{}lllll@{}}
\toprule
                                             & \multicolumn{4}{c}{For all two-qubits channels}                                                                                                                                                                                                                                                                                                                        \\ \midrule
Full-rank $\alpha, \beta$             & $U, \alpha, \beta$ is product preserved tuple to $\alpha',\beta'$                                          & $\Leftrightarrow$ & $U$ s.t. $t_j -t_k \bmod \frac{\pi}{2}$ or $t_j + t_k \bmod \frac{\pi}{2} =0$ under \eqref{eq:standard-two-qubit-unitary-decomposition} & Thm.\ref{thm:all-GG}        \\ \midrule
\multirow{2}{*}{Pure $\alpha, \beta$} & $\forall U, \ketbra{\beta}{\beta}$                                                                         & $\Rightarrow$     & $\exists \ketbra{\alpha}{\alpha}$ for which $U, \alpha,\beta$ is product preserved                                                                                                     & Thm.\ref{thm:all-2q-U-PIPO} \\
                                             & $(U,\ket{\alpha},\ket{\beta})$, $(U,\ket{\alpha_\perp},\ket{\beta_\perp})$ are product-preserved & $\Leftrightarrow$ & $\exists H_A, H_B$ for which $(U,H_A,H_B)$ is generalized thermal                                                                                                                      & Thm.\ref{thm:GG-vs-PP}      \\ \midrule
Full rank $\alpha$, pure $\beta$          & $U, \alpha, \ketbra{\beta}$ is product preserved tuple to $\alpha',\beta'$                                 & $\Rightarrow$     & $U, H_A, H_B$ is a generalized thermal tuple to $H_A', H_B'$                                                                                                                           & Thm.\ref{thm:puremixed}     \\ \bottomrule
\end{tabular}%
}
    \caption{Summary of key results for channels with two-qubit dilations}
    \label{tab:t2}
\end{table*}

\begin{figure}
    \begin{center}{\centering
    \includegraphics{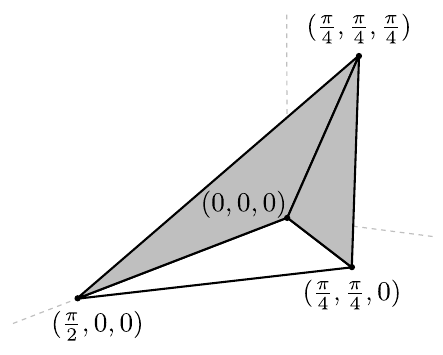}}\end{center}
    \caption{\label{fig:GG-Hamiltonians}All two-qubit unitaries are equivalent up to local unitaries to the Weyl chamber \cite{two-qubit-decomposition}, as plotted here with coordinates $(t_1,t_2,t_3)$. The generalized thermal unitaries are marked out in gray, and occur only on the surfaces $(0,0,0)\text{--}(\frac{\pi}{4},\frac{\pi}{4},\frac{\pi}{4})\text{--}(\frac{\pi}{4},\frac{\pi}{4},0)$ and $(0,0,0)\text{--}(\frac{\pi}{4},\frac{\pi}{4},\frac{\pi}{4})\text{--}(\frac{\pi}{2},0,0)$.}
\end{figure}

\subsection{Product-preservation Involving Rank-Deficient States}\label{sec:cannotanyhowpipo}

The discrepancy between product-preserving and generalized thermal unitaries boils down to the states they are acting upon. In the preceding sections, we have given explicit examples of unitaries that are product-preserving, but not generalized thermal, with respect to pure states.

However, Gibbs states are full rank for finite inverse temperature. So, pure states, or more generally rank-deficient states, are zero-temperature states in the thermodynamic picture, and are not physically feasible except in the limiting sense. Even outside the field of thermodynamics, since real-world experiments are always susceptible to noise and uncertainty, pure states are really idealizations that can never be actually prepared in the lab.

Whichever the motivation, a product-preserving unitary with respect to a rank-deficient state is pragmatically useful only when there is a neighborhood of full-rank states around it whose product structure is also preserved by the same unitary. If so, one can choose a full-rank approximation of the target state whose product structure is also preserved by the same unitary.

For two-qubit unitaries, which are all product-preserving, Theorem~\ref{thm:GG-vs-PP} therefore provides a simple check for when that unitary has a product-preserving property. Its proof also offers an exact construction of the input and ancilla Hamiltonians for which they are a generalized thermal tuple together with interaction unitary.

Meanwhile, the Stinespring dilation theorem asserts the uniqueness of the dilation unitary for every quantum channel when the ancilla state is pure, up to an isomorphism on the ancilla \cite{Stinespring-dilation}. Therefore, for channels with dilations whose input and ancilla are both a single qubit, Theorem~\ref{thm:puremixed} connects the product-preserving property of the Stinespring dilation and the generalized-thermal property of the dilation unitary.

It must be emphasized here that considering full-rank product-preserving states in the neighborhood of a rank-deficient state, and taking the former to be full-rank approximations of the latter, is a convenient choice: one that is motivated by thermodynamic arguments, but an otherwise arbitrary one. Bayes' rule and the Petz alike involve the inverse of the propagated reference. When this state is rank-deficient, an inverse does not exist, rendering such channels undefined. One might consider a neighborhood of full-rank states whose inverses exists, and define the retrodiction channel with some limiting process. However, our prevening discussion and an explicit example in appendix~\ref{apd:cannot-anyhow-PIPO} shows that defining retrodiction this way depends on the chosen neighborhood of full-rank states.

\subsection{Composable Tabletop Time-reversibility}
\begin{figure}[ht]
    \centering
    \includegraphics{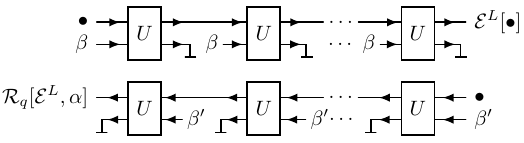}
    \caption{Where the composition of many copies of the same channel involves the same $\beta$ as the ancilla in each step, it can be desirable for the reversal of the composite channel to involve the same $\beta'$ as the ancilla in each step in the opposite direction. We call channels that satisfy Eq.~\eqref{pwfriendly} as composable tabletop time-reversible channels.}
    \label{fig:retrodiction-friendly-chain}
\end{figure}
Apart from the reduction in reversal complexity of a tabletop-time-reversible channel, one might also desire for this behavior to apply for \emph{compositions} of the same channel. For example, when the prior and ancilla are thermal states, the reverse channel of the composition of a thermal operation is the composition of the reverse channel of a single thermal operation \cite{AWWW18}.

More generally, it would be convenient if the reverse channel of a composition of many copies of the same quantum operation can be implemented as illustrated in Fig.~\ref{fig:retrodiction-friendly-chain}. Formally, for a tabletop-time-reversible channel $\chn$ with unitary dilation ($U$, $\beta$) and positive integer $L$, one desires the reverse channel of $\chn^L := \underbrace{\chn\circ\chn\circ \dots \circ \chn}_{L\text{ times}}$ to be
\begin{equation}\label{pwfriendly}
\begin{aligned}
    \ret{\chn^L ,\rf} &= \hat{\chn}_\rf \circ \hat{\chn}_{\chn[\rf]} \circ \hat{\chn}_{\chn^2 [\rf]} \circ \dots \circ \hat{\chn}_{\chn^{L-1} [\rf]} \\
    &\overset{!}{=} \hat{\chn}_\rf \circ \hat{\chn}_{\rf} \circ \hat{\chn}_{\rf} \circ \dots \circ \hat{\chn}_{\rf} \\
    &= \ret{\chn,\rf}^L \quad\forall L \in \mathbb{Z}^+, \text{where}\\
    &\hspace{7em}\hat{\chn}_{\rf}[\bullet] = \Tr[U^\dag(\bullet\otimes\beta')U].
\end{aligned}
\end{equation}
For brevity, we shall use $\TR_c(\alpha,\beta'|U,\beta)$ to denote such \emph{composable} tabletop-time-reversible channels---$\TR(\alpha,\beta'|U,\beta)$ that also satisfy Eq.~\eqref{pwfriendly}. Let us provide a few examples of such channels.

\emph{Unitary channels.} The action of both the unitary and its inverse on the input is independent of the state of the ancilla, so unitary channels are composably tabletop time-reversible. While this is a trivial case, it aligns with the intended definition of composable tabletop time-reversible channels as illustrated in Fig.~\ref{fig:retrodiction-friendly-chain}.

\emph{Reverse channel with $N$-steady state prior.} Consider $\chn[\bullet] = \TrB[U(\bullet\otimes\beta)U^\dag]$ that is \emph{not} a unitary channel, with a set of priors $\{\alpha_0,\alpha_1,\dots,\alpha_{N-1}\}$ that satisfy
\begin{equation}\label{eq:N-steady-cond}
    U(\alpha_n\otimes\beta)U^\dag = \alpha_{(n+1)\bmod N}\otimes\beta'.
\end{equation}
Then, $\hat{\chn}_{\alpha_n}[\bullet] = \Tr[U^\dag(\bullet\otimes\beta')U]$ for all $n$, hence $\TR_c(\alpha_n,\beta'|U,\beta)$. We call this the ``$N$\emph{-steady state}'' as the set $\{\alpha_n\}_{n=0}^{N-1}$ is unchanged under the channel. The $N=1$ case is the usual steady state of a channel, which includes the class of thermal operations as previously studied \cite{AWWW18}. An example of such a channel for any $N$ is an $N$-dit channel $\chn$ given by the ancilla $\beta = \sum_{k=0}^{N-1}b_k\ketbra{k}$ and the dilation unitary
\begin{equation}
    U = \sum_{k=0}^{N-1} u_k\otimes\ketbra{\psi_k}{k}.
\end{equation}
Here, $\sigma$ is a permutation of order $N$, $\{\ket{\psi_k}\}_{k=0}^{N-1}$ and $\{\ket{k}\}_{k=0}^{N-1}$ are orthonormal bases, and
\begin{equation}
    u_k = \sum_{j=0}^{N-1} e^{i\phi_j^{(k)}} \ketbra{\sigma(j)}{j},
\end{equation}
with $u_{k}^\dag u_{k'} \not\propto \one$ for all $k \neq k'$, which ensures that $\chn$ is not a unitary channel. Meanwhile, the priors are defined as $\alpha_n = \sum_{k=0}^{N-1}a_k\ketbra{\sigma^n(k)}$ where $\alpha_n$ is not degenerate, so that $\alpha_n\neq\alpha_{n'}$ for all $n\neq n'$. Since $u_k\ketbra{\sigma^n(j)}u_k^\dag = \ketbra{\sigma^{n+1}(j)}$, we have that
\begin{equation}
\begin{aligned}
    &U(\alpha_n\otimes\beta)U^\dag\\
    &\quad{}={}
    \sum_{k=0}^{N-1}\sum_{j=0}^{N-1} a_jb_k \pqty\Big{u_k \ketbra{\sigma^n(j)} u_k^\dag} \otimes\ketbra{\psi_k} \\
    &\quad{}={}
    \underbrace{\pqty{\sum_{j=0}^{N-1} a_j \ketbra{\sigma^{n+1}(j)}}}_{=\alpha_{(n+1)\bmod N}} \otimes \underbrace{\pqty{ \sum_{k=0}^{N-1}b_k \ketbra{\psi_k} }}_{:=\beta'}.
\end{aligned}
\end{equation}
This satisfies Eq.~\eqref{eq:N-steady-cond}, and hence this channel is composably tabletop time-reversible with respect to $\alpha_n$.

\emph{Channels with idempotent reverse channels.} These are tabletop time-reversible channels whose reverse channels are idempotent, both with themselves and subsequent priors: that is, they have the property $\hat{\chn}_{\alpha}\circ(\hat{\chn}^{L-1})_{\chn[\alpha]} =  \hat{\chn}_\alpha\circ \hat{\chn}_\alpha = \hat{\chn}_\alpha$. If so, Eq.~\eqref{pwfriendly} is also satisfied. An example of such a channel is when the dilation unitary is a swap, as
\begin{equation}
\begin{aligned}
    U(\alpha\otimes\beta)U^\dag &= \beta\otimes\alpha \\
    \implies \hat{\chn}_{\alpha}[\bullet] &= 
    \TrB[U^\dag(\bullet\otimes\alpha)U]
\end{aligned}
\end{equation}
for every $\alpha$. Clearly, $\hat{\chn}_{\alpha}\circ\hat{\chn}_{\chn[\alpha]}\circ\dots\circ\hat{\chn}_{\chn^{L-1}[\alpha]}[\bullet] = \alpha = \hat{\chn}_{\alpha}[\bullet] = \hat{\chn}_{\alpha}\circ \hat{\chn}_{\alpha}[\bullet]$, so the reverse channel with respect to the prior $\alpha$ is idempotent.

These are the classes of composable tabletop time-reversible channels that we have thus far identified. In our study of generalized thermal channels with two-qubit dilations, all $\TR_c$ channels we have found belong to one of the above classes. It is not known if this is an exhaustive list, as there might be richer families of $\TR_c$ channels in higher dimensions.

\section{Conclusion} \label{Concl}

Using a recipe from Bayesian inference, one can associate any physical channel, however irreversible, to a family of reverse channels indexed by the choice of a reference prior. We proved that an apparently different recipe, based on dilating the system to include the bath into which information is dissipated, leads exactly to the same family of reverse channels.

For thermal channels, it was known that a natural reverse channel consists of reversing the unitary evolution while in contact with the same thermal bath (in our framework, this is obtained by choosing the Gibbs state as reference prior). We ask for which other channels such phenomenon happens: that the reverse channel can be implemented with the same devices and similar baths as the original channel (``tabletop reversibility''). These questions further inspire the related definition of ``product-preserving'' maps that contains generalized thermal channels as an important subclass. We then proved several relations between these classes, both in general and with more detail for two-qubit unitaries. In particular, we show that when the reverse channel is well-defined (Section \ref{sec:cannotanyhowpipo}), product-preservation leads to tabletop reversibility (Theorem \ref{thm:main}). As such, with full-rank states, the latter is strictly more general (Theorem \ref{thm:TR-not-nec-PP}). As a by-product of this work, we found that the preservation of local energy spectra is a necessary and sufficient characterization of generalized thermal operations. Characterizing tabletop reversibility in a necessary and sufficient way remains an open problem.

\section*{Acknowledgments}

We thank Francesco Buscemi for many discussions, and in particular for stressing the role of assignment maps. We acknowledge discussions with Mile Gu, Matteo Lostaglio, Eric Lutz, and Nelly Ng. 

This work is supported by the National Research Foundation, Singapore and A*STAR under its CQT Bridging Grant; and by the Ministry of Education, Singapore, under the Tier 2 grant ``Bayesian approach to irreversibility'' (Grant No. MOE-000504-01). M.B.J.~acknowledges support from the Government of Spain (Severo Ochoa CEX2019-000910-S and TRANQI), Fundaci\'o Cellex, Fundaci\'o Mir-Puig, Generalitat de Catalunya (CERCA program), as well as funding from the European Union's Horizon 2020 research and innovation programme under the Marie Sk\l{}odowska-Curie Grant Agreement No. 847517.

\clearpage 
\bibliography{refsretro}

\begin{thebibliography}{66}%
\makeatletter
\providecommand \@ifxundefined [1]{%
 \@ifx{#1\undefined}
}%
\providecommand \@ifnum [1]{%
 \ifnum #1\expandafter \@firstoftwo
 \else \expandafter \@secondoftwo
 \fi
}%
\providecommand \@ifx [1]{%
 \ifx #1\expandafter \@firstoftwo
 \else \expandafter \@secondoftwo
 \fi
}%
\providecommand \natexlab [1]{#1}%
\providecommand \enquote  [1]{``#1''}%
\providecommand \bibnamefont  [1]{#1}%
\providecommand \bibfnamefont [1]{#1}%
\providecommand \citenamefont [1]{#1}%
\providecommand \href@noop [0]{\@secondoftwo}%
\providecommand \href [0]{\begingroup \@sanitize@url \@href}%
\providecommand \@href[1]{\@@startlink{#1}\@@href}%
\providecommand \@@href[1]{\endgroup#1\@@endlink}%
\providecommand \@sanitize@url [0]{\catcode `\\12\catcode `\$12\catcode
  `\&12\catcode `\#12\catcode `\^12\catcode `\_12\catcode `\%12\relax}%
\providecommand \@@startlink[1]{}%
\providecommand \@@endlink[0]{}%
\providecommand \url  [0]{\begingroup\@sanitize@url \@url }%
\providecommand \@url [1]{\endgroup\@href {#1}{\urlprefix }}%
\providecommand \urlprefix  [0]{URL }%
\providecommand \Eprint [0]{\href }%
\providecommand \doibase [0]{https://doi.org/}%
\providecommand \selectlanguage [0]{\@gobble}%
\providecommand \bibinfo  [0]{\@secondoftwo}%
\providecommand \bibfield  [0]{\@secondoftwo}%
\providecommand \translation [1]{[#1]}%
\providecommand \BibitemOpen [0]{}%
\providecommand \bibitemStop [0]{}%
\providecommand \bibitemNoStop [0]{.\EOS\space}%
\providecommand \EOS [0]{\spacefactor3000\relax}%
\providecommand \BibitemShut  [1]{\csname bibitem#1\endcsname}%
\let\auto@bib@innerbib\@empty
\bibitem [{\citenamefont {Gibbs}(1878)}]{gibbs1879equilibrium}%
  \BibitemOpen
  \bibfield  {author} {\bibinfo {author} {\bibfnamefont {J.~W.}\ \bibnamefont
  {Gibbs}},\ }\bibfield  {title} {\bibinfo {title} {On the equilibrium of
  heterogeneous substances},\ }\href {https://doi.org/10.2475/ajs.s3-16.96.441}
  {\bibfield  {journal} {\bibinfo  {journal} {American Journal of Science}\
  }\textbf {\bibinfo {volume} {s3-16}},\ \bibinfo {pages} {441} (\bibinfo
  {year} {1878})},\ \Eprint
  {https://arxiv.org/abs/https://www.ajsonline.org/content/s3-16/96/441.full.pdf}
  {https://www.ajsonline.org/content/s3-16/96/441.full.pdf} \BibitemShut
  {NoStop}%
\bibitem [{\citenamefont {Onsager}(1931)}]{onsager1931reciprocal}%
  \BibitemOpen
  \bibfield  {author} {\bibinfo {author} {\bibfnamefont {L.}~\bibnamefont
  {Onsager}},\ }\bibfield  {title} {\bibinfo {title} {Reciprocal relations in
  irreversible processes. ii.},\ }\href
  {https://doi.org/10.1103/PhysRev.38.2265} {\bibfield  {journal} {\bibinfo
  {journal} {Physical review}\ }\textbf {\bibinfo {volume} {38}},\ \bibinfo
  {pages} {2265} (\bibinfo {year} {1931})}\BibitemShut {NoStop}%
\bibitem [{\citenamefont {Seifert}(2012)}]{Seifert_2012}%
  \BibitemOpen
  \bibfield  {author} {\bibinfo {author} {\bibfnamefont {U.}~\bibnamefont
  {Seifert}},\ }\bibfield  {title} {\bibinfo {title} {Stochastic
  thermodynamics, fluctuation theorems and molecular machines},\ }\href
  {https://doi.org/10.1088/0034-4885/75/12/126001} {\bibfield  {journal}
  {\bibinfo  {journal} {Reports on Progress in Physics}\ }\textbf {\bibinfo
  {volume} {75}},\ \bibinfo {pages} {126001} (\bibinfo {year}
  {2012})}\BibitemShut {NoStop}%
\bibitem [{\citenamefont {Brandao}\ \emph {et~al.}(2015)\citenamefont
  {Brandao}, \citenamefont {Horodecki}, \citenamefont {Ng}, \citenamefont
  {Oppenheim},\ and\ \citenamefont {Wehner}}]{nelly-brandao2015second}%
  \BibitemOpen
  \bibfield  {author} {\bibinfo {author} {\bibfnamefont {F.}~\bibnamefont
  {Brandao}}, \bibinfo {author} {\bibfnamefont {M.}~\bibnamefont {Horodecki}},
  \bibinfo {author} {\bibfnamefont {N.}~\bibnamefont {Ng}}, \bibinfo {author}
  {\bibfnamefont {J.}~\bibnamefont {Oppenheim}},\ and\ \bibinfo {author}
  {\bibfnamefont {S.}~\bibnamefont {Wehner}},\ }\bibfield  {title} {\bibinfo
  {title} {The second laws of quantum thermodynamics},\ }\href
  {https://doi.org/10.1073/pnas.1411728112} {\bibfield  {journal} {\bibinfo
  {journal} {Proceedings of the National Academy of Sciences}\ }\textbf
  {\bibinfo {volume} {112}},\ \bibinfo {pages} {3275} (\bibinfo {year}
  {2015})}\BibitemShut {NoStop}%
\bibitem [{\citenamefont {Evans}\ and\ \citenamefont
  {Searles}(2002)}]{evans2002fluctuation}%
  \BibitemOpen
  \bibfield  {author} {\bibinfo {author} {\bibfnamefont {D.~J.}\ \bibnamefont
  {Evans}}\ and\ \bibinfo {author} {\bibfnamefont {D.~J.}\ \bibnamefont
  {Searles}},\ }\bibfield  {title} {\bibinfo {title} {The fluctuation
  theorem},\ }\href {https://doi.org/10.1080/00018730210155133} {\bibfield
  {journal} {\bibinfo  {journal} {Advances in Physics}\ }\textbf {\bibinfo
  {volume} {51}},\ \bibinfo {pages} {1529} (\bibinfo {year}
  {2002})}\BibitemShut {NoStop}%
\bibitem [{\citenamefont {Crooks}(2008)}]{crooks-reversal}%
  \BibitemOpen
  \bibfield  {author} {\bibinfo {author} {\bibfnamefont {G.~E.}\ \bibnamefont
  {Crooks}},\ }\bibfield  {title} {\bibinfo {title} {Quantum operation time
  reversal},\ }\href {https://doi.org/10.1103/PhysRevA.77.034101} {\bibfield
  {journal} {\bibinfo  {journal} {Phys. Rev. A}\ }\textbf {\bibinfo {volume}
  {77}},\ \bibinfo {pages} {034101} (\bibinfo {year} {2008})}\BibitemShut
  {NoStop}%
\bibitem [{\citenamefont {Campisi}\ \emph {et~al.}(2011)\citenamefont
  {Campisi}, \citenamefont {H\"anggi},\ and\ \citenamefont
  {Talkner}}]{campisi-haenggi-review-2011}%
  \BibitemOpen
  \bibfield  {author} {\bibinfo {author} {\bibfnamefont {M.}~\bibnamefont
  {Campisi}}, \bibinfo {author} {\bibfnamefont {P.}~\bibnamefont {H\"anggi}},\
  and\ \bibinfo {author} {\bibfnamefont {P.}~\bibnamefont {Talkner}},\
  }\bibfield  {title} {\bibinfo {title} {Colloquium: Quantum fluctuation
  relations: Foundations and applications},\ }\href
  {https://doi.org/10.1103/RevModPhys.83.771} {\bibfield  {journal} {\bibinfo
  {journal} {Rev. Mod. Phys.}\ }\textbf {\bibinfo {volume} {83}},\ \bibinfo
  {pages} {771} (\bibinfo {year} {2011})}\BibitemShut {NoStop}%
\bibitem [{\citenamefont {Jaynes}(1957)}]{jaynes1957information}%
  \BibitemOpen
  \bibfield  {author} {\bibinfo {author} {\bibfnamefont {E.~T.}\ \bibnamefont
  {Jaynes}},\ }\bibfield  {title} {\bibinfo {title} {Information theory and
  statistical mechanics},\ }\href@noop {} {\bibfield  {journal} {\bibinfo
  {journal} {Physical review}\ }\textbf {\bibinfo {volume} {106}},\ \bibinfo
  {pages} {620} (\bibinfo {year} {1957})}\BibitemShut {NoStop}%
\bibitem [{\citenamefont {Shannon}(1948)}]{shannon1948mathematical}%
  \BibitemOpen
  \bibfield  {author} {\bibinfo {author} {\bibfnamefont {C.~E.}\ \bibnamefont
  {Shannon}},\ }\bibfield  {title} {\bibinfo {title} {A mathematical theory of
  communication},\ }\href@noop {} {\bibfield  {journal} {\bibinfo  {journal}
  {The Bell system technical journal}\ }\textbf {\bibinfo {volume} {27}},\
  \bibinfo {pages} {379} (\bibinfo {year} {1948})}\BibitemShut {NoStop}%
\bibitem [{\citenamefont {Bennett}(1982)}]{bennett1982thermodynamics}%
  \BibitemOpen
  \bibfield  {author} {\bibinfo {author} {\bibfnamefont {C.~H.}\ \bibnamefont
  {Bennett}},\ }\bibfield  {title} {\bibinfo {title} {The thermodynamics of
  computation—a review},\ }\href@noop {} {\bibfield  {journal} {\bibinfo
  {journal} {International Journal of Theoretical Physics}\ }\textbf {\bibinfo
  {volume} {21}},\ \bibinfo {pages} {905} (\bibinfo {year} {1982})}\BibitemShut
  {NoStop}%
\bibitem [{\citenamefont {Barnum}\ and\ \citenamefont
  {Knill}(2002)}]{barnum-knill}%
  \BibitemOpen
  \bibfield  {author} {\bibinfo {author} {\bibfnamefont {H.}~\bibnamefont
  {Barnum}}\ and\ \bibinfo {author} {\bibfnamefont {E.}~\bibnamefont {Knill}},\
  }\bibfield  {title} {\bibinfo {title} {Reversing quantum dynamics with
  near-optimal quantum and classical fidelity},\ }\href
  {https://doi.org/10.1063/1.1459754} {\bibfield  {journal} {\bibinfo
  {journal} {Journal of Mathematical Physics}\ }\textbf {\bibinfo {volume}
  {43}},\ \bibinfo {pages} {2097} (\bibinfo {year} {2002})}\BibitemShut
  {NoStop}%
\bibitem [{\citenamefont {Hilt}\ \emph {et~al.}(2011)\citenamefont {Hilt},
  \citenamefont {Shabbir}, \citenamefont {Anders},\ and\ \citenamefont
  {Lutz}}]{anders-hilt2011landauer}%
  \BibitemOpen
  \bibfield  {author} {\bibinfo {author} {\bibfnamefont {S.}~\bibnamefont
  {Hilt}}, \bibinfo {author} {\bibfnamefont {S.}~\bibnamefont {Shabbir}},
  \bibinfo {author} {\bibfnamefont {J.}~\bibnamefont {Anders}},\ and\ \bibinfo
  {author} {\bibfnamefont {E.}~\bibnamefont {Lutz}},\ }\bibfield  {title}
  {\bibinfo {title} {Landauer’s principle in the quantum regime},\ }\href
  {https://doi.org/10.1103/PhysRevE.83.030102} {\bibfield  {journal} {\bibinfo
  {journal} {Physical Review E}\ }\textbf {\bibinfo {volume} {83}},\ \bibinfo
  {pages} {030102} (\bibinfo {year} {2011})}\BibitemShut {NoStop}%
\bibitem [{\citenamefont {Jarzynski}(2000)}]{jarz2000}%
  \BibitemOpen
  \bibfield  {author} {\bibinfo {author} {\bibfnamefont {C.}~\bibnamefont
  {Jarzynski}},\ }\bibfield  {title} {\bibinfo {title} {Hamiltonian derivation
  of a detailed fluctuation theorem},\ }\href
  {https://doi.org/10.1023/A:1018670721277} {\bibfield  {journal} {\bibinfo
  {journal} {Journal of Statistical Physics}\ }\textbf {\bibinfo {volume}
  {98}},\ \bibinfo {pages} {77} (\bibinfo {year} {2000})}\BibitemShut {NoStop}%
\bibitem [{\citenamefont {Alhambra}\ \emph {et~al.}(2018)\citenamefont
  {Alhambra}, \citenamefont {Wehner}, \citenamefont {Wilde},\ and\
  \citenamefont {Woods}}]{AWWW18}%
  \BibitemOpen
  \bibfield  {author} {\bibinfo {author} {\bibfnamefont {A.~M.}\ \bibnamefont
  {Alhambra}}, \bibinfo {author} {\bibfnamefont {S.}~\bibnamefont {Wehner}},
  \bibinfo {author} {\bibfnamefont {M.~M.}\ \bibnamefont {Wilde}},\ and\
  \bibinfo {author} {\bibfnamefont {M.~P.}\ \bibnamefont {Woods}},\ }\bibfield
  {title} {\bibinfo {title} {Work and reversibility in quantum
  thermodynamics},\ }\href {https://doi.org/10.1103/PhysRevA.97.062114}
  {\bibfield  {journal} {\bibinfo  {journal} {Phys. Rev. A}\ }\textbf {\bibinfo
  {volume} {97}},\ \bibinfo {pages} {062114} (\bibinfo {year}
  {2018})}\BibitemShut {NoStop}%
\bibitem [{\citenamefont {Manzano}\ \emph {et~al.}(2018)\citenamefont
  {Manzano}, \citenamefont {Horowitz},\ and\ \citenamefont
  {Parrondo}}]{manzano-PRX}%
  \BibitemOpen
  \bibfield  {author} {\bibinfo {author} {\bibfnamefont {G.}~\bibnamefont
  {Manzano}}, \bibinfo {author} {\bibfnamefont {J.~M.}\ \bibnamefont
  {Horowitz}},\ and\ \bibinfo {author} {\bibfnamefont {J.~M.~R.}\ \bibnamefont
  {Parrondo}},\ }\bibfield  {title} {\bibinfo {title} {Quantum fluctuation
  theorems for arbitrary environments: Adiabatic and nonadiabatic entropy
  production},\ }\href {https://doi.org/10.1103/PhysRevX.8.031037} {\bibfield
  {journal} {\bibinfo  {journal} {Phys. Rev. X}\ }\textbf {\bibinfo {volume}
  {8}},\ \bibinfo {pages} {031037} (\bibinfo {year} {2018})}\BibitemShut
  {NoStop}%
\bibitem [{\citenamefont {Buscemi}\ and\ \citenamefont {Scarani}(2021)}]{BS21}%
  \BibitemOpen
  \bibfield  {author} {\bibinfo {author} {\bibfnamefont {F.}~\bibnamefont
  {Buscemi}}\ and\ \bibinfo {author} {\bibfnamefont {V.}~\bibnamefont
  {Scarani}},\ }\bibfield  {title} {\bibinfo {title} {Fluctuation theorems from
  bayesian retrodiction},\ }\href {https://doi.org/10.1103/PhysRevE.103.052111}
  {\bibfield  {journal} {\bibinfo  {journal} {Phys. Rev. E}\ }\textbf {\bibinfo
  {volume} {103}},\ \bibinfo {pages} {052111} (\bibinfo {year}
  {2021})}\BibitemShut {NoStop}%
\bibitem [{\citenamefont {Aw}\ \emph {et~al.}(2021)\citenamefont {Aw},
  \citenamefont {Buscemi},\ and\ \citenamefont {Scarani}}]{AwBS}%
  \BibitemOpen
  \bibfield  {author} {\bibinfo {author} {\bibfnamefont {C.~C.}\ \bibnamefont
  {Aw}}, \bibinfo {author} {\bibfnamefont {F.}~\bibnamefont {Buscemi}},\ and\
  \bibinfo {author} {\bibfnamefont {V.}~\bibnamefont {Scarani}},\ }\bibfield
  {title} {\bibinfo {title} {Fluctuation theorems with retrodiction rather than
  reverse processes},\ }\href {https://doi.org/10.1116/5.0060893} {\bibfield
  {journal} {\bibinfo  {journal} {AVS Quantum Science}\ }\textbf {\bibinfo
  {volume} {3}},\ \bibinfo {pages} {045601} (\bibinfo {year} {2021})},\ \Eprint
  {https://arxiv.org/abs/https://doi.org/10.1116/5.0060893}
  {https://doi.org/10.1116/5.0060893} \BibitemShut {NoStop}%
\bibitem [{\citenamefont {Alhambra}\ \emph {et~al.}(2016)\citenamefont
  {Alhambra}, \citenamefont {Masanes}, \citenamefont {Oppenheim},\ and\
  \citenamefont {Perry}}]{Alhambra16}%
  \BibitemOpen
  \bibfield  {author} {\bibinfo {author} {\bibfnamefont {A.~M.}\ \bibnamefont
  {Alhambra}}, \bibinfo {author} {\bibfnamefont {L.}~\bibnamefont {Masanes}},
  \bibinfo {author} {\bibfnamefont {J.}~\bibnamefont {Oppenheim}},\ and\
  \bibinfo {author} {\bibfnamefont {C.}~\bibnamefont {Perry}},\ }\bibfield
  {title} {\bibinfo {title} {Fluctuating work: From quantum thermodynamical
  identities to a second law equality},\ }\href
  {https://doi.org/10.1103/PhysRevX.6.041017} {\bibfield  {journal} {\bibinfo
  {journal} {Phys. Rev. X}\ }\textbf {\bibinfo {volume} {6}},\ \bibinfo {pages}
  {041017} (\bibinfo {year} {2016})}\BibitemShut {NoStop}%
\bibitem [{\citenamefont {Kwon}\ and\ \citenamefont {Kim}(2019)}]{kwon-kim}%
  \BibitemOpen
  \bibfield  {author} {\bibinfo {author} {\bibfnamefont {H.}~\bibnamefont
  {Kwon}}\ and\ \bibinfo {author} {\bibfnamefont {M.~S.}\ \bibnamefont {Kim}},\
  }\bibfield  {title} {\bibinfo {title} {Fluctuation theorems for a quantum
  channel},\ }\href {https://doi.org/10.1103/PhysRevX.9.031029} {\bibfield
  {journal} {\bibinfo  {journal} {Phys. Rev. X}\ }\textbf {\bibinfo {volume}
  {9}},\ \bibinfo {pages} {031029} (\bibinfo {year} {2019})}\BibitemShut
  {NoStop}%
\bibitem [{\citenamefont {Kwon}\ \emph {et~al.}(2022)\citenamefont {Kwon},
  \citenamefont {Mukherjee},\ and\ \citenamefont {Kim}}]{kwon2022}%
  \BibitemOpen
  \bibfield  {author} {\bibinfo {author} {\bibfnamefont {H.}~\bibnamefont
  {Kwon}}, \bibinfo {author} {\bibfnamefont {R.}~\bibnamefont {Mukherjee}},\
  and\ \bibinfo {author} {\bibfnamefont {M.~S.}\ \bibnamefont {Kim}},\
  }\bibfield  {title} {\bibinfo {title} {Reversing lindblad dynamics via
  continuous petz recovery map},\ }\href
  {https://doi.org/10.1103/PhysRevLett.128.020403} {\bibfield  {journal}
  {\bibinfo  {journal} {Phys. Rev. Lett.}\ }\textbf {\bibinfo {volume} {128}},\
  \bibinfo {pages} {020403} (\bibinfo {year} {2022})}\BibitemShut {NoStop}%
\bibitem [{\citenamefont {Petz}(1988)}]{petz}%
  \BibitemOpen
  \bibfield  {author} {\bibinfo {author} {\bibfnamefont {D.}~\bibnamefont
  {Petz}},\ }\bibfield  {title} {\bibinfo {title} {{Sufficiency of channels
  over von Neumann algebras}},\ }\href {https://doi.org/10.1093/qmath/39.1.97}
  {\bibfield  {journal} {\bibinfo  {journal} {The Quarterly Journal of
  Mathematics}\ }\textbf {\bibinfo {volume} {39}},\ \bibinfo {pages} {97}
  (\bibinfo {year} {1988})}\BibitemShut {NoStop}%
\bibitem [{\citenamefont {Wilde}(2013)}]{wilde_2013}%
  \BibitemOpen
  \bibfield  {author} {\bibinfo {author} {\bibfnamefont {M.~M.}\ \bibnamefont
  {Wilde}},\ }\href {https://doi.org/10.1017/CBO9781139525343} {\emph {\bibinfo
  {title} {Quantum Information Theory}}}\ (\bibinfo  {publisher} {Cambridge
  University Press},\ \bibinfo {year} {2013})\BibitemShut {NoStop}%
\bibitem [{\citenamefont {Leifer}\ and\ \citenamefont
  {Spekkens}(2013)}]{Leifer-Spekkens}%
  \BibitemOpen
  \bibfield  {author} {\bibinfo {author} {\bibfnamefont {M.~S.}\ \bibnamefont
  {Leifer}}\ and\ \bibinfo {author} {\bibfnamefont {R.~W.}\ \bibnamefont
  {Spekkens}},\ }\bibfield  {title} {\bibinfo {title} {Towards a formulation of
  quantum theory as a causally neutral theory of bayesian inference},\ }\href
  {https://doi.org/10.1103/PhysRevA.88.052130} {\bibfield  {journal} {\bibinfo
  {journal} {Phys. Rev. A}\ }\textbf {\bibinfo {volume} {88}},\ \bibinfo
  {pages} {052130} (\bibinfo {year} {2013})}\BibitemShut {NoStop}%
\bibitem [{\citenamefont {Parzygnat}\ and\ \citenamefont
  {Buscemi}(2023)}]{PB22}%
  \BibitemOpen
  \bibfield  {author} {\bibinfo {author} {\bibfnamefont {A.~J.}\ \bibnamefont
  {Parzygnat}}\ and\ \bibinfo {author} {\bibfnamefont {F.}~\bibnamefont
  {Buscemi}},\ }\bibfield  {title} {\bibinfo {title} {Axioms for retrodiction:
  achieving time-reversal symmetry with a prior},\ }\href
  {https://doi.org/10.22331/q-2023-05-23-1013} {\bibfield  {journal} {\bibinfo
  {journal} {Quantum}\ }\textbf {\bibinfo {volume} {7}},\ \bibinfo {pages}
  {1013} (\bibinfo {year} {2023})}\BibitemShut {NoStop}%
\bibitem [{\citenamefont {Parzygnat}\ and\ \citenamefont
  {Fullwood}(2023)}]{PF22}%
  \BibitemOpen
  \bibfield  {author} {\bibinfo {author} {\bibfnamefont {A.~J.}\ \bibnamefont
  {Parzygnat}}\ and\ \bibinfo {author} {\bibfnamefont {J.}~\bibnamefont
  {Fullwood}},\ }\bibfield  {title} {\bibinfo {title} {From time-reversal
  symmetry to quantum bayes' rules},\ }\href
  {https://doi.org/10.1103/PRXQuantum.4.020334} {\bibfield  {journal} {\bibinfo
   {journal} {PRX Quantum}\ }\textbf {\bibinfo {volume} {4}},\ \bibinfo {pages}
  {020334} (\bibinfo {year} {2023})}\BibitemShut {NoStop}%
\bibitem [{\citenamefont {Aw}\ \emph {et~al.}(2023)\citenamefont {Aw},
  \citenamefont {Onggadinata}, \citenamefont {Kaszlikowski},\ and\
  \citenamefont {Scarani}}]{QPRPetzPaper}%
  \BibitemOpen
  \bibfield  {author} {\bibinfo {author} {\bibfnamefont {C.~C.}\ \bibnamefont
  {Aw}}, \bibinfo {author} {\bibfnamefont {K.}~\bibnamefont {Onggadinata}},
  \bibinfo {author} {\bibfnamefont {D.}~\bibnamefont {Kaszlikowski}},\ and\
  \bibinfo {author} {\bibfnamefont {V.}~\bibnamefont {Scarani}},\ }\bibfield
  {title} {\bibinfo {title} {Quantum bayesian inference in quasiprobability
  representations},\ }\href {https://doi.org/10.1103/PRXQuantum.4.020352}
  {\bibfield  {journal} {\bibinfo  {journal} {PRX Quantum}\ }\textbf {\bibinfo
  {volume} {4}},\ \bibinfo {pages} {020352} (\bibinfo {year}
  {2023})}\BibitemShut {NoStop}%
\bibitem [{\citenamefont {Gily\'en}\ \emph {et~al.}(2022)\citenamefont
  {Gily\'en}, \citenamefont {Lloyd}, \citenamefont {Marvian}, \citenamefont
  {Quek},\ and\ \citenamefont {Wilde}}]{quek2020quantum}%
  \BibitemOpen
  \bibfield  {author} {\bibinfo {author} {\bibfnamefont {A.}~\bibnamefont
  {Gily\'en}}, \bibinfo {author} {\bibfnamefont {S.}~\bibnamefont {Lloyd}},
  \bibinfo {author} {\bibfnamefont {I.}~\bibnamefont {Marvian}}, \bibinfo
  {author} {\bibfnamefont {Y.}~\bibnamefont {Quek}},\ and\ \bibinfo {author}
  {\bibfnamefont {M.~M.}\ \bibnamefont {Wilde}},\ }\bibfield  {title} {\bibinfo
  {title} {{Quantum Algorithm for Petz Recovery Channels and Pretty Good
  Measurements}},\ }\href {https://doi.org/10.1103/PhysRevLett.128.220502}
  {\bibfield  {journal} {\bibinfo  {journal} {Phys. Rev. Lett.}\ }\textbf
  {\bibinfo {volume} {128}},\ \bibinfo {pages} {220502} (\bibinfo {year}
  {2022})}\BibitemShut {NoStop}%
\bibitem [{\citenamefont {Di~Biagio}\ \emph {et~al.}(2021)\citenamefont
  {Di~Biagio}, \citenamefont {Don{\`{a}}},\ and\ \citenamefont
  {Rovelli}}]{DiBiagio2021arrowoftimein}%
  \BibitemOpen
  \bibfield  {author} {\bibinfo {author} {\bibfnamefont {A.}~\bibnamefont
  {Di~Biagio}}, \bibinfo {author} {\bibfnamefont {P.}~\bibnamefont
  {Don{\`{a}}}},\ and\ \bibinfo {author} {\bibfnamefont {C.}~\bibnamefont
  {Rovelli}},\ }\bibfield  {title} {\bibinfo {title} {The arrow of time in
  operational formulations of quantum theory},\ }\href
  {https://doi.org/10.22331/q-2021-08-09-520} {\bibfield  {journal} {\bibinfo
  {journal} {{Quantum}}\ }\textbf {\bibinfo {volume} {5}},\ \bibinfo {pages}
  {520} (\bibinfo {year} {2021})}\BibitemShut {NoStop}%
\bibitem [{\citenamefont {Chiribella}\ and\ \citenamefont {Liu}(2022)}]{CL22}%
  \BibitemOpen
  \bibfield  {author} {\bibinfo {author} {\bibfnamefont {G.}~\bibnamefont
  {Chiribella}}\ and\ \bibinfo {author} {\bibfnamefont {Z.}~\bibnamefont
  {Liu}},\ }\bibfield  {title} {\bibinfo {title} {Quantum operations with
  indefinite time direction},\ }\href
  {https://doi.org/10.1038/s42005-022-00967-3} {\bibfield  {journal} {\bibinfo
  {journal} {Communications Physics}\ }\textbf {\bibinfo {volume} {5}},\
  \bibinfo {pages} {190} (\bibinfo {year} {2022})}\BibitemShut {NoStop}%
\bibitem [{\citenamefont {Chiribella}\ \emph {et~al.}(2021)\citenamefont
  {Chiribella}, \citenamefont {Aurell},\ and\ \citenamefont
  {{\.Z}yczkowski}}]{chiribella2021symmetries}%
  \BibitemOpen
  \bibfield  {author} {\bibinfo {author} {\bibfnamefont {G.}~\bibnamefont
  {Chiribella}}, \bibinfo {author} {\bibfnamefont {E.}~\bibnamefont {Aurell}},\
  and\ \bibinfo {author} {\bibfnamefont {K.}~\bibnamefont {{\.Z}yczkowski}},\
  }\bibfield  {title} {\bibinfo {title} {Symmetries of quantum evolutions},\
  }\href@noop {} {\bibfield  {journal} {\bibinfo  {journal} {Physical Review
  Research}\ }\textbf {\bibinfo {volume} {3}},\ \bibinfo {pages} {033028}
  (\bibinfo {year} {2021})}\BibitemShut {NoStop}%
\bibitem [{\citenamefont {Ciliberto}\ \emph {et~al.}(2010)\citenamefont
  {Ciliberto}, \citenamefont {Joubaud},\ and\ \citenamefont
  {Petrosyan}}]{exp-ciliberto2010fluctuations}%
  \BibitemOpen
  \bibfield  {author} {\bibinfo {author} {\bibfnamefont {S.}~\bibnamefont
  {Ciliberto}}, \bibinfo {author} {\bibfnamefont {S.}~\bibnamefont {Joubaud}},\
  and\ \bibinfo {author} {\bibfnamefont {A.}~\bibnamefont {Petrosyan}},\
  }\bibfield  {title} {\bibinfo {title} {Fluctuations in out-of-equilibrium
  systems: from theory to experiment},\ }\href
  {https://doi.org/10.1088/1742-5468/2010/12/P12003} {\bibfield  {journal}
  {\bibinfo  {journal} {Journal of Statistical Mechanics: Theory and
  Experiment}\ }\textbf {\bibinfo {volume} {2010}},\ \bibinfo {pages} {P12003}
  (\bibinfo {year} {2010})}\BibitemShut {NoStop}%
\bibitem [{\citenamefont {Collin}\ \emph {et~al.}(2005)\citenamefont {Collin},
  \citenamefont {Ritort}, \citenamefont {Jarzynski}, \citenamefont {Smith},
  \citenamefont {Tinoco},\ and\ \citenamefont {Bustamante}}]{exp-Collin05}%
  \BibitemOpen
  \bibfield  {author} {\bibinfo {author} {\bibfnamefont {D.}~\bibnamefont
  {Collin}}, \bibinfo {author} {\bibfnamefont {F.}~\bibnamefont {Ritort}},
  \bibinfo {author} {\bibfnamefont {C.}~\bibnamefont {Jarzynski}}, \bibinfo
  {author} {\bibfnamefont {S.~B.}\ \bibnamefont {Smith}}, \bibinfo {author}
  {\bibfnamefont {I.}~\bibnamefont {Tinoco}},\ and\ \bibinfo {author}
  {\bibfnamefont {C.}~\bibnamefont {Bustamante}},\ }\bibfield  {title}
  {\bibinfo {title} {{Verification of the Crooks fluctuation theorem and
  recovery of RNA folding free energies}},\ }\href
  {https://doi.org/10.1038/nature04061} {\bibfield  {journal} {\bibinfo
  {journal} {Nature}\ }\textbf {\bibinfo {volume} {437}},\ \bibinfo {pages}
  {231} (\bibinfo {year} {2005})}\BibitemShut {NoStop}%
\bibitem [{\citenamefont {Hayashi}(2018)}]{exp-hayashi2018application}%
  \BibitemOpen
  \bibfield  {author} {\bibinfo {author} {\bibfnamefont {K.}~\bibnamefont
  {Hayashi}},\ }\bibfield  {title} {\bibinfo {title} {Application of the
  fluctuation theorem to motor proteins: from f1-atpase to axonal cargo
  transport by kinesin and dynein},\ }\href
  {https://doi.org/10.1007/s12551-018-0440-5} {\bibfield  {journal} {\bibinfo
  {journal} {Biophysical reviews}\ }\textbf {\bibinfo {volume} {10}},\ \bibinfo
  {pages} {1311} (\bibinfo {year} {2018})}\BibitemShut {NoStop}%
\bibitem [{\citenamefont {Hoang}\ \emph {et~al.}(2018)\citenamefont {Hoang},
  \citenamefont {Pan}, \citenamefont {Ahn}, \citenamefont {Bang}, \citenamefont
  {Quan},\ and\ \citenamefont {Li}}]{exp-hoang2018experimental}%
  \BibitemOpen
  \bibfield  {author} {\bibinfo {author} {\bibfnamefont {T.~M.}\ \bibnamefont
  {Hoang}}, \bibinfo {author} {\bibfnamefont {R.}~\bibnamefont {Pan}}, \bibinfo
  {author} {\bibfnamefont {J.}~\bibnamefont {Ahn}}, \bibinfo {author}
  {\bibfnamefont {J.}~\bibnamefont {Bang}}, \bibinfo {author} {\bibfnamefont
  {H.}~\bibnamefont {Quan}},\ and\ \bibinfo {author} {\bibfnamefont
  {T.}~\bibnamefont {Li}},\ }\bibfield  {title} {\bibinfo {title} {Experimental
  test of the differential fluctuation theorem and a generalized jarzynski
  equality for arbitrary initial states},\ }\href
  {https://doi.org/10.1103/PhysRevLett.120.080602} {\bibfield  {journal}
  {\bibinfo  {journal} {Physical review letters}\ }\textbf {\bibinfo {volume}
  {120}},\ \bibinfo {pages} {080602} (\bibinfo {year} {2018})}\BibitemShut
  {NoStop}%
\bibitem [{\citenamefont {Carberry}\ \emph {et~al.}(2004)\citenamefont
  {Carberry}, \citenamefont {Reid}, \citenamefont {Wang}, \citenamefont
  {Sevick}, \citenamefont {Searles},\ and\ \citenamefont
  {Evans}}]{exp-carberry2004fluctuations}%
  \BibitemOpen
  \bibfield  {author} {\bibinfo {author} {\bibfnamefont {D.}~\bibnamefont
  {Carberry}}, \bibinfo {author} {\bibfnamefont {J.~C.}\ \bibnamefont {Reid}},
  \bibinfo {author} {\bibfnamefont {G.}~\bibnamefont {Wang}}, \bibinfo {author}
  {\bibfnamefont {E.~M.}\ \bibnamefont {Sevick}}, \bibinfo {author}
  {\bibfnamefont {D.~J.}\ \bibnamefont {Searles}},\ and\ \bibinfo {author}
  {\bibfnamefont {D.~J.}\ \bibnamefont {Evans}},\ }\bibfield  {title} {\bibinfo
  {title} {Fluctuations and irreversibility: an experimental demonstration of a
  second-law-like theorem using a colloidal particle held in an optical trap},\
  }\href {https://doi.org/10.1103/PhysRevLett.92.140601} {\bibfield  {journal}
  {\bibinfo  {journal} {Physical review letters}\ }\textbf {\bibinfo {volume}
  {92}},\ \bibinfo {pages} {140601} (\bibinfo {year} {2004})}\BibitemShut
  {NoStop}%
\bibitem [{\citenamefont {Horodecki}\ and\ \citenamefont
  {Oppenheim}(2013)}]{horo-oppen-thermal}%
  \BibitemOpen
  \bibfield  {author} {\bibinfo {author} {\bibfnamefont {M.}~\bibnamefont
  {Horodecki}}\ and\ \bibinfo {author} {\bibfnamefont {J.}~\bibnamefont
  {Oppenheim}},\ }\bibfield  {title} {\bibinfo {title} {Fundamental limitations
  for quantum and nanoscale thermodynamics},\ }\href
  {https://doi.org/10.1038/ncomms3059} {\bibfield  {journal} {\bibinfo
  {journal} {Nature Communications}\ }\textbf {\bibinfo {volume} {4}},\
  \bibinfo {pages} {2059} (\bibinfo {year} {2013})}\BibitemShut {NoStop}%
\bibitem [{\citenamefont {Santos}\ \emph {et~al.}(2020)\citenamefont {Santos},
  \citenamefont {Timpanaro},\ and\ \citenamefont {Landi}}]{santos_landi}%
  \BibitemOpen
  \bibfield  {author} {\bibinfo {author} {\bibfnamefont {J.}~\bibnamefont
  {Santos}}, \bibinfo {author} {\bibfnamefont {A.}~\bibnamefont {Timpanaro}},\
  and\ \bibinfo {author} {\bibfnamefont {G.}~\bibnamefont {Landi}},\ }\bibfield
   {title} {\bibinfo {title} {Joint fluctuation theorems for sequential heat
  exchange},\ }\bibfield  {journal} {\bibinfo  {journal} {Entropy}\ }\textbf
  {\bibinfo {volume} {22}},\ \href {https://doi.org/10.3390/e22070763}
  {10.3390/e22070763} (\bibinfo {year} {2020})\BibitemShut {NoStop}%
\bibitem [{\citenamefont {Landi}\ and\ \citenamefont
  {Paternostro}(2021)}]{landi-pater-21-review}%
  \BibitemOpen
  \bibfield  {author} {\bibinfo {author} {\bibfnamefont {G.~T.}\ \bibnamefont
  {Landi}}\ and\ \bibinfo {author} {\bibfnamefont {M.}~\bibnamefont
  {Paternostro}},\ }\bibfield  {title} {\bibinfo {title} {Irreversible entropy
  production: From classical to quantum},\ }\href
  {https://doi.org/10.1103/RevModPhys.93.035008} {\bibfield  {journal}
  {\bibinfo  {journal} {Reviews of Modern Physics}\ }\textbf {\bibinfo {volume}
  {93}},\ \bibinfo {pages} {035008} (\bibinfo {year} {2021})}\BibitemShut
  {NoStop}%
\bibitem [{\citenamefont {Horodecki}\ \emph {et~al.}(2003)\citenamefont
  {Horodecki}, \citenamefont {Horodecki},\ and\ \citenamefont
  {Oppenheim}}]{horo2003-thermal-first}%
  \BibitemOpen
  \bibfield  {author} {\bibinfo {author} {\bibfnamefont {M.}~\bibnamefont
  {Horodecki}}, \bibinfo {author} {\bibfnamefont {P.}~\bibnamefont
  {Horodecki}},\ and\ \bibinfo {author} {\bibfnamefont {J.}~\bibnamefont
  {Oppenheim}},\ }\bibfield  {title} {\bibinfo {title} {Reversible
  transformations from pure to mixed states and the unique measure of
  information},\ }\href {https://doi.org/10.1103/PhysRevA.67.062104} {\bibfield
   {journal} {\bibinfo  {journal} {Physical Review A}\ }\textbf {\bibinfo
  {volume} {67}},\ \bibinfo {pages} {062104} (\bibinfo {year}
  {2003})}\BibitemShut {NoStop}%
\bibitem [{\citenamefont {Brandao}\ \emph {et~al.}(2013)\citenamefont
  {Brandao}, \citenamefont {Horodecki}, \citenamefont {Oppenheim},
  \citenamefont {Renes},\ and\ \citenamefont {Spekkens}}]{brandao2013-thermal}%
  \BibitemOpen
  \bibfield  {author} {\bibinfo {author} {\bibfnamefont {F.~G.}\ \bibnamefont
  {Brandao}}, \bibinfo {author} {\bibfnamefont {M.}~\bibnamefont {Horodecki}},
  \bibinfo {author} {\bibfnamefont {J.}~\bibnamefont {Oppenheim}}, \bibinfo
  {author} {\bibfnamefont {J.~M.}\ \bibnamefont {Renes}},\ and\ \bibinfo
  {author} {\bibfnamefont {R.~W.}\ \bibnamefont {Spekkens}},\ }\bibfield
  {title} {\bibinfo {title} {Resource theory of quantum states out of thermal
  equilibrium},\ }\href {https://doi.org/10.1103/PhysRevLett.111.250404}
  {\bibfield  {journal} {\bibinfo  {journal} {Physical review letters}\
  }\textbf {\bibinfo {volume} {111}},\ \bibinfo {pages} {250404} (\bibinfo
  {year} {2013})}\BibitemShut {NoStop}%
\bibitem [{\citenamefont {Lostaglio}\ \emph {et~al.}(2018)\citenamefont
  {Lostaglio}, \citenamefont {Alhambra},\ and\ \citenamefont
  {Perry}}]{lostaglio2018-ETOs}%
  \BibitemOpen
  \bibfield  {author} {\bibinfo {author} {\bibfnamefont {M.}~\bibnamefont
  {Lostaglio}}, \bibinfo {author} {\bibfnamefont {{\'A}.~M.}\ \bibnamefont
  {Alhambra}},\ and\ \bibinfo {author} {\bibfnamefont {C.}~\bibnamefont
  {Perry}},\ }\bibfield  {title} {\bibinfo {title} {Elementary thermal
  operations},\ }\href {https://doi.org/10.22331/q-2018-02-08-52} {\bibfield
  {journal} {\bibinfo  {journal} {Quantum}\ }\textbf {\bibinfo {volume} {2}},\
  \bibinfo {pages} {52} (\bibinfo {year} {2018})}\BibitemShut {NoStop}%
\bibitem [{\citenamefont {Faist}\ \emph {et~al.}(2015)\citenamefont {Faist},
  \citenamefont {Oppenheim},\ and\ \citenamefont
  {Renner}}]{faist-oppo-renner-2015-thermal}%
  \BibitemOpen
  \bibfield  {author} {\bibinfo {author} {\bibfnamefont {P.}~\bibnamefont
  {Faist}}, \bibinfo {author} {\bibfnamefont {J.}~\bibnamefont {Oppenheim}},\
  and\ \bibinfo {author} {\bibfnamefont {R.}~\bibnamefont {Renner}},\
  }\bibfield  {title} {\bibinfo {title} {Gibbs-preserving maps outperform
  thermal operations in the quantum regime},\ }\href
  {https://doi.org/10.1088/1367-2630/17/4/043003} {\bibfield  {journal}
  {\bibinfo  {journal} {New Journal of Physics}\ }\textbf {\bibinfo {volume}
  {17}},\ \bibinfo {pages} {043003} (\bibinfo {year} {2015})}\BibitemShut
  {NoStop}%
\bibitem [{\citenamefont {Hu}\ and\ \citenamefont
  {Ding}(2019)}]{hu2019thermal}%
  \BibitemOpen
  \bibfield  {author} {\bibinfo {author} {\bibfnamefont {X.}~\bibnamefont
  {Hu}}\ and\ \bibinfo {author} {\bibfnamefont {F.}~\bibnamefont {Ding}},\
  }\bibfield  {title} {\bibinfo {title} {Thermal operations involving a
  single-mode bosonic bath},\ }\href
  {https://doi.org/10.1103/PhysRevA.99.012104} {\bibfield  {journal} {\bibinfo
  {journal} {Phys. Rev. A}\ }\textbf {\bibinfo {volume} {99}},\ \bibinfo
  {pages} {012104} (\bibinfo {year} {2019})}\BibitemShut {NoStop}%
\bibitem [{\citenamefont {Watanabe}(1955)}]{watanabe55}%
  \BibitemOpen
  \bibfield  {author} {\bibinfo {author} {\bibfnamefont {S.}~\bibnamefont
  {Watanabe}},\ }\bibfield  {title} {\bibinfo {title} {Symmetry of physical
  laws. part iii. prediction and retrodiction},\ }\href
  {https://doi.org/10.1103/RevModPhys.27.179} {\bibfield  {journal} {\bibinfo
  {journal} {Rev. Mod. Phys.}\ }\textbf {\bibinfo {volume} {27}},\ \bibinfo
  {pages} {179} (\bibinfo {year} {1955})}\BibitemShut {NoStop}%
\bibitem [{\citenamefont {Watanabe}(1965)}]{watanabe65}%
  \BibitemOpen
  \bibfield  {author} {\bibinfo {author} {\bibfnamefont {S.}~\bibnamefont
  {Watanabe}},\ }\bibfield  {title} {\bibinfo {title} {Conditional
  probabilities in physics},\ }\href
  {https://doi.org/https://doi.org/10.1143/PTPS.E65.135} {\bibfield  {journal}
  {\bibinfo  {journal} {Progr. Theor. Phys. Suppl.}\ }\textbf {\bibinfo
  {volume} {E65}},\ \bibinfo {pages} {135} (\bibinfo {year}
  {1965})}\BibitemShut {NoStop}%
\bibitem [{\citenamefont {Parzygnat}\ and\ \citenamefont
  {Russo}(2022)}]{parzygnat2022noncommutingbayes}%
  \BibitemOpen
  \bibfield  {author} {\bibinfo {author} {\bibfnamefont {A.~J.}\ \bibnamefont
  {Parzygnat}}\ and\ \bibinfo {author} {\bibfnamefont {B.~P.}\ \bibnamefont
  {Russo}},\ }\bibfield  {title} {\bibinfo {title} {{A non-commutative Bayes'
  theorem}},\ }\href {https://doi.org/10.1016/j.laa.2022.02.030} {\bibfield
  {journal} {\bibinfo  {journal} {Linear Algebra and its Applications}\
  }\textbf {\bibinfo {volume} {644}},\ \bibinfo {pages} {28} (\bibinfo {year}
  {2022})}\BibitemShut {NoStop}%
\bibitem [{\citenamefont {Surace}\ and\ \citenamefont
  {Scandi}(2023)}]{surace2023state-retrieval-beyond-bayes}%
  \BibitemOpen
  \bibfield  {author} {\bibinfo {author} {\bibfnamefont {J.}~\bibnamefont
  {Surace}}\ and\ \bibinfo {author} {\bibfnamefont {M.}~\bibnamefont
  {Scandi}},\ }\bibfield  {title} {\bibinfo {title} {State retrieval beyond
  bayes' retrodiction},\ }\href {https://doi.org/10.22331/q-2023-04-27-990}
  {\bibfield  {journal} {\bibinfo  {journal} {Quantum}\ }\textbf {\bibinfo
  {volume} {7}},\ \bibinfo {pages} {990} (\bibinfo {year} {2023})}\BibitemShut
  {NoStop}%
\bibitem [{\citenamefont {Petz}(1986)}]{petz1}%
  \BibitemOpen
  \bibfield  {author} {\bibinfo {author} {\bibfnamefont {D.}~\bibnamefont
  {Petz}},\ }\bibfield  {title} {\bibinfo {title} {Sufficient subalgebras and
  the relative entropy of states of a von neumann algebra},\ }\href
  {https://doi.org/10.1007/BF01212345} {\bibfield  {journal} {\bibinfo
  {journal} {Comm. Math. Phys.}\ }\textbf {\bibinfo {volume} {105}},\ \bibinfo
  {pages} {123} (\bibinfo {year} {1986})}\BibitemShut {NoStop}%
\bibitem [{\citenamefont {Wilde}(2015)}]{wilde-recov}%
  \BibitemOpen
  \bibfield  {author} {\bibinfo {author} {\bibfnamefont {M.}~\bibnamefont
  {Wilde}},\ }\bibfield  {title} {\bibinfo {title} {Recoverability in quantum
  information theory},\ }\href {https://doi.org/10.1098/rspa.2015.0338}
  {\bibfield  {journal} {\bibinfo  {journal} {Proceedings of the Royal Society
  A}\ }\textbf {\bibinfo {volume} {471}},\ \bibinfo {pages} {20150338}
  (\bibinfo {year} {2015})}\BibitemShut {NoStop}%
\bibitem [{\citenamefont {Pechukas}(1994)}]{pechukas}%
  \BibitemOpen
  \bibfield  {author} {\bibinfo {author} {\bibfnamefont {P.}~\bibnamefont
  {Pechukas}},\ }\bibfield  {title} {\bibinfo {title} {Reduced dynamics need
  not be completely positive},\ }\href
  {https://doi.org/10.1103/PhysRevLett.73.1060} {\bibfield  {journal} {\bibinfo
   {journal} {Phys. Rev. Lett.}\ }\textbf {\bibinfo {volume} {73}},\ \bibinfo
  {pages} {1060} (\bibinfo {year} {1994})}\BibitemShut {NoStop}%
\bibitem [{\citenamefont {Alicki}(1995)}]{alicki-comment}%
  \BibitemOpen
  \bibfield  {author} {\bibinfo {author} {\bibfnamefont {R.}~\bibnamefont
  {Alicki}},\ }\bibfield  {title} {\bibinfo {title} {Comment on ``{R}educed
  dynamics need not be completely positive''},\ }\href
  {https://doi.org/10.1103/PhysRevLett.75.3020} {\bibfield  {journal} {\bibinfo
   {journal} {Phys. Rev. Lett.}\ }\textbf {\bibinfo {volume} {75}},\ \bibinfo
  {pages} {3020} (\bibinfo {year} {1995})}\BibitemShut {NoStop}%
\bibitem [{\citenamefont {Buscemi}(2014)}]{buscemi-NCPTP}%
  \BibitemOpen
  \bibfield  {author} {\bibinfo {author} {\bibfnamefont {F.}~\bibnamefont
  {Buscemi}},\ }\bibfield  {title} {\bibinfo {title} {Complete positivity,
  markovianity, and the quantum data-processing inequality, in the presence of
  initial system-environment correlations},\ }\href
  {https://doi.org/10.1103/PhysRevLett.113.140502} {\bibfield  {journal}
  {\bibinfo  {journal} {Phys. Rev. Lett.}\ }\textbf {\bibinfo {volume} {113}},\
  \bibinfo {pages} {140502} (\bibinfo {year} {2014})}\BibitemShut {NoStop}%
\bibitem [{\citenamefont {Nielsen}\ and\ \citenamefont
  {Chuang}(2002)}]{nielsen-chuang}%
  \BibitemOpen
  \bibfield  {author} {\bibinfo {author} {\bibfnamefont {M.~A.}\ \bibnamefont
  {Nielsen}}\ and\ \bibinfo {author} {\bibfnamefont {I.}~\bibnamefont
  {Chuang}},\ }\href@noop {} {\bibinfo {title} {Quantum computation and quantum
  information}} (\bibinfo {year} {2002})\BibitemShut {NoStop}%
\bibitem [{\citenamefont {Lautenbacher}\ \emph {et~al.}(2022)\citenamefont
  {Lautenbacher}, \citenamefont {de~Melo},\ and\ \citenamefont
  {Bernardes}}]{numerical-retrodiction}%
  \BibitemOpen
  \bibfield  {author} {\bibinfo {author} {\bibfnamefont {L.}~\bibnamefont
  {Lautenbacher}}, \bibinfo {author} {\bibfnamefont {F.}~\bibnamefont
  {de~Melo}},\ and\ \bibinfo {author} {\bibfnamefont {N.~K.}\ \bibnamefont
  {Bernardes}},\ }\bibfield  {title} {\bibinfo {title} {Approximating
  invertible maps by recovery channels: Optimality and an application to
  non-markovian dynamics},\ }\href
  {https://doi.org/10.1103/PhysRevA.105.042421} {\bibfield  {journal} {\bibinfo
   {journal} {Phys. Rev. A}\ }\textbf {\bibinfo {volume} {105}},\ \bibinfo
  {pages} {042421} (\bibinfo {year} {2022})}\BibitemShut {NoStop}%
\bibitem [{\citenamefont {\AA{}berg}(2018)}]{aberg-quantum-fluct}%
  \BibitemOpen
  \bibfield  {author} {\bibinfo {author} {\bibfnamefont {J.}~\bibnamefont
  {\AA{}berg}},\ }\bibfield  {title} {\bibinfo {title} {Fully quantum
  fluctuation theorems},\ }\href {https://doi.org/10.1103/PhysRevX.8.011019}
  {\bibfield  {journal} {\bibinfo  {journal} {Phys. Rev. X}\ }\textbf {\bibinfo
  {volume} {8}},\ \bibinfo {pages} {011019} (\bibinfo {year}
  {2018})}\BibitemShut {NoStop}%
\bibitem [{\citenamefont {Pucha\l{}a}\ \emph {et~al.}(2015)\citenamefont
  {Pucha\l{}a}, \citenamefont {Rudnicki}, \citenamefont {Chabuda},
  \citenamefont {Paraniak},\ and\ \citenamefont {\ifmmode~\dot{Z}\else
  \.{Z}\fi{}yczkowski}}]{karol2015certainty}%
  \BibitemOpen
  \bibfield  {author} {\bibinfo {author} {\bibfnamefont {Z.}~\bibnamefont
  {Pucha\l{}a}}, \bibinfo {author} {\bibfnamefont {L.}~\bibnamefont
  {Rudnicki}}, \bibinfo {author} {\bibfnamefont {K.}~\bibnamefont {Chabuda}},
  \bibinfo {author} {\bibfnamefont {M.}~\bibnamefont {Paraniak}},\ and\
  \bibinfo {author} {\bibfnamefont {K.}~\bibnamefont {\ifmmode~\dot{Z}\else
  \.{Z}\fi{}yczkowski}},\ }\bibfield  {title} {\bibinfo {title} {Certainty
  relations, mutual entanglement, and nondisplaceable manifolds},\ }\href
  {https://doi.org/10.1103/PhysRevA.92.032109} {\bibfield  {journal} {\bibinfo
  {journal} {Phys. Rev. A}\ }\textbf {\bibinfo {volume} {92}},\ \bibinfo
  {pages} {032109} (\bibinfo {year} {2015})}\BibitemShut {NoStop}%
\bibitem [{\citenamefont {Brahmachari}\ \emph {et~al.}(2022)\citenamefont
  {Brahmachari}, \citenamefont {Rajmohan}, \citenamefont {Rather},\ and\
  \citenamefont {Lakshminarayan}}]{brahmachari2022dual}%
  \BibitemOpen
  \bibfield  {author} {\bibinfo {author} {\bibfnamefont {S.}~\bibnamefont
  {Brahmachari}}, \bibinfo {author} {\bibfnamefont {R.~N.}\ \bibnamefont
  {Rajmohan}}, \bibinfo {author} {\bibfnamefont {S.~A.}\ \bibnamefont
  {Rather}},\ and\ \bibinfo {author} {\bibfnamefont {A.}~\bibnamefont
  {Lakshminarayan}},\ }\bibfield  {title} {\bibinfo {title} {Dual unitaries as
  maximizers of the distance to local product gates},\ }\href
  {https://doi.org/10.48550/arXiv.2210.13307} {\bibfield  {journal} {\bibinfo
  {journal} {arXiv preprint arXiv:2210.13307}\ } (\bibinfo {year}
  {2022})}\BibitemShut {NoStop}%
\bibitem [{\citenamefont {Chen}\ \emph {et~al.}(2008)\citenamefont {Chen},
  \citenamefont {Duan}, \citenamefont {Ji}, \citenamefont {Ying},\ and\
  \citenamefont {Yu}}]{UniversalEntangler}%
  \BibitemOpen
  \bibfield  {author} {\bibinfo {author} {\bibfnamefont {J.}~\bibnamefont
  {Chen}}, \bibinfo {author} {\bibfnamefont {R.}~\bibnamefont {Duan}}, \bibinfo
  {author} {\bibfnamefont {Z.}~\bibnamefont {Ji}}, \bibinfo {author}
  {\bibfnamefont {M.}~\bibnamefont {Ying}},\ and\ \bibinfo {author}
  {\bibfnamefont {J.}~\bibnamefont {Yu}},\ }\bibfield  {title} {\bibinfo
  {title} {{Existence of universal entangler}},\ }\href
  {https://doi.org/10.1063/1.2829895} {\bibfield  {journal} {\bibinfo
  {journal} {Journal of Mathematical Physics}\ }\textbf {\bibinfo {volume}
  {49}},\ \bibinfo {pages} {012103} (\bibinfo {year} {2008})}\BibitemShut
  {NoStop}%
\bibitem [{\citenamefont {Plenio}(2005)}]{log-negativity}%
  \BibitemOpen
  \bibfield  {author} {\bibinfo {author} {\bibfnamefont {M.~B.}\ \bibnamefont
  {Plenio}},\ }\bibfield  {title} {\bibinfo {title} {Logarithmic negativity: A
  full entanglement monotone that is not convex},\ }\href
  {https://doi.org/10.1103/PhysRevLett.95.090503} {\bibfield  {journal}
  {\bibinfo  {journal} {Phys. Rev. Lett.}\ }\textbf {\bibinfo {volume} {95}},\
  \bibinfo {pages} {090503} (\bibinfo {year} {2005})}\BibitemShut {NoStop}%
\bibitem [{\citenamefont {Zhang}\ \emph {et~al.}(2003)\citenamefont {Zhang},
  \citenamefont {Vala}, \citenamefont {Sastry},\ and\ \citenamefont
  {Whaley}}]{two-qubit-decomposition}%
  \BibitemOpen
  \bibfield  {author} {\bibinfo {author} {\bibfnamefont {J.}~\bibnamefont
  {Zhang}}, \bibinfo {author} {\bibfnamefont {J.}~\bibnamefont {Vala}},
  \bibinfo {author} {\bibfnamefont {S.}~\bibnamefont {Sastry}},\ and\ \bibinfo
  {author} {\bibfnamefont {K.~B.}\ \bibnamefont {Whaley}},\ }\bibfield  {title}
  {\bibinfo {title} {Geometric theory of nonlocal two-qubit operations},\
  }\href {https://doi.org/10.1103/PhysRevA.67.042313} {\bibfield  {journal}
  {\bibinfo  {journal} {Phys. Rev. A}\ }\textbf {\bibinfo {volume} {67}},\
  \bibinfo {pages} {042313} (\bibinfo {year} {2003})}\BibitemShut {NoStop}%
\bibitem [{\citenamefont {Korzekwa}(2016)}]{thermodynamics-thesis}%
  \BibitemOpen
  \bibfield  {author} {\bibinfo {author} {\bibfnamefont {K.}~\bibnamefont
  {Korzekwa}},\ }\emph {\bibinfo {title} {Coherence, thermodynamics and
  uncertainty relations}},\ \href {https://doi.org/10.25560/43343} {Ph.D.
  thesis},\ \bibinfo  {school} {Imperial College London} (\bibinfo {year}
  {2016})\BibitemShut {NoStop}%
\bibitem [{\citenamefont {Lostaglio}(2019)}]{thermodynamics-review}%
  \BibitemOpen
  \bibfield  {author} {\bibinfo {author} {\bibfnamefont {M.}~\bibnamefont
  {Lostaglio}},\ }\bibfield  {title} {\bibinfo {title} {An introductory review
  of the resource theory approach to thermodynamics},\ }\href
  {https://doi.org/10.1088/1361-6633/ab46e5} {\bibfield  {journal} {\bibinfo
  {journal} {Reports on Progress in Physics}\ }\textbf {\bibinfo {volume}
  {82}},\ \bibinfo {pages} {114001} (\bibinfo {year} {2019})}\BibitemShut
  {NoStop}%
\bibitem [{\citenamefont {vom Ende}(2022)}]{bath-given-system}%
  \BibitemOpen
  \bibfield  {author} {\bibinfo {author} {\bibfnamefont {F.}~\bibnamefont {vom
  Ende}},\ }\bibfield  {title} {\bibinfo {title} {{Which bath Hamiltonians
  matter for thermal operations?}},\ }\href {https://doi.org/10.1063/5.0117534}
  {\bibfield  {journal} {\bibinfo  {journal} {Journal of Mathematical Physics}\
  }\textbf {\bibinfo {volume} {63}},\ \bibinfo {pages} {112202} (\bibinfo
  {year} {2022})}\BibitemShut {NoStop}%
\bibitem [{\citenamefont {Korzekwa}(2017)}]{future-thermal-cone}%
  \BibitemOpen
  \bibfield  {author} {\bibinfo {author} {\bibfnamefont {K.}~\bibnamefont
  {Korzekwa}},\ }\bibfield  {title} {\bibinfo {title} {Structure of the
  thermodynamic arrow of time in classical and quantum theories},\ }\href
  {https://doi.org/10.1103/PhysRevA.95.052318} {\bibfield  {journal} {\bibinfo
  {journal} {Phys. Rev. A}\ }\textbf {\bibinfo {volume} {95}},\ \bibinfo
  {pages} {052318} (\bibinfo {year} {2017})}\BibitemShut {NoStop}%
\bibitem [{\citenamefont {Stinespring}(1955)}]{Stinespring-dilation}%
  \BibitemOpen
  \bibfield  {author} {\bibinfo {author} {\bibfnamefont {W.~F.}\ \bibnamefont
  {Stinespring}},\ }\bibfield  {title} {\bibinfo {title} {{Positive Functions
  on C*-Algebras}},\ }\href {http://www.jstor.org/stable/2032342} {\bibfield
  {journal} {\bibinfo  {journal} {Proceedings of the American Mathematical
  Society}\ }\textbf {\bibinfo {volume} {6}},\ \bibinfo {pages} {211} (\bibinfo
  {year} {1955})}\BibitemShut {NoStop}%
\bibitem [{\citenamefont {Khatri}\ and\ \citenamefont
  {Wilde}(2020)}]{quantum-information-notes}%
  \BibitemOpen
  \bibfield  {author} {\bibinfo {author} {\bibfnamefont {S.}~\bibnamefont
  {Khatri}}\ and\ \bibinfo {author} {\bibfnamefont {M.~M.}\ \bibnamefont
  {Wilde}},\ }\href {https://doi.org/10.48550/arXiv.2011.04672} {\bibinfo
  {title} {{Principles of Quantum Communication Theory: A Modern Approach}}}
  (\bibinfo {year} {2020}),\ \Eprint {https://arxiv.org/abs/2011.04672}
  {arXiv:2011.04672 [quant-ph]} \BibitemShut {NoStop}%
\end{thebibliography}%

\appendix
\section{Thermal Operations \& Gibbs States} \label{GTherm}
In \cite{AWWW18}, we have it that $U$ commutes with a global Hamiltonian $H = H_A \otimes \one_B + \one_A \otimes H_B$. We simply that under this definition of $H$, the Gibbs state of $H$ is none other than the product of the Gibbs states of the local Hamiltonian:
\begin{align}
    \frac{e^{\kappa H}}{Z_H} &= \frac{1}{Z_A Z_B}e^{\kappa (H_A \otimes \one_B + \one_A \otimes H_B)} \nonumber\\
    &= \frac{1}{Z_A Z_B} \sum_{nm} e^{\kappa a_n} \ketbra{a_n b_m}{a_n b_m} \sum_{ij} e^{\kappa b_j} \ketbra{a_i b_j}{a_i b_j} \nonumber\\ 
    &= \frac{1}{Z_A Z_B} \sum_{nmij} e^{\kappa a_n} e^{\kappa b_j} \delta_{in} \delta_{jm} \ketbra{a_n b_m}{a_i b_j} \\ 
    &= \frac{1}{Z_A}  \sum_{n} e^{\kappa a_n} \ketbra{a_n}{a_n} \otimes \frac{1}{Z_B} \sum_{j} e^{\kappa b_j} \ketbra{b_j} \nonumber\\
    &= \therm{H_A} \otimes \therm{H_B}\nonumber
\end{align}
Thus, if $[U,H] =0$ then $[U,\therm{H_A} \otimes \therm{H_B}]=0$. Clearly then this means $(U, \therm{H_A}, \therm{H_B})$ is a product-preserved tuple. For a generalized thermal scenarios, we have that $U (H_A + H_B) U^\dag = H_A' + H_B'$. Invoking \eqref{uonf}, it holds that for generalized thermal maps, $U [\therm{H_A} \otimes \therm{H_B}] U^\dag = \therm{H_A'} \otimes \therm{H_B'}$. 

\section{Simple Examples of Dilations \& Priors That Are Not Tabletop Time-reversible}
Here we include simple classical dilations and priors for which the retrodiction map does not fulfill tabletop time-reversibility \eqref{eqfriendly}. We begin first with a classical example. For a dilation
\begin{equation}
    \mathbf{\Phi} = \pmqty{0 & 0 & 0 & 1 \\
    1 & 0 & 0 & 0 \\
    0 & 1 & 0 & 0 \\
    0 & 0 & 1 & 0
    }, 
    \quad \boldsymbol{\cbath} = \pmqty{\beta_0 \\ 1- \beta_0},
\end{equation}
the resultant bit-channel is given by
\begin{equation}
     \boldsymbol{\cchn} = \pmqty{\beta_0 & 1- \beta_0 \\ 1- \beta_0 & \beta_0}. \vspace{1em}
\end{equation} 
The retrodiction channel $\rvp = \ctrace_B \circ \Phi^{-1} \circ \hat{\ctrace}_{B, \Phi[\crf \otimes \cbath]}$, with an arbitrary reference $\boldsymbol{\crf}^\tpose = (\crf_0 \; \; 1-\crf_0)$, is thus
\begin{equation}
\begin{aligned}
\boldsymbol{\rvp} &= \left(
\begin{array}{cc}
 \frac{\crf_0 \beta_0}{(\crf_0-1) (\beta_0-1)+\crf_0 \beta_0} & \frac{\crf_0-\crf_0 \beta_0}{-2 \crf_0 \beta_0+\crf_0+\beta_0} \\
 \frac{(\crf_0-1) (\beta_0-1)}{\crf_0 (2 \beta_0-1)-\beta_0+1} & \frac{\beta_0-\crf_0 \beta_0}{-2 \crf_0 \beta_0+\crf_0+\beta_0} \\
\end{array} \right)
\end{aligned}
\end{equation}
Now, every tabletop-reversible map, $\rvp^{\text{TR}} = \ctrace_B \circ \Phi^{-1} \circ \hat{\ctrace}_{B, \square \otimes \eta}$ where $\boldsymbol{\eta}^\tpose = (\eta_0 \; \; 1-\eta_0)$, for this dilation gives
\begin{equation}
    \boldsymbol{\hat{\cchn}}^{\text{TR}}_{\boldsymbol{\crf}} = \left(
\begin{array}{cc}
 1-\eta_0 & \eta_0 \\
 \eta_0 & 1-\eta_0 \\
\end{array}
\right)
\end{equation}
This means that unless $\crf_0 = 1/2$ and $\eta_0= 1-\beta_0$, or $\beta_0 =0$ and $\eta_0=1$, or $\beta_0 =1$ and $\eta_0=0$, $\rvp \neq \rvp^{\text{TR}}$ always. For instance, if $\crf_0 = \beta_0 = 1/4$, then the two maps will not be equal. This simple example highlights how common non-tabletop reversible tuples (of dilations and priors) are, and that these occur easily even when confined to classical correlations, formed by a two-bit channel.

Meanwhile, for a quantum channel $\chn$ with dilation $U = \cos\theta\one + \sin\theta\operatorname{SWAP}$, ancilla $\beta$, and prior $\alpha\neq\beta$, we have used semidefinite programming to verify that $\hat{\chn}_{\alpha}[\bullet] \neq \TrB\{U^\dag (\bullet\otimes\beta') U\}$ for all $\beta'$ when $\theta \bmod (\pi/2) \neq 0$.

\section{Proofs Regarding Dilation \& Retrodiction}\label{app:proofs}
\subsection{Composability of Bayes' Rule \& The Petz Map}\label{app:proofcomposability}
Here we provide some simple proofs for \eqref{composability_c} and \eqref{composability_q}. For any concatenation $\cchn = \cchn_2 \circ \cchn_1$, applying \eqref{bayesmat} on $\boldsymbol{\cchn}$ gives: 
\begin{equation}
    \begin{aligned}
    \cret{\cchn_2 \circ \cchn_1,\crf}  &= \mathbf{D}_\crf (\boldsymbol{\varphi}_2\boldsymbol{\varphi}_1)^\tpose \mathbf{D}_{\cchn_2 \circ \cchn_1 [\crf]}^{-1} \\
                            &= \mathbf{D}_\crf \,\boldsymbol{\varphi}_1^\tpose\, \boldsymbol{\varphi}_2^\tpose\, \mathbf{D}_{\cchn_2 \circ \cchn_1 [\crf]}^{-1} \\
                            &= \mathbf{D}_\crf \,\boldsymbol{\varphi}_1^\tpose\, \mathbf{D}_{\cchn_1 [\crf]}^{-1} \mathbf{D}_{\cchn_1 [\crf]} \boldsymbol{\varphi}_2^\tpose \mathbf{D}_{\cchn_2 \circ \cchn_1 [\crf]}^{-1} \\ 
                            &= \boldsymbol{\hat{\varphi}}_{1,\crf}\boldsymbol{\hat{\varphi}}_{2,\cchn_1[\crf]},
    \end{aligned}
\end{equation}
with slight abuse of notation: on the right-hand side we have matrices instead of maps.
We may write this more generally: 
\begin{equation} 
 \begin{aligned}
    \mathcal{R}_{c}[\varphi_L\circ \dots \circ \varphi_1,\crf] 
    &= \cret{\varphi_1, \crf} \circ \cret{\varphi_2, \cchn_1[\crf]} \circ \cdots \\
        &\qquad \cdots\circ\cret{\varphi_L, \cchn_{L-1} \circ \dots\circ\cchn_1 [\crf] }\,.
\end{aligned}
\end{equation}
A similar proof is available for the Petz map:
\begin{widetext}
    \begin{equation} 
 \begin{aligned}
 \ret{\chn_2 \circ \chn_1,\rf}[\bullet] 
 & = \sqrt{\rf} (\chn_2 \circ \chn_1)^\dagger \left( \frac{1}{\sqrt{\chn_2 \circ \chn_1[\alpha]}} \bullet \frac{1}{\sqrt{\chn_2 \circ \chn_1[\alpha]}} \right) \sqrt{\rf} \\
& = \sqrt{\rf} (\chn_1^\dag \circ \chn_2^\dag) \left( \frac{1}{\sqrt{\chn_2 \circ \chn_1[\alpha]}} \bullet \frac{1}{\sqrt{\chn_2 \circ \chn_1[\alpha]}} \right) \sqrt{\rf} \\
& = \sqrt{\rf} \; \chn_1^\dag \Bigg( \frac{1}{\sqrt{\chn_1[\rf]}} 
\underbrace{\sqrt{\chn_1[\rf]} \; \chn_2^\dag \left( \frac{1}{\sqrt{\chn_2 \circ \chn_1[\alpha]}} \bullet \frac{1}{\sqrt{\chn_2 \circ \chn_1[\alpha]}} \right) \sqrt{\chn_1[\rf]}}_{\hat{\chn}_{2,{\chn_1 [\rf]}}}
\frac{1}{\sqrt{\chn_1[\rf]}} \Bigg) \sqrt{\rf} \\
&= \ret{\chn_1,\rf} \circ \ret{{\chn}_{2}, \chn_1[\rf]}.
\end{aligned}
\end{equation}
\end{widetext}

\subsection{Bayes' Rule on Classical Decomposition}\label{app:proofs-ccomp-bayes}
In this Appendix, we lay out an explicit proof for Result \ref{res1} on the level of matrices. This elucidates how retrodiction occurs on every step of the decomposition and how Bayesian inversion applies on physically valid channels, including those that are not stochastic (that is, even channels that do not correspond to a Markov matrix, such as the marginalization of a subspace and the assigning of an environment). Note that here, for the sake of symmetry with the quantum case, we will assume that the global process $\Phi$ is bijective and that the environment $\beta$ is uncorrelated with the input. Nevertheless, general insights also apply when these assumptions are removed. 

Firstly, we write Eq.~\eqref{eq:cdil} in terms of matrices:
\begin{equation} \label{cchndecomp}
    \boldsymbol{\cchn} = \bctrace_B \boldsymbol{\Phi} \hat{\bctrace}_{B,\square \otimes \cbath}.
\end{equation}
These matrices can be formalized in the following way. $\bctrace_B$ is a $d_A$ by $d_A d_B$ matrix while $\hat{\bctrace}_{B,\Lambda_{AB}}$ is $d_A d_B$ by $d_A$, defined as follows:
\begin{align}
        \bctrace_B &= \one_A \otimes \vone^\tpose \\
        \hat{\bctrace}_{B,\Lambda_{AB}} &= \pmqty{ 
        \vass{1} & \vzero & \cdots & \vzero \\
        \vzero & \vass{2} & \cdots & \vzero \\
        \vdots & \vdots & \ddots & \vdots \\
        \vzero & \vzero & \cdots & \vass{d_A}
        }.
\end{align}
Here, $\vone$ and $\vzero$ are $d_B$-dimensional column vectors of ones and zeros respectively, and $\vass{a}^\tpose = \pmqty{
    \frac{\Lambda(a,1)}{\sum_{\tilde{b}} \Lambda(a,\tilde{b})} &
    \frac{\Lambda(a,2)}{\sum_{\tilde{b}} \Lambda(a,\tilde{b})} &
    \cdots &
    \frac{\Lambda(a,d_B)}{\sum_{\tilde{b}} \Lambda(a,\tilde{b})} 
    }$, which comes from Eq.~\eqref{eq:cgenassign}. Now, $\square$ indicates any choice of state in $A$. The matrix $\hat{\ctrace}_{B,\square \otimes \cbath}$ is independent of this choice, since
\begin{equation}
    v(({\square \otimes \cbath})_{aB}) = \pmqty{
    \frac{\cbath(1)\square(a)}{\sum_{\tilde{b}} \cbath(\tilde{b}) \square(a)} \\
    \frac{\cbath(2)\square(a)}{\sum_{\tilde{b}} \cbath(\tilde{b}) \square(a)} \\
    \vdots \\
    \frac{\cbath(d_B)\square(a)}{\sum_{\tilde{b}} \cbath(\tilde{b}) \square(a)} 
    } =
    \pmqty{
    \cbath(1) \\
    \cbath(2) \\
    \vdots \\
    \cbath(d_B)
    } = \boldsymbol{\beta}.
\end{equation}
Thus, it follows that $\hat{\ctrace}_{B,\square \otimes \cbath}[\bullet_A] = \bullet_A \otimes \cbath_{B}$. Hence, each matrix performs the role of its corresponding channel. Turning now to retrodiction, we recall that Bayes' rule is \textit{composable} \cite{PB22}. This simply captures the time-reverse ordering and propagation of the reference prior, expected of Bayesian inversion. When composability~\eqref{composability_c} is applied to Eq.~\eqref{cchndecomp}, insights are yielded:
\begin{equation} 
 \begin{aligned}
    \rvp
    = \; \;  
    & \cret{\hat{\Sigma}_{B,\square \otimes \cbath},\crf} \cret{\Phi,\hat{\Sigma}_{B,\square \otimes \cbath}[\crf]}  \\
    & \cret{\ctrace_B,\Phi\circ\hat{\ctrace}_{B,\square \otimes \cbath}[\crf]} \\
    = \; \;  
    & \cret{\hat{\Sigma}_{B,\square \otimes \cbath},\crf} \cret{\Phi,\crf \otimes \cbath} \\
    & \cret{\ctrace_B,\Phi[\crf \otimes \cbath]}.
\end{aligned}
\end{equation}
Now, for both $\cret{\hat{\ctrace}_{B,  \square \otimes \cbath},\crf}$ and $\cret{\ctrace_B,\Phi[\crf \otimes \cbath]}$ in matrix form, applying Eq.~\eqref{bayesmat}, we get 

\begin{widetext}
\begin{equation} 
\begin{aligned}
\mathbf{D}_{\crf} \hat{\bctrace}_{B, \square \otimes \cbath} \mathbf{D}_{\crf \otimes \beta}^{-1} 
 &=\mathbf{D}_{\crf}
 \pmqty{\boldsymbol{\beta} & \vzero^\tpose & \cdots & \vzero^\tpose \\
        \vzero^\tpose & \boldsymbol{\beta} & \cdots & \vzero^\tpose \\
        \vdots & \vdots &  \ddots & \vdots \\
        \vzero^\tpose & \vzero^\tpose & \cdots & \boldsymbol{\beta}
 }\! \pmqty{\crf(1) \beta(1)  & 0 & \cdots & 0 \\
            0 & \crf(1) \beta(2)  & \cdots & 0 \\
            \vdots & \vdots &  \ddots & \vdots \\
            0 & 0 & \cdots & \crf(d_A)\beta(d_B)\\}^{-1} \\
&= \mathbf{D}_{\crf} \pmqty{\crf(1)^{-1} \cdots \crf(1)^{-1} & \vzero^\tpose & \cdots & \vzero^\tpose \\
                              \vdots & \vdots &  \ddots & \vdots \\ 
                              \vzero^\tpose & \vzero^\tpose & \cdots & \crf( d_A )^{-1} \cdots \crf( d_A )^{-1} }\\
&= \pmqty{\vone^\tpose & \vzero^\tpose & \cdots & \vzero^\tpose \\
        \vzero^\tpose & \vone^\tpose & \cdots & \vzero^\tpose \\
                              \vdots & \vdots &  \ddots & \vdots \\ 
                              \vzero^\tpose & \vzero^\tpose & \cdots & \vone^\tpose }\\
&= \bctrace_B.
\end{aligned}
\end{equation}

  \begin{equation} 
 \begin{aligned}
\mathbf{D}_{\Phi[\crf \otimes \cbath]}(\ctrace_B)^\tpose \mathbf{D}_{\ctrace_B \Phi[\crf \otimes \cbath]}^{-1}
    &= \mathbf{D}_{\Phi[\crf \otimes \cbath]} 
        \pmqty{\vone  & \vzero  & \cdots & \vzero  \\
           \vzero  & \vone  & \cdots & \vzero  \\ 
           \vdots & \vdots &  \ddots & \vdots \\ 
           \vzero  & \vzero  & \cdots & \vone  } \mathbf{D}_{\ctrace_B \Phi[\crf \otimes \cbath]}^{-1} \\
    &= \mathbf{D}_{\Phi[\crf \otimes \cbath]} 
    \pmqty{ v_a^{-1}(\ctrace_B \Phi[\crf \otimes \cbath](1)) & \vzero & \cdots & \vzero \\
            \vzero  & v_a^{-1}(\ctrace_B \Phi[\crf \otimes \cbath](2))  & \cdots & \vzero  \\ 
           \vdots & \vdots &  \ddots & \vdots \\     
           \vzero  & \vzero  & \cdots & v_a^{-1}(\ctrace_B \Phi[\crf \otimes \cbath](d_A))  } \\ 
    &= \pmqty{ 
        v(\Phi[\crf \otimes \cbath]_{1,B}) & \vzero & \cdots & \vzero \\
        \vzero & v(\Phi[\crf \otimes \cbath]_{2,B})  & \cdots & \vzero \\
        \vdots & \vdots & \ddots & \vdots \\
        \vzero & \vzero & \cdots & v(\Phi[\crf \otimes \cbath]_{d_A ,B}) 
        } \\
        &= \hat{\bctrace}_{B,\Phi[\crf \otimes \cbath]},
\end{aligned}
\end{equation}
where the $d_B$-length vector $v_a(p(\tilde{a}))^\tpose= \pmqty{p(\tilde{a}) & \cdots & p(\tilde{a})}$.

Thus it is shown that, given the propagated reference priors, the Bayesian inversion of the assignment map gives a marginalizing channel and vice versa:
\begin{align}
    &\cret{\hat{\ctrace}_{B,\square \otimes \cbath},\crf} = \ctrace_B, \\
    &\cret{\Phi, \crf \otimes \beta} = \hat{\Phi}_{\crf \otimes \beta} = \Phi^{-1}, \; \; \because \eqref{eq:deterministic} \\ 
    &\cret{\ctrace_B,\Phi[\crf \otimes \cbath]} = \hat{\ctrace}_{B,\Phi[\crf \otimes \cbath]}.
\end{align}
Together, we can write:
\begin{equation}
    \rvp = \ctrace_B \circ \Phi^{-1} \circ \hat{\ctrace}_{B, \Phi[\crf \otimes \cbath]}\,.\label{eq:cretdecomp} 
\end{equation}

\subsection{The Petz Map on Quantum Dilations}\label{app:retonqdil}
Here we provide some explicit proofs for each step of Eq.~\eqref{eq:qdildecomp}. Firstly, we know $\ret{\mathcal{U}, \qam{\square \otimes \beta}[\rf]} = \hat{\mathcal{U}}_{\rf \otimes \beta} = \mathcal{U}^\dag$ because of Eqs.~\eqref{eq:qassprod}~and~\eqref{eq:deterministic}. Secondly, one can easily verify that the adjoint \eqref{eq:adj} of the quantum assignment map is
\begin{equation}
    \left( \qam{\Gamma_{AB}} \right)^\dag[\bullet] = \TrB\left(\frac{1}{\sqrt{\TrB(\Gamma_{AB})}} \otimes \one_B \sqrt{\Gamma_{AB}} \bullet \sqrt{\Gamma_{AB}} \frac{1}{\sqrt{\TrB(\Gamma_{AB})}} \otimes \one_B \right). 
\end{equation}
Hence, 
\begin{equation}
\begin{aligned}
    \left( \qam{\square \otimes \cbath} \right)^\dag[\bullet] 
    & = \TrB\left(  \frac{1}{\sqrt{\square}} \otimes \one_B \sqrt{\square \otimes \cbath} \bullet \sqrt{\square \otimes \cbath}  \frac{1}{\sqrt{\square}} \otimes \one_B \right) \\
    \Rightarrow \ret{\qam{\square \otimes \beta}, \rf}  &= \sqrt{\rf} \left( \qam{\square \otimes \cbath} \right)^\dag \left( 
    \frac{1}{\sqrt{\rf \otimes \beta}} \bullet  \frac{1}{\sqrt{\rf \otimes \beta}}
    \right) \sqrt{\rf} \\
    &= \sqrt{\rf} \TrB \left( \frac{1}{\sqrt{\square}} \otimes \one_B \sqrt{\square \otimes \cbath} \frac{1}{\sqrt{\rf \otimes \beta}} \bullet \frac{1}{\sqrt{\rf \otimes \beta}} \sqrt{\square \otimes \cbath}  \frac{1}{\sqrt{\square}} \otimes \one_B \right)\sqrt{\rf} \\
    & = \sqrt{\rf} \TrB\left( \frac{1}{\sqrt{\rf}} \otimes \one_B \bullet \frac{1}{\sqrt{\rf}} \otimes \one_B \right) \sqrt{\rf}
    = \TrB(\bullet).
    \end{aligned}
\end{equation}
Thirdly, one can easily show that the adjoint of the partial trace is given by 
\begin{equation}
    \TrB^\dag [\bullet_A] = \bullet_A \otimes \one_B.
\end{equation}
Which implies, 
\begin{equation}
\begin{aligned}
\ret{\TrB{}, \Gamma_{AB}} & = \sqrt{\Gamma_{AB}} \TrB^\dag \left(  \frac{1}{\sqrt{\TrB(\Gamma_{AB})}} \bullet \frac{1}{\sqrt{\TrB(\Gamma_{AB})}}\right) \sqrt{\Gamma_{AB}}   \\ 
& =  \sqrt{\Gamma_{AB}} \left(  \frac{1}{\sqrt{\TrB(\Gamma_{AB})}} \bullet \frac{1}{\sqrt{\TrB(\Gamma_{AB})}} \otimes \one_B \right) \sqrt{\Gamma_{AB}} = \qam{\Gamma_{AB}} \\
& \Rightarrow \ret{\TrB{}, \mathcal{U}[\rf \otimes \beta]} = \qam{\, \mathcal{U}[\rf\otimes\beta]}.
\end{aligned}
\end{equation}
Thus, every step of Eq.~\eqref{eq:qdildecomp} is proven, thus giving an alternative route to Eq.~\eqref{ptz3}. These derivations highlight that Bayesian inference can be applied in a logically consistent and physically insightful way to any valid channel, including that of marginalization and the assignment of environments and so on.
\end{widetext}
\section{Special Case of Pure Product-preserving Tuple for Two-qubit Unitaries}\label{apd:special-case-pure} 
For $a_0 \neq 0$, the explicit form of Eq.~\eqref{eq:pure-PIPO-equation} is
\begin{widetext}
\begin{equation}
\begin{aligned}
(u_A^\dag \otimes u_B^\dag) U (\ket{\alpha}\otimes\ket{\beta})
    \;&\widehat{=}
    \pmqty{
        e^{it_3}\cos(t_1-t_2) & 0 & 0 & ie^{it_3}\sin(t_1-t_2) \\
        0 & e^{-it_3}\cos(t_1+t_2) & ie^{-it_3}\sin(t_1+t_2) & 0 \\
        0 & ie^{-it_3}\sin(t_1+t_2) & e^{-it_3}\cos(t_1+t_2) & 0 \\
        ie^{it_3}\sin(t_1-t_2) & 0 & 0 & e^{it_3}\cos(t_1-t_2) 
    }\bqty{\pmqty{1\\x}\otimes\pmqty{b_0\\b_1}} \\
    &= \pmqty{
        b_0 e^{i t_3} \cos{\left(t_1 - t_2 \right)} + i b_1 x e^{i t_3} \sin{\left(t_1 - t_2 \right)}\\
        i b_0 x e^{- i t_3} \sin{\left(t_1 + t_2 \right)} + b_1 e^{- i t_3} \cos{\left(t_1 + t_2 \right)}\\
        b_0 x e^{- i t_3} \cos{\left(t_1 + t_2 \right)} + i b_1 e^{- i t_3} \sin{\left(t_1 + t_2 \right)}\\
        i b_0 e^{i t_3} \sin{\left(t_1 - t_2 \right)} + b_1 x e^{i t_3} \cos{\left(t_1 - t_2 \right)}
    }
\end{aligned}
\end{equation}
\end{widetext}
The condition for $(U,\ket{\alpha},\ket{\beta})$ to be product-preserving is given in Eq.~\eqref{eq:quadratic-equation}, which can be expressed as a quadratic equation $ax^2 + bx + c = 0$, with
\begin{equation}\label{eq:quadratic-equation-full}
\begin{aligned}
    a &= \frac{i}{2}\bqty{
        b_0^2 e^{-i2t_3}\sin(2t_1+2t_2) -
        b_1^2 e^{i2t_3}\sin(2t_1-2t_2)
    }, \\
    b &= b_0b_1\bqty{
        e^{-i2t_3}\cos(2t_1+2t_2) -
        e^{i2t_3}\cos(2t_1-2t_2) 
    }, \\
    c &= -\frac{i}{2}\bqty{
        b_0^2 e^{i2t_3}\sin(2t_1-2t_2) -
        b_1^2 e^{-i2t_3}\sin(2t_1+2t_2)
    }.
\end{aligned}
\end{equation}
Unless $a=b=0$ and $c\neq 0$, a solution always exists for $x$, and the corresponding $\ket{\alpha}$ for the given $U$ and $\ket{\beta}$ to make the tuple product-preserving.

However, if $a=b=0$ and $c\neq 0$, the condition simplifies to $c=0$, which is a contradiction. In that case, setting $\ket{\alpha} = \ket{1}$ gives
\begin{equation}
(u_A^\dag \otimes u_B^\dag) U (\ket{1}\otimes\ket{\beta})
\widehat{=}
\pmqty{
i b_1 e^{i t_3} \sin{\left(t_1 - t_2 \right)}\\
i b_0 e^{- i t_3} \sin{\left(t_1 + t_2 \right)}\\
b_0 e^{- i t_3} \cos{\left(t_1 + t_2 \right)}\\
b_1 e^{i \cchn} e^{i t_3} \cos{\left(t_1 - t_2 \right)}
},
\end{equation}
and it can be verified that Eq.~\eqref{eq:quadratic-equation} is satisfied with $a=0$.

In summary, for a given $U$ and $\ket{\beta}$, if $a$ as defined in Eq.~\eqref{eq:quadratic-equation-full} is nonzero, $ax^2+bx+c=0$ is solved for $x$ with the coefficients given in Eq.~\eqref{eq:quadratic-equation-full}, and we set $\ket{\alpha} = v_A \pqty{\ket{0} + x\ket{1}}/\sqrt{1+\abs{x}^2}$. Otherwise, if $a=0$, we set $\ket{\alpha} = \ket{1}$. In either case, $(U,\ket{\alpha},\ket{\beta})$ forms a product-preserving tuple.

\section{Lemmas About Two-qubit Generalized-thermal Unitaries}\label{app:lemmas2q}

\begin{lemma}\label{thm:general-GG}
    For $U$ such that $t_k \bmod (\pi/4) = 0$ for at most one $k\in\{1,2,3\}$, $U$ is generalized thermal with respect to $H_A$, $H_B$, $H_A'$, and $H_B'$ if and only if
    \begin{equation}\label{eq:all-GG-Hamiltonians}
    \begin{aligned}
        H_A &= v_A H  v_A^\dag, &
        H_B &= v_B  \tilde{\sigma} H \tilde{\sigma}^\dag v_B^\dag, \\
        H_A' &= u_A \tilde{\varsigma} H \tilde{\varsigma}^\dag u_A^\dag, &
        H_B' &= u_B \tilde{\varsigma} \tilde{\sigma} H \tilde{\sigma}^\dag \tilde{\varsigma}^\dag u_B^\dag,
    \end{aligned}
    \end{equation}
    where $\tilde{\sigma} := \prod_{j=1}^3 \sigma_j^{m_{-j}}$ and $\tilde{\varsigma} := \prod_{j=1}^3 \sigma_j^{m_{+j}}$ are ``bit flips'' with $m_{\pm j}$ defined below, and $H$ is any Hamiltonian such that $\Tr(H\sigma_j) \neq 0$ only when $(t_j-t_k)\bmod(\pi/2)=0$ or $(t_j+t_k)\bmod(\pi/2)=0$ for all permutations $(j,k,l)$ of $(1,2,3)$.
    
    The definition of $m_{\pm j}$ is as follows: if $(t_k-t_l)\bmod(\pi/2)=0$ or $(t_k+t_l)\bmod(\pi/2)=0$, where $j\neq k \neq l$, then
    \begin{equation}
    \begin{aligned}
        (-1)^{m_{+j}} &= \operatorname{sgn}\!\left[\cos(2t_k)\cos(2t_l)\right], \\
        (-1)^{m_{-j}} &= \operatorname{sgn}\!\left[\tan(2t_k)\tan(2t_l) \right].
    \end{aligned}
    \end{equation}
    Otherwise, if $(t_j \pm t_k)\bmod(\pi/2) \neq 0$, $m_{\pm j}$ should be chosen so that $(-1)^{m_{\pm 1} + m_{\pm 2} + m_{\pm 3}} = 1$.
\end{lemma}
\begin{proof}
Let us first consider the unitary $\widetilde{U} = e^{i\widetilde{H}}$ where $\widetilde{H} = \sum_{k=1}^3 t_k \sigma_k\otimes \sigma_k$, with traceless Hamiltonians $\widetilde{H}_A$, $\widetilde{H}_B$, $\widetilde{H}_A'$, and $\widetilde{H}_B'$ with the decomposition $\widetilde{H}_A = \sum_{k=1}^3 h_{A,k}\sigma_k$ (and similarly for $\widetilde{H}_B$, $\widetilde{H}_{A/B}'$).

The key step is to work out that for $\mu \in \{0,1,2,3\}$, with $\sigma_0=\one$, and cyclic permutations $(j,k,l)$ of $(1,2,3)$,
\begin{equation}\label{eq:before-BKH}
\begin{aligned}
    [i\widetilde{H},\sigma_\mu\otimes\sigma_\mu] &= 0 \\
    [i\widetilde{H},\one\otimes\sigma_j \pm \sigma_j\otimes\one] &=
    2(t_k \mp t_l)(\sigma_k\otimes\sigma_l \pm \sigma_l\otimes\sigma_k) \\
    [i\widetilde{H},\sigma_k\otimes\sigma_l \pm \sigma_l\otimes\sigma_k] &=
    - 2(t_k \mp t_l)(\one\otimes\sigma_j \pm \sigma_j\otimes\one).
\end{aligned}
\end{equation}
Together with the Baker--Campbell--Hausdorff lemma, Eq.~\eqref{eq:before-BKH} can be used to find $\widetilde{U}(\sigma_\mu\otimes\sigma_\nu)\widetilde{U}^\dag$ for all $\mu,\nu \in \{0,1,2,3\}$.

Having these transformations, we notice that the condition of generalized thermal channel
\begin{equation}\label{eq:diagonal-condition}
\widetilde{U}( \widetilde{H}_A\otimes\one + \one\otimes\widetilde{H}_B )\widetilde{U}^\dag = \widetilde{H}_A'\otimes\one + \one\otimes\widetilde{H}_B', 
\end{equation}
implies that there are no products $\sigma_j\otimes\sigma_k$ on the right-hand side. Setting the same components to zero on the left leads to
\begin{equation}\label{eq:diagonal-condition-terms}
\begin{aligned}
    h_{A,j}\sin(2t_k)\cos(2t_l) &= h_{B,j}\sin(2t_l)\cos(2t_k),\\
    h_{A,j}\sin(2t_l)\cos(2t_k) &= h_{B,j}\sin(2t_k)\cos(2t_l).
\end{aligned}
\end{equation}
Let us first take $t_k\bmod(\pi/4) \neq 0$ for all $k$, so that $\{\tan(2t_k)\}_{k=1}^3$ are all nonzero and finite.

To satisfy Eq.~\eqref{eq:diagonal-condition-terms}, the following must hold:
\begin{itemize}
    \item if $h_{A,j},h_{B,j} \neq 0$, then $\tan^2(2t_k) = \tan^2(2t_l)$, which further implies $h_{A,j} = \pm h_{B,j}$,
    \item if $\tan^2(2t_k) \neq \tan^2(2t_l)$, then $h_{A,j} = h_{B,j} = 0$.
\end{itemize}
So, it is always true that $h_{B,j} = \pm h_{A,j}$, and they are nonzero only if $\lvert\tan(2t_k)\rvert = \lvert\tan(2t_l)\rvert$, which in turn implies that $(t_j-t_k)\bmod(\pi/2)=0$ or $(t_j+t_k)\bmod(\pi/2)=0$. The sign ambiguity is due to the absolute value in the latter expression.

Now, turning to the $\one\otimes\sigma_j$ and $\sigma_j\otimes\one$ components of Eq.~\eqref{eq:diagonal-condition}, we have
\begin{equation}\label{eq:off-diagonal-condition}
\begin{aligned}
    h'_{A,j} &= h_{A,j}\cos(2t_k)\cos(2t_l) + h_{B,j}\sin(2t_k)\sin(2t_l) \\
    = &\pm h_{A,j}, \\
    h'_{B,j} &= h_{B,j}\cos(2t_k)\cos(2t_l) + h_{A,j}\sin(2t_k)\sin(2t_l) \\
    &= \pm h_{B,j}.
\end{aligned}
\end{equation}
In the end, this means that $h_{B,j} = \pm h_{A,j} = \pm h_{A,j}' = \pm h_{B,j}'$. Therefore, after defining $H := \widetilde{H}_{A}$, we have $H_{B} = \sigma_\mu H \sigma_\mu$, $H_A'=\sigma_\nu H \sigma_\nu$, and $H_B' = \sigma_\gamma H \sigma_\gamma$ for some ``bit-flips'' or identities $\sigma_\mu$, $\sigma_\nu$, and $\sigma_\gamma$ that correct the sign ambiguities.

Finally, we arrive at Eq.~\eqref{eq:all-GG-Hamiltonians} by substituting $\widetilde{U}$ into Eq.~\eqref{eq:standard-two-qubit-unitary-decomposition}. The form of $\tilde{\sigma}$ and $\tilde{\varsigma}$ is found by keeping track of the signs of the trigonometric functions.

The proof also holds when $t_k\bmod(\pi/4) = 0$ for at most one $k$: substituting $t_k$ into Eq.~\eqref{eq:diagonal-condition-terms} gives $h_{A,j} = h_{B,j} = 0$ for all $j\neq k$, after which the rest of the proof follows as stated.
\end{proof}

\begin{lemmap}{thm:general-GG}\label{thm:general-GG-special-1}
    For $U$ such that $t_j,t_k \bmod (\pi/4) = 0$ and $t_l \bmod (\pi/4) \neq 0$, $U$ is generalized thermal with respect to $H_A$, $H_B$, $H_A'$, and $H_B'$ if and only if
    \begin{itemize}
        \item $t_j\bmod(\pi/2) = t_k\bmod(\pi/2) = 0$ and, up to a constant offset,
        \begin{equation}
        \begin{aligned}
            H_A &= h_A v_A\sigma_lv_A^\dag, &
            H_B &= h_B v_B\sigma_l v_B^\dag, \\
            H_A' &= h_A u_A \sigma_l u_A^\dag, &
            H_B' &= h_B u_B \sigma_l u_B^\dag,
        \end{aligned}
        \end{equation}
        or;
        \item $t_j\bmod(\pi/2) = t_k\bmod(\pi/2) = \pi/4$ and, up to a constant offset, 
        \begin{equation}
        \begin{aligned}
            H_A &= h_A v_A\sigma_lv_A^\dag, &
            H_B &= h_B v_B\sigma_l v_B^\dag, \\
            H_A' &= h_B u_A \sigma_l u_A^\dag, &
            H_B' &= h_A u_B \sigma_l u_B^\dag.
        \end{aligned}
        \end{equation}
    \end{itemize}
\end{lemmap}

\begin{proof}
    First, let $2t_k = n_k\pi/2$ for integer $n_k$, which gives $2\cos(2t_k) = (-1)^{\lfloor n_k/2\rfloor}(1+(-1)^{n_k})$ and
     $2\sin(2t_k) = (-1)^{\lfloor n_k/2 \rfloor}(1-(-1)^{n_k})$. Substituting it into Eq.~\eqref{eq:diagonal-condition-terms},
    \begin{equation}
    \begin{aligned}
    h_{A,j}(1-(-1)^{n_k})\cos(2t_l) &= h_{B,j}\sin(2t_l)(1+(-1)^{n_k}),\\
    h_{B,j}\sin(2t_l)(1+(-1)^{n_k}) &= h_{A,j}(1-(-1)^{n_k})\cos(2t_l).
    \end{aligned}
    \end{equation}
    Since one of $1+(-1)^{n_k}$ or $1-(-1)^{n_k}$ must be zero, it is either the case that both right columns are zero, or both left columns are zero. As such, $h_{A,j}=h_{B,j}=0$.
    
    Similarly letting $2t_j = n_j\pi/2$, repeating the above steps give $h_{A,k}=h_{B,k}=0$.

    Meanwhile, for the $h_{A,l}$ and $h_{B,l}$ case,
    \begin{equation}
    \begin{aligned}
    &h_{A,l}(1-(-1)^{n_k})(1+(-1)^{n_j}) \\
    &\qquad{}={} h_{B,l}(1+(-1)^{n_k})(1-(-1)^{n_j}),\\
    &h_{B,l}(1+(-1)^{n_k})(1-(-1)^{n_j}) \\
    &\qquad{}={} h_{A,l}(1-(-1)^{n_k})(1+(-1)^{n_j}).
    \end{aligned}
    \end{equation}
    When the parities of $n_j$ and $n_k$ do not match, $h_{A,l} = h_{B,l} = 0$. In other words, if $t_j\bmod(\pi/2) \neq t_k\bmod(\pi/2)$, then $U$ is not generalized thermal.
    
    Otherwise, the equations are trivially satisfied. Then, Eq.~\eqref{eq:off-diagonal-condition} gives
    \begin{equation}
        h_{A,l}' = \begin{cases}
        h_{A,l} & \text{if }t_j\bmod\frac{\pi}{2}=0 \\ 
        h_{B,l} & \text{if }t_j\bmod\frac{\pi}{2}=\frac{\pi}{4}, \\ 
        \end{cases}
    \end{equation}
    with analogous expressions for $h_{B,l}'$. Placing these expressions back into the definitions of the Hamiltonians end the proof.
\end{proof}

\begin{lemmapp}{thm:general-GG}\label{thm:general-GG-special-2}
    For $U$ such that $t_k \bmod (\pi/4) = 0$ for all $k \in \{1,2,3\}$, $U$ is generalized thermal with respect to $H_A$, $H_B$, $H_A'$, and $H_B'$ if and only if
    \begin{itemize}
        \item $t_1 \bmod \frac{\pi}{2} = t_2  \bmod \frac{\pi}{2}  = t_3 \bmod \frac{\pi}{2} = 0 $, $H_A$ and $H_B$ are arbitrary, and
        \begin{equation}
        \begin{aligned}
        H_A' &= u_A v_A^\dag H_A v_A u_A^\dag,\\
        H_B' &= u_B v_B^\dag H_B v_B u_B^\dag,
        \end{aligned}
        \end{equation}
        or;
        \item $t_1 \bmod \frac{\pi}{2} = t_2  \bmod \frac{\pi}{2}  = t_3 \bmod \frac{\pi}{2} = \frac{\pi}{4}$, $H_A$ and $H_B$ are arbitrary, and
        \begin{equation}
        \begin{aligned}
        H_A' &= u_A v_B^\dag H_B v_B u_A^\dag, \\
        H_B' &= u_B v_A^\dag H_A v_A u_B^\dag.
        \end{aligned}
        \end{equation}
    \end{itemize}
\end{lemmapp}

\begin{proof}
    The proof is almost identical to that of Lemma~\ref{thm:general-GG-special-1}, so we shall only provide a sketch. Substituting $2t_k = n_k\pi/2$ for all $k \in \{1,2,3\}$ into Eq.~\eqref{eq:diagonal-condition-terms}, we will find that $h_{A,k}=h_{B,k}=0$ for all $k$ if the parities of $n_1$, $n_2$, and $n_3$ do not all match. This means that $U$ is not generalized thermal if $t_j\bmod\frac{\pi}{2} \neq t_k\bmod\frac{\pi}{2}$ for any two $j \neq k$.
    
    Otherwise, Eq.~\eqref{eq:off-diagonal-condition} gives, for all $k$,
    \begin{equation}
        h_{A,k}' = \begin{cases}
        h_{A,k} & \text{if }t_j\bmod\frac{\pi}{2}=0 \\ 
        h_{B,k} & \text{if }t_j\bmod\frac{\pi}{2}=\frac{\pi}{4}, \\ 
        \end{cases}
    \end{equation}
    with similar expressions for $h_{B,k}$. Substituting these back into the Hamiltonians concludes the proof.
\end{proof}

\section{Retrodiction with Rank-deficient Outputs is Ill-defined in General}\label{apd:ill-defined-how}
The Petz recovery map in Eq.~\eqref{petz} involves the operator $(\chn[\alpha])^{-\frac{1}{2}}$, which does not exist when $\chn[\alpha]$ is rank-deficient. One might attempt to circumvent this problem in in several ways:

(1) \emph{Define $\hat{\chn}_{\alpha,1}$ using the pseudoinverse}. The inverse can be defined only on the support of $\chn[\alpha]$, resulting in the so-called pseudoinverse of $\chn[\alpha]$, as is the convention in quantum information \cite{quantum-information-notes}.

(2) \emph{Define $\hat{\chn}_{\alpha,2}$ in a full-rank neighbourhood}. Another possibility is to perturb the prior or the channel so that the output state is perturbed as $\chn[\alpha] \to \chn^{(\varepsilon)}[\alpha]$, such that $\lim_{\varepsilon\to 0}\chn^{(\varepsilon)}[\alpha] = \chn[\alpha]$ and $\chn^{(\varepsilon)}[\alpha]$ is full rank for every $\varepsilon > 0$. Then, the Petz map $\hat{\chn}_\alpha^{(\varepsilon)}$ is well-defined for every $\varepsilon > 0$, and the retrodiction channel is defined as the limit $\hat{\chn}_{\alpha,2} := \lim_{\varepsilon \to 0}\hat{\chn}_{\alpha}^{(\varepsilon)}$.

(3) \emph{Na\"ively define $\hat{\chn}_{
\alpha,3}$ with product-preserving tuples}. As Eq.~\eqref{pp-tr} gives a simple form for the retrodiction channel when $(U,\alpha,\beta)$ is a product-preserving tuple for full rank $\alpha$ and $\beta$, one might extend it to the rank-deficient case by imposing that the retrodiction channel should also take the same form
\begin{equation}\label{eq:naive-pp-retrodiction}
\hat{\chn}_{\alpha,3}[\bullet] := \TrB\bqty{U^\dag\pqty{\bullet\otimes\beta'}U}
\end{equation}
even when $\alpha'$ is not full rank.

In the following section, we show that the retrodiction channels as defined by the above approaches do not agree in general. As such, the retrodiction channel for a rank-deficient output state depends on the convention chosen, and cannot be consistently defined when the output state is rank-deficient.

\subsection{Comparison of different approaches to retrodicting rank-deficient outputs}

\subsubsection{Unitary channels}
Consider a unitary channel $\mathcal{U}[\bullet] := U\bullet U^\dag$ with the dilation $\mathcal{U}[\bullet] = \Tr_B[ (U\otimes V)(\bullet\otimes\beta)(U^\dag \otimes V^\dag) ]$ for arbitrary unitary $V$ and ancilla $\beta$. $(U\otimes V,\alpha,\beta)$ is clearly a product-preserving tuple, since
\begin{equation}
    (U\otimes V)(\alpha\otimes\beta)(U^\dag\otimes V) = (U\alpha U^\dag) \otimes (V\beta V^\dag) =: \alpha'\otimes\beta'.
\end{equation}

(1) \emph{Define $\hat{\mathcal{U}}_{\alpha,1}$ using the pseudoinverse}. Given $\alpha' = \sum_{k=1}^{|\text{supp}(\alpha')|} a'_k\ketbra{a'_k}$ with $a'_k > 0$ for all $k \in \{1,2,\dots,|\text{supp}(\alpha')|\}$, the pseudoinverse of $\sqrt{\alpha'}$ is defined as
\begin{equation}
    {\sqrt{\alpha'}}^{\,\boldsymbol{\pmb{+}}} := \sum_{k=1}^{|\text{supp}(\alpha')|} \frac{1}{\sqrt{a'_k}}\ketbra{a'_k}.
\end{equation}
The pseudoinverse has the property $\sqrt{\alpha'}{\sqrt{\alpha'}}^{\,\boldsymbol{\pmb{+}}} = {\sqrt{\alpha'}}^{\,\boldsymbol{\pmb{+}}}\sqrt{\alpha'} = \Pi_{\alpha'}$, where $\Pi_{\alpha'} := \sum_{k=1}^{|\text{supp}(\alpha')|} \ketbra{a'_k}$ is the projector onto the support of $\alpha'$. Therefore, with the replacement $\alpha^{\prime -\frac{1}{2}} \to {\sqrt{\alpha'}}^{\,\boldsymbol{\pmb{+}}}$, the retrodiction of a unitary channel can be found to be
\begin{equation}
    \hat{\mathcal{U}}_{\alpha,1}[\bullet] = U^\dag \Pi_{\alpha'}\bullet\Pi_{\alpha'} U.
\end{equation}
This channel is completely positive but not trace-preserving in general.

(2) \emph{Define $\hat{\chn}_{\alpha,2}$ in a full-rank neighbourhood}. Given any state $\alpha_{\varepsilon}$ in the full-rank neighbourhood of $\alpha$, that is, $\alpha_{\varepsilon}$ is full rank and $\lim_{\varepsilon \to 0}\alpha_{\varepsilon} = \alpha$, we find the retrodiction defined on prior $\alpha_{\varepsilon}$ to be $\hat{\mathcal{U}}_{\alpha}^{(\varepsilon)}[\bullet] = U^\dag\bullet U$. Taking the limit gives
\begin{equation}
    \hat{\mathcal{U}}_{\alpha,2}[\bullet] = \lim_{\varepsilon\to 0}\hat{\mathcal{U}}_{\alpha}^{(\varepsilon)}[\bullet] = U^\dag \bullet U,
\end{equation}
which does not depend on how we chose to perturb $\alpha$.

(3) \emph{Na\"ively define $\hat{\chn}_{
\alpha,3}$ with product-preserving tuples}. Substituting the product-preserving tuple $(U, \alpha, \beta)$ into Eq.~\eqref{eq:naive-pp-retrodiction}, we have
\begin{equation}
    \hat{\mathcal{U}}_{\alpha,3}[\bullet] = U^\dag \bullet U.
\end{equation}

Therefore, we find that $\hat{\mathcal{U}}_{\alpha,2} = \hat{\mathcal{U}}_{\alpha,3}$ but $\hat{\mathcal{U}}_{\alpha,1}\neq\hat{\mathcal{U}}_{\alpha,2}$: choosing the pseudoinverse results in a different retrodiction channel in comparison to the other two choices.

\subsubsection{Erasure channels}\label{apd:cannot-retrodict-garbage}

Consider the erasure channel $\chn[\bullet] := \Tr(\bullet) \ketbra{0}$ with input and output dimensions $d_A$. Let its dilation be
\begin{equation}
    \chn[\bullet] = \Tr_B\bqty{U\pqty{\bullet\otimes\ketbra{0}\otimes\ketbra{0}\otimes\ketbra{0}}U^\dag},
\end{equation}
where the ancilla space is $\mathcal{H}_B = \mathcal{H}_{B_1}\otimes\mathcal{H}_{B_2}\otimes\mathcal{H}_{B_3} = \mathcal{C}^{d_A} \otimes \mathcal{C}^{d_A} \otimes \mathcal{C}^{2}$, and
\begin{equation}
    U = \operatorname{SWAP}{}\otimes{}\one\otimes\ketbra{0} + V\otimes\ketbra{1},
\end{equation}
with $\operatorname{SWAP}$ acting on $\mathcal{H}_{A}\otimes\mathcal{H}_{B_1}$ and $V$ acting on $\mathcal{H}_A \otimes \mathcal{H}_{B_1}\otimes\mathcal{H}_{B_2}$. Since $U(\alpha\otimes\ketbra{000})U^\dag = \ketbra{0}\otimes\alpha\otimes\ketbra{00}$, we again have a product-preserving tuple.

(1) \emph{Define $\hat{\chn}_{\alpha,1}$ using the pseudoinverse}. This gives $\hat{\chn}_{\alpha,1}[\bullet] = \bra{0}\!\bullet\!\ket{0}\alpha$.

(2) \emph{Define $\hat{\chn}_{\alpha,2}$ in a full-rank neighbourhood}. Since $\chn[\alpha]=\ketbra{0}$ is rank-deficient for all $\alpha$, it is not enough to perturb just the prior. Instead, we need to perturb the ancilla as
\begin{equation}
    \beta_\varepsilon := (1-\varepsilon)\ketbra{000} + \varepsilon \Gamma_{12}\otimes\ketbra{1},
\end{equation}
with $\Gamma_{12} \in \mathcal{H}_{B_1}\otimes\mathcal{H}_{B_2}$. This effectively perturbs the channel as
\begin{equation}
    \chn^{(\varepsilon)}[\bullet] := \Tr_B\bqty{U\pqty{\bullet\otimes\beta_\varepsilon}} = (1-\varepsilon)\chn[\bullet] + \varepsilon\Phi[\bullet],
\end{equation}
where $\Phi[\bullet]:=\Tr_{B_1,B_2}[V(\bullet\otimes\Gamma_{12})V^\dag]$ with $V$ and $\Gamma_{12}$ chosen such that $\Phi[\alpha]$ is full rank. This ensures that the inverse of $\chn^{(\varepsilon)}[\alpha]$ exists for every $\varepsilon > 0$. While obviously $\lim_{\varepsilon\to 0}\chn^{(\varepsilon)} = \chn$, we are going to show that $\lim_{\varepsilon\to0}\hat{\chn}_\alpha^{(\varepsilon)}$ still depends on $\Phi$ in general. As such, if we are given the rank-deficient erasure channel, there is no unique consistent definition of $\hat{\chn}_{\alpha}$.

For simplicity, let us take $\Phi$ to be a classical channel. That is,
\begin{equation}
    \Phi[\ketbra{j}{j'}] = \delta_{j,j'} \sum_{k=0}^{d-1} \cchn(k|j) \ketbra{k},
\end{equation}
where $0 \leq \cchn(k|j) \leq 1$ is a conditional probability distribution, normalized as $\sum_{k=0}^{d-1} \cchn(k|j) = 1$ but otherwise arbitrary. Likewise, we take the classical state $\alpha = \sum_{j=0}^{d-1} a_j \ketbra{j}$ as the prior. By direct computation of Eq.~\eqref{petz}, the action of $\hat{\chn}_\alpha^{(\varepsilon)}$ on the basis operators $\Bqty{\ketbra{k}{k'}}_{k,k'=0}^{d-1}$ is found to be 
\begin{equation}\label{eq:perturbed-garbage}
    \hat{\chn}_\alpha^{(\varepsilon)}[\ketbra{k}{k'}] = \frac{(1-\varepsilon)\delta_{k,k'}\delta_{k,0} \alpha + \varepsilon \ev{\Phi[\alpha]}{k} \hat{\Phi}_\alpha[\ketbra{k}{k'}]}{(1-\varepsilon)\delta_{k,0} + \varepsilon \ev{\Phi[\alpha]}{k}},
\end{equation}
where
\begin{equation}
    \hat{\Phi}_\alpha[\bullet] = \sum_{j=0}^{d-1} \pqty{\sum_{k=0}^{d-1} \frac{\ev{\bullet}{k} \cchn(k|j) a_j}{\ev{\Phi[\alpha]}{k}}} \ketbra{j}
\end{equation}
is the retrodiction channel of $\Phi$ with respect to the prior $\alpha$, similarly computed using Eq.~\eqref{petz}.

All this is well defined, but taking the limit of Eq.~\eqref{eq:perturbed-garbage}, we find
\begin{equation}\label{eq:limited-garbage}
    \hat{\chn}_{\alpha,2}[\bullet] = \lim_{\varepsilon\to0}\hat{\chn}_\alpha^{(\varepsilon)}[\bullet] = 
    \ev{\bullet}{0} \alpha + \hat{\Phi}_\alpha\bqty\big{ \bullet - \ketbra{0}\!\bullet\!\ketbra{0} }.
\end{equation}
We would have liked this expression to be independent of $\Phi$ as long as $\Phi[\alpha]$ is full rank, but this is clearly not the case. 

(3) \emph{Na\"ively define $\hat{\chn}_{\alpha,3}$ with product-preserving tuples}. A simple substitution gives
\begin{equation}
    \hat{\chn}_{\alpha,3}[\bullet] = \Tr(\bullet)\alpha.
\end{equation}

We see here that all three approaches result in different retrodiction channels in general, with $\hat{\chn}_{\alpha,2} = \hat{\chn}_{\alpha,3}$ only for the specific choice $\hat{\Phi}_{\alpha}[\bullet] =  \Tr(\bullet)\alpha$.

Since the Petz map is equivalent to classical Bayesian inversion when all channels and states involved are diagonal in a chosen basis \cite{PB22}, we have demonstrated that retrodiction channels cannot be consistently defined in general when the output state is rank-deficient, even in classical theory. This is of course the case \textit{a fortiori} in quantum theory. \\

\subsubsection{Product-preserving two-qubit dilations}\label{apd:cannot-anyhow-PIPO}

Finally, let us investigate the difference between the three approaches in more detail for the two-qubit case. Consider the product-preserving tuple $(U,\ketbra{\alpha_+},\ketbra{\beta_+})$, with $U(\ket{\alpha_+}\otimes\ket{\beta_+}) = \ket{\alpha_+'}\otimes\ket{\beta_+'}$. Let $\ket{\alpha_-}$ be the state orthogonal to $\ket{\alpha_+}$, with $\ket{\beta_-}$, $\ket{\alpha_-'}$, and $\ket{\beta_-'}$ defined analogously.

Meanwhile, define $u_A$ as the local unitary such that $u_A\ket{\alpha_\pm} = \ket{\alpha'_\pm}$ and $u_B$ as the local unitary such that $u_A\ket{\beta_\pm} = \ket{\beta'_\pm}$.

Then, in the $\{\ket{\alpha_\pm}\} \otimes \{\ket{\beta_\pm}\}$ basis, $U$ takes the form

\begin{equation}
    \pqty{u_A^\dag \otimes u_B^\dag}U \;\widehat{=}
    \pmqty{
    1 & 0 & 0 & 0\\
    0 & u_{11} & u_{12} & u_{13} \\
    0 & u_{21} & u_{22} & u_{23} \\
    0 & u_{31} & u_{32} & u_{33}
    }.
\end{equation}
The first row and column is fixed by the product-preserving property, while $u_{jk}$ might take on any value, so long as unitarity is satisfied.

(1) \emph{Define $\hat{\chn}_{\ketbra{\alpha_+},1}$ using the pseudoinverse}. This has the action on the basis operators
\begin{equation}\label{eq:pseudoinverse-pp-retrodiction-qubit}
\begin{aligned}
    \hat{\chn}_{\ketbra{\alpha_+},1}\bqty{\ketbra{\alpha'_+}} &= \ketbra{\alpha_+},\\
    \hat{\chn}_{\ketbra{\alpha_+},1}\bqty{\ketbra{\alpha'_-}} &= 0,\\
    \hat{\chn}_{\ketbra{\alpha_+},1}\bqty{\ketbra{\alpha'_+}{\alpha'_-}} &= 0.
\end{aligned}
\end{equation}

\begin{widetext}
(2) \emph{Define $\hat{\chn}_{\ketbra{\alpha_+},2}$ in a full-rank neighbourhood}. We perturb the prior as $\alpha_\varepsilon := (1-\varepsilon)\ketbra{\alpha_+} + \varepsilon\one/2 = (1-\varepsilon/2)\ketbra{\alpha_+} + (\varepsilon/2)\ketbra{\alpha_-}$. If $\chn[\alpha_\varepsilon]$ is invertible, Eq.~\eqref{ptz3} gives
\begin{equation}
    \mathcal{E}_{\alpha_\varepsilon}[\bullet] = \TrB\Bqty{
        U^\dag
            \sqrt{\omega} \bqty{\pqty{
                \frac{1}{\sqrt{\TrB\omega}} \bullet \frac{1}{\sqrt{\TrB\omega}}
            }\otimes\mathbbm{1}}
            \sqrt{\omega}
        U
    },
\end{equation}
where we have defined $\omega := U(\alpha_\varepsilon\otimes\ketbra{\beta_+})U^\dag$ for brevity. With this, we have
\begin{equation}
    \omega = \pqty{1-\frac{\varepsilon}{2}}\ketbra{\alpha_+',\beta_+'} + \frac{\varepsilon}{2} \underbrace{U \ketbra{\alpha_-,\beta_+} U^\dag}_{=:\ketbra{\omega_\perp}}.
\end{equation}
Here, we have further defined
\begin{equation}
    \ket{\omega_\perp} := u_{12}\ket{\alpha'_+,\beta'_-} + u_{22}\ket{\alpha'_-,\beta'_+} + u_{32}\ket{\alpha'_-,\beta'_-}.
\end{equation}
From this, we also have the reduced state
\begin{equation}
\begin{aligned}
    \TrB\omega &= \frac{1}{2}\mathbbm{1} + \frac{1}{2}\pqty{1-\varepsilon\pqty{
        1 - \abs{u_{12}}^2
    }}\pqty{\ketbra{\alpha_+'} - \ketbra{\alpha'_-}} \\
    &\qquad{}+{}\frac{1}{2} \varepsilon u_{12}u_{22}^*\ketbra{\alpha'_+}{\alpha'_-} + \frac{1}{2} \varepsilon u_{12}^*u_{22}\ketbra{\alpha'_-}{\alpha'_+} \\
    &=: \frac{1}{2}\mathbbm{1} + \frac{1}{2}\vec{r}\cdot\vec{\sigma},
\end{aligned}
\end{equation}
where $\vec{\sigma} = (\sigma_1,\sigma_2,\sigma_3)$ are the usual Pauli oeprators defined in the $\Bqty{\ket{\alpha'_\pm}}$ basis, and $\vec{r}$ is a three-dimensional vector given in spherical coordinates $(r,\theta,\phi)$ with
\begin{equation}\label{eq:partial-state-polar}
\begin{aligned}
    r\cos\theta &= 1-\varepsilon\pqty{
        1 - \abs{u_{12}}^2},\\
    r\sin\theta e^{i\phi} &= \varepsilon u_{12}^*u_{22}\\
    r &= \sqrt{\pqty{1-\varepsilon\pqty{
        1 - \abs{u_{12}}^2}}^2 + \varepsilon^2\abs{u_{12}}^2\abs{u_{22}}^2 }.
\end{aligned}
\end{equation}
The eigenvalues of $\TrB\omega$ can be verified to be $(1 + r)/2$ and $(1 - r)/2$, with eigenstates $\cos(\frac{\theta}{2})\ket{\alpha_+'} + e^{i\phi}\sin(\frac{\theta}{2})\ket{\alpha_-'}$ and $\sin(\frac{\theta}{2})\ket{\alpha_+'} - e^{i\phi}\cos(\frac{\theta}{2})\ket{\alpha_-'}$ respectively. Then, $\TrB\omega$ can be written in its diagonal form and $(\TrB\omega)^{-\frac{1}{2}}$ can be worked out. Since we are considering the cases where $\chn_{\alpha_\varepsilon}$ is well-defined, we require $\abs{u_{12}}^2 \neq 1$, so that the inverse of $\TrB\omega$ exists. With this,
\begin{equation}\label{eq:partial-state-inverse}
\begin{aligned}
    \frac{1}{\sqrt{\TrB\omega}} &=
    \frac{1}{\sqrt{2}r}\Bigg\{
        \bqty{
            \frac{r+r\cos\theta}{\sqrt{1+r}} - \sqrt{1-r} + \sqrt{\varepsilon} \times \sqrt{\frac{{\varepsilon}}{{1-r^2}}}\times \sqrt{1+r} \times (1-\abs{u_{12}}^2)
        }\ketbra{\alpha'_+} \\
        &\qquad {}+{}\bqty\bigg{
            \frac{r-r\cos\theta}{\sqrt{1+r}} + 
            \underbrace{\frac{1}{\sqrt{\varepsilon}}}_{\hspace{-3em}\text{\footnotesize \emph{possibly problematic}}\hspace{-3em}} \times \sqrt{\frac{\varepsilon}{1-r^2}} \times \sqrt{1+r} \times (r+r\cos\theta)
        }\ketbra{\alpha'_-} \\
        &\qquad {}+{}\bqty{
            \frac{r\sin\theta}{\sqrt{1+r}} - \sqrt{\frac{\varepsilon}{1-r^2}} \times \frac{r\sin\theta}{\sqrt{\varepsilon}} \times \sqrt{1+r}
        }\pqty\Big{\ketbra{\alpha'_+}{\alpha'_-}e^{-i\phi} + \ketbra{\alpha'_-}{\alpha'_+}e^{i\phi}}
    \Bigg\}.
\end{aligned}
\end{equation}
From Eq.~\eqref{eq:partial-state-polar}, we have the limits
\begin{equation}
\begin{gathered}
    \lim_{\varepsilon \to 0} r\cos\theta = 1,\qquad
    \lim_{\varepsilon \to 0} \frac{r\sin\theta}{\sqrt{\varepsilon}} = 0,\qquad
    \lim_{\varepsilon \to 0} r = 1,\\[2ex]
    \lim_{\varepsilon \to 0} \frac{\varepsilon}{1-r^2} =
    \lim_{\varepsilon \to 0} \frac{\varepsilon}{2\varepsilon\pqty{1-\abs{u_{12}}^2} - \varepsilon^2\pqty{ (1 - \abs{u_{12}}^2)^2 + \abs{u_{12}}^2\abs{u_{22}}^2 }} = \frac{1}{2\pqty{1-\abs{u_{12}}^2}}.
\end{gathered}
\end{equation}
Keeping these limits in mind, we have singled out the only \emph{possibly problematic} term in Eq.~\eqref{eq:partial-state-inverse}, which would diverge if the limit of $(\TrB\omega)^{-\frac{1}{2}}$ is considered on its own. Nonetheless, the following limits hold:
\begin{equation}\label{eq:basis-limits}
    \lim_{\varepsilon\to 0} \bra{\alpha'_+}\frac{1}{\sqrt{\TrB\omega}} = \bra{\alpha'_+},\qquad
    \lim_{\varepsilon\to 0} \sqrt{\frac{\varepsilon}{2}}\bra{\alpha'_-}\frac{1}{\sqrt{\TrB\omega}} = \frac{1}{\sqrt{\abs{u_{22}}^2 + \abs{u_{32}}^2}}\bra{\alpha'_-}.
\end{equation}
With this, we have
\begin{equation}
    \lim_{\varepsilon\to 0} \sqrt{\omega}\pqty{\frac{1}{\sqrt{\TrB\omega}}\otimes\one}
    = \ketbra{\alpha'_+,\beta'_+} + \frac{1}{\sqrt{1-\abs{u_{12}}^2}}\pqty{\ketbra{\omega_\perp} - \frac{1}{u_{12}^*}\ketbra{\omega_\perp}{\alpha'_+,\beta'_-}}.
\end{equation}
As such, the action of the limit of $\bqty{\pqty{
                \frac{1}{\sqrt{\TrB\omega}} \bullet \frac{1}{\sqrt{\TrB\omega}}
            }\otimes\mathbbm{1}}
            \sqrt{\omega}$ can be found to be
\begin{equation}
\begin{aligned}
    \lim_{\varepsilon\to 0}\sqrt{\omega} \bqty{\pqty{
                \frac{1}{\sqrt{\TrB\omega}} \bullet \frac{1}{\sqrt{\TrB\omega}}
            }\otimes\mathbbm{1}}
            \sqrt{\omega} &= 
    \ketbra{\alpha'_+,\beta'_+}{\alpha'_+}\bullet\ketbra{\alpha'_+}{\alpha'_+,\beta'_+} +
    \ketbra{\omega_\perp}{\alpha'_-}\bullet\ketbra{\alpha'_-}{\omega_\perp} \\
    &\qquad{}+{}
     \frac{u_{22}}{\sqrt{\abs{u_{22}}^2 + \abs{u_{32}}^2}} 
     \ketbra{\alpha'_+,\beta'_+}{\alpha'_+}\bullet\ketbra{\alpha'_-}{\omega_\perp} \\
    &\qquad{}+{}
     \frac{u_{22}^*}{\sqrt{\abs{u_{22}}^2 + \abs{u_{32}}^2}} 
     \ketbra{\omega_\perp}{\alpha'_-}\bullet\ketbra{\alpha'_+}{\alpha'_+,\beta'_+}.
\end{aligned}
\end{equation}
Finally, by taking the partial trace, we find
\begin{equation}\label{eq:perturb-pp-retrodiction-qubit}
\begin{aligned}
    \hat{\chn}_{\ketbra{\alpha_+},2}\bqty{\ketbra{\alpha'_+}} = \lim_{\varepsilon \to 0}\hat{\chn}_{\alpha_\varepsilon}\bqty{\ketbra{\alpha'_+}} &= \ketbra{\alpha_+},\\
    \hat{\chn}_{\ketbra{\alpha_+},2}\bqty{\ketbra{\alpha'_-}} = \lim_{\varepsilon \to 0}\hat{\chn}_{\alpha_\varepsilon}\bqty{\ketbra{\alpha'_-}} &= \ketbra{\alpha_-},\\
    \hat{\chn}_{\ketbra{\alpha_+},3}\bqty{\ketbra{\alpha'_+}{\alpha'_-}} = \lim_{\varepsilon \to 0}\hat{\chn}_{\alpha_\varepsilon}\bqty{\ketbra{\alpha'_+}{\alpha'_-}} &= \frac{u_{22}}{\sqrt{\abs{u_{22}}^2 + \abs{u_{32}}^2}} \ketbra{\alpha_+}{\alpha_-}.
\end{aligned}
\end{equation}

(3) \emph{Na\"ively define $\hat{\chn}_{\ketbra{\alpha_+},3}$ with product-preserving tuples}. In terms of $u_{jk}$, the action of $\hat{\chn}_{\ketbra{\alpha_+},3}$ as defined in Eq.~\eqref{eq:naive-pp-retrodiction} can be worked out directly to be
\begin{equation}\label{eq:naive-pp-retrodiction-qubit}
\begin{aligned}
    \hat{\chn}_{\ketbra{\alpha_+},3}\bqty{\ketbra{\alpha'_+}} &= \ketbra{\alpha_+},\\
    \hat{\chn}_{\ketbra{\alpha_+},3}\bqty{\ketbra{\alpha'_-}} &= \pqty{1-\abs{u_{21}}^2}\ketbra{\alpha_-} 
    + \abs{u_{21}}^2\ketbra{\alpha_+} 
    + u_{21}^*u_{23}\ketbra{\alpha_+}{\alpha_-} 
    + u_{21}u_{23}^*\ketbra{\alpha_+}{\alpha_-},\\
    \hat{\chn}_{\ketbra{\alpha_+},3}\bqty{\ketbra{\alpha'_+}{\alpha'_-}} &= u_{22}\ketbra{\alpha_+}{\alpha_-}.
\end{aligned}
\end{equation}
Once again, we find that the three approaches result in different retrodiction channels. Meanwhile, Eq.~\eqref{eq:naive-pp-retrodiction-qubit} agrees with Eq.~\eqref{eq:perturb-pp-retrodiction-qubit} if and only if $u_{21} = u_{12}= 0$, or $u_{21} = u_{22} = 0$. Therefore, $\hat{\chn}_{\ketbra{\alpha_+},2} \neq \hat{\chn}_{\ketbra{\alpha_+},3}$ in general, so we cannot na\"ively assume Eq.~\eqref{pp-tr} to hold for pure output states.
\end{widetext}

\subsection{Product-preserved tuples are almost tabletop time-reversible with the pseudoinverse convention}
Although we have not given a preference to any particular approach for retrodicting with rank-deficient outputs, it is worth pointing out that the product-preserved tuples are almost tabletop time-reversible with the pseudoinverse convention.

Given a product-preserving tuple $(U,\alpha,\beta)$ with $U(\alpha\otimes\beta)U^\dag = \alpha'\otimes\beta'$ where $\alpha'$ is not necessarily full-rank, Eq.~\eqref{qass} with the pseudoinverse reads
\begin{equation}
\begin{aligned}
    \qam{\alpha'\otimes\beta'}[\bullet] &= \sqrt{\alpha'} \sqrt{\alpha'}^{\,\boldsymbol{\pmb{+}}} \bullet 
    \sqrt{\alpha'}^{\,\boldsymbol{\pmb{+}}} \sqrt{\alpha'} \otimes \beta'\\
    &= \Pi_{\alpha'}\bullet\Pi_{\alpha'}\otimes\beta',
\end{aligned}
\end{equation}
which in turn gives a trace-non-increasing retrodiction channel
\begin{equation}\label{eq:almost-tr}
    \hat{\chn}_{\alpha}[\bullet] = \TrB\bqty{U^\dag\pqty{\Pi_{\alpha'}\bullet\Pi_{\alpha'}\otimes\beta'}U}.
\end{equation}
We say that it is \emph{almost} tabletop time-reversible as Eq.~\eqref{eq:almost-tr} is the same as Eq.~\eqref{eqfriendly} up to a projection onto the support of $\chn[\alpha]$ upon the input state.

\end{document}